%% file: AdaptInfMax.tex
\RequirePackage{fix-cm}
\documentclass[twocolumn]{svjour3}          
\smartqed  
\usepackage{marvosym}
\usepackage{graphicx}
\usepackage{balance,enumitem}  
\usepackage{xcolor}
\usepackage{multirow}
\usepackage{tabularx,xspace,url}
\usepackage{amsmath}\allowdisplaybreaks
\usepackage{amssymb}
\usepackage[numbers,sort]{natbib}
\usepackage[vlined,linesnumbered,ruled,norelsize]{algorithm2e}
\usepackage[title]{appendix}
\usepackage[caption=false,labelformat=simple]{subfig}

\graphicspath{{figs/}}

\input{0-setting.tex}

\title{Efficient Approximation Algorithms for Adaptive Influence Maximization
}

\author{Keke~Huang\textsuperscript{\#1} \and Jing~Tang\textsuperscript{\#2} \and Kai~Han\textsuperscript{3} \and Xiaokui~Xiao\textsuperscript{4} \and Wei~Chen\textsuperscript{5}  \and Aixin~Sun\textsuperscript{1} \and Xueyan~Tang\textsuperscript{1} \and Andrew~Lim\textsuperscript{2}\thanks{\textsuperscript{\#}\,Keke Huang and Jing Tang have contributed equally. \newline * Corresponding author: Jing Tang.}}

\authorrunning{K. Huang, J. Tang, K. Han, X. Xiao, W. Chen, A. Sun, X. Tang, and A. Lim} 

\institute{
Keke Huang \at
\email{khuang005@ntu.edu.sg}
\and
Jing Tang (\Letter) \at
\email{isejtang@nus.edu.sg}
\and
Kai Han \at
\email{hankai@ustc.edu.cn}
\and
Xiaokui Xiao \at
\email{xkxiao@nus.edu.sg}
\and
Wei Chen \at
\email{weic@microsoft.com}
\and
Aixin Sun \at
\email{axsun@ntu.edu.sg}
\and
Xueyan Tang \at
\email{asxytang@ntu.edu.sg}
\and
Andrew Lim \at
\email{isealim@nus.edu.sg}
\at
\begin{description}
	\item[\textsuperscript{1}] School of Computer Science and Engineering, Nanyang Technological University, Singapore
	\item[\textsuperscript{2}] Department of Industrial Systems Engineering and Management, National University of Singapore, Singapore
	\item[\textsuperscript{3}] School of Computer Science and Technology, University of Science and Technology of China, China
	\item[\textsuperscript{4}] School of Computing, National University of Singapore, Singapore
	\item[\textsuperscript{5}] Microsoft Research, China	
\end{description}
}

\date{}
\journalname{Accepted by The VLDB Journal.}

\begin{document}
\maketitle

\begin{abstract}
Given a social network $G$ and an integer $k$, the influence maximization (IM) problem asks for a seed set $S$ of $k$ nodes from $G$ to maximize the expected number of nodes influenced via a propagation model. The majority of the existing algorithms for the IM problem are developed only under the {\it non-adaptive} setting, \ie where all $k$ seed nodes are selected in one batch without observing how they influence other users in real world. In this paper, we study the {\it adaptive} IM problem where the $k$ seed nodes are selected in batches of equal size $b$, such that the $i$-th batch is identified after the actual influence results of the former $i-1$ batches are observed. In this paper, we propose the first practical algorithm for the adaptive IM problem that could provide the {\it worst-case} approximation guarantee of $1-\e^{\rho_b(\varepsilon-1)}$, where $\rho_b=1-(1-1/b)^b$ and $\varepsilon \in (0, 1)$ is a user-specified parameter. In particular, we propose a general framework \AG that could be instantiated by any existing non-adaptive IM algorithms with expected approximation guarantee. Our approach is based on a novel randomized policy that is applicable to the general adaptive stochastic maximization problem, which may be of independent interest. In addition, we propose a novel non-adaptive IM algorithm called \EP which not only provides strong expected approximation guarantee, but also presents superior performance compared with the existing IM algorithms. Meanwhile, we clarify some existing misunderstandings in recent work and shed light on further study of the adaptive IM problem. We conduct experiments on real social networks to evaluate our proposed algorithms comprehensively, and the experimental results strongly corroborate the superiorities and effectiveness of our approach.
\end{abstract}

\begin{sloppy}
\input{1-intro.tex}
\input{2-prelim.tex}
\input{3-frame.tex}
\input{4-batch.tex}

\input{5-revisit.tex}
\input{6-related.tex}

\input{7-exp.tex}
\input{8-conclusion.tex}

\begin{acknowledgements}
This research is supported by Singapore National Research Foundation under grant NRF-RSS2016-004, by Singapore Ministry of Education Academic Research Fund Tier 1 under grant MOE2017-T1-002-024, by Singapore Ministry of Education Academic Research Fund Tier 2 under grant MOE2015-T2-2-069, by National University of Singapore under an SUG, by National Natural Science Foundation of China under grant No.61772491 and No.61472460, and by Natural Science Foundation of Jiangsu Province under grant No.BK20161256.
\end{acknowledgements}

\balance
\bibliographystyle{spbasic}      
\bibliography{reference}   

\end{sloppy}

\end{document}

%% file: 0-setting.tex
\newcommand{\etal}{{et al.\@}\xspace}
\newcommand{\ie}{{i.e.,}\xspace}

\newcommand{\Ele}{\ensuremath{\mathcal{E}}}
\newcommand{\dom}{\operatorname{dom}}
\newcommand{\argmax}{\mathop{\arg\max}}

\newcommand{\E}{{\mathbb E}\xspace}
\newcommand{\R}{\mathcal{R}}

\newcommand{\e}{{\ensuremath{\mathrm{e}}}}
\newcommand{\ag}{\ensuremath{\mathrm{ag}}}

\newcommand{\avg}{\ensuremath{\mathrm{avg}}}
\newcommand{\OPT}{\operatorname{OPT}}
\newcommand{\Cov}{\operatorname{Cov}}

\newcommand{\lb}{F^l}
\newcommand{\ub}{F^u}
\newcommand{\maxMC}{\operatorname{maxMC}}

\newcommand{\EO}{\E_{\omega}}
\newcommand{\spara}[1]{\vspace{2mm}\noindent\textbf{#1.}}
\newcommand{\abs}[1]{\lvert#1\rvert}

\newcommand{\sharpP}{\ensuremath{\#\mathrm{P}}}

\newcommand{\PequalNP}{\ensuremath{\mathrm{P}\!=\!\mathrm{NP}}\xspace}

\newcommand{\adapone}{\textsc{AdaptIM-1}\xspace}

\newcommand{\adapim}{\textsc{AdaptIM}\xspace}
\newcommand{\adaimm}{\textsc{AdaIMM}\xspace}

\newcommand{\expepic}{\textsc{EptAIM}\xspace}
\newcommand{\expopim}{\textsc{EptAIM-N}\xspace}
\newcommand{\fixaim}{\textsc{FixAIM}\xspace}
\newcommand{\worepic}{\textsc{WstAIM}\xspace}

\newcommand{\imm}{\textsc{IMM}\xspace} 
\newcommand{\ssa}{\textsc{SSA}\xspace}
\newcommand{\dssa}{\textsc{D-SSA}\xspace}

\newcommand{\AG}{\textsc{AdaptGreedy}\xspace}
\newcommand{\EP}{\textsc{EPIC}\xspace}
\newcommand{\MC}{\textsc{MaxCover}\xspace}

\newcommand{\OPIMC}{\textsc{OPIM-C}\xspace}

%% file: 1-intro.tex
\section{Introduction} \label{sec:intro}

The proliferation of online social networks such as Facebook and Twitter has motivated considerable research on viral marketing as an optimization problem. For example, an advertiser could provide a few individuals (referred to as ``seed nodes'') in a social network with free product samples, in exchange for them to spread the good words about the product, so as to create a large cascade of influence on other social network users via word-of-mouth recommendations. This phenomenon has been firstly formulated as {\it Influence Maximization (IM)} problem in \cite{Kempe_maxInfluence_2003}, which aims to select a number of seed nodes to maximize the influence propagation created.

Formally, the input to IM consists of a social network $G = (V, E)$, a budget $k$, and an influence model $M$. The influence model $M$ captures the uncertainty of influence propagation in $G$, and it defines a set of \textit{realizations}, each of which represents a possible scenario of the influence propagation among the nodes in $G$. The problem seeks to {\it activate} (\ie~influence) a seed set $S$ of $k$ nodes that can maximize the expected number of influenced individuals over all realizations.

A plethora of techniques have been proposed for IM \cite{LeskovecKGFVG07, Goyal_SIMPATH_2011, OhsakaAYK14, Tang_TIM_2014, Tang_IMM_2015, GalhotraAR16, Nguyen_DSSA_2016, Borgs_RIS_2014, Kempe_maxInfluence_2003, Tang_infMax_2017, Arora_debunk_2017, Huang_SSA_2017, OhsakaSFK17, Tang_IMhop_2018, Tang_OPIM_2018}. Almost all techniques, however, require that the seed set $S$ be decided before the influence propagation process, which means that they work in a ``non-adaptive'' manner. In other words, if an advertiser has $k$ product samples, she would have to commit all samples to $k$ chosen social network users before observing how they may influence other users. In practice, however, an advertiser could employ a more {\it adaptive} strategy to disseminate the product samples. For example, she may choose to give out half of the samples, and then wait for a while to find out which users are influenced; after that, she could examine the set $U$ of users that have not been influenced, and then disseminate the remaining samples to $k/2$ users that have a large influence on $U$. This strategy is likely to be more effective than giving out all $k$ samples all at once, since the dissemination of the second batch of products is optimized using the knowledge obtained from the first batch's results.

In fact, the above adaptive approach has been applied in HEALER~\cite{Yadav_DIM_2016}, a software agent deployed in practice since 2016, which recommends sequential intervention plans for homeless shelters. HEALER aims to raise awareness about HIV among homeless youth by maximizing the spread of awareness in the social network of the target population. It chooses people as the seed nodes, who are ``activated'' by participating the intervention plans for HIV. The choices of seed nodes are adaptive, \ie~they are selected in batches and the choice of a batch depends on the observed results of all previous batches.

Golovin and Krause~\cite{Golovin_adaptive_2011} are the first to study IM under the adaptive setting, assuming that the $k$ seed nodes are chosen in a {\it sequential} manner, such that the selection of the $(i+1)$-th node is performed after the influence of the first $i$ nodes has been observed. Specifically, they consider that (i) the social network conforms to a realization $\phi$ that is generated by independently same every edge in graph $G$ (according to the {\em independent cascade model}), but (ii) $\phi$ is not known to the advertiser before the selection of the first seed node. Then, after the $i$-th seed node $v_i$ is chosen, the part of $\phi$ relevant to $\{v_1, v_2, \ldots, v_i\}$ (\ie~the nodes that they can influence in $\phi$) is revealed to the advertiser, based on which she can (i) eliminate the realizations that contradict what she observes, and (ii) select the next seed node as one that has a large expected influence over the remaining realizations.

Golovin and Krause~\cite{Golovin_adaptive_2011} propose a simple greedy algorithm for adaptive IM that returns a seed set $S$ whose influence is at least $1-1/\e$ of the optimum under the case that only one seed is selected in each batch (\ie~$b=1$). Nevertheless, the algorithm requires knowing the {\it exact} expected influence of every node, which is impractical since the computation of expected spread is \sharpP-hard in general \cite{Chen_MIA_2010,Chen_LDAG_2010}. Vaswani and Lakshmanan~\cite{Vaswani_adapIM_2016} extend Golovin \etal's model by allowing selecting $b\geq 1$ seed nodes in each batch, and by accommodating errors in the estimation of expected spreads. Their method returns an $(1-\e^{-(1-1/\e)^2/\eta})$-approximation under this setting, where $\eta$ is certain number bigger than $1$. However, this relaxed approach is still impractical in that its requirement on the accuracy of expected spread estimation cannot be met by any existing algorithms (see Section~\ref{sec:prelim-existing} for a discussion).

To mitigate the above defects, there are two recent papers for adaptive IM, \ie our preliminary work~\cite{Han_AIM_2018} and Sun \etal's paper~\cite{Sun_MRIM_2018}. Han \etal \cite{Han_AIM_2018} propose the first practical algorithm \adapim. Meanwhile, Sun \etal \cite{Sun_MRIM_2018} propose another approximation algorithm \adaimm for a variant of the adaptive IM problem, referred to as {\it Multi-Round Influence Maximization (MRIM)}. These two algorithms are claimed to provide the same worst-case approximation guarantee of $1-\e^{(1-1/\e)(\varepsilon-1)}$ with high probability, where $\varepsilon\in (0,1)$ is a user-specified parameter. Unfortunately, both of their theoretical analyses on the approximation guarantee contain some gaps that invalidate their claims. We shall elaborate these misclaims in Section~\ref{sec:revisited}.
	
\spara{Contribution} Motivated by the deficiency of existing techniques and misunderstandings, we conduct an intensive study on the adaptive IM problem, and propose the first practical solution. Meanwhile, we derive a rigorous theoretical analysis that clarifies existing confusing points and lays a solid foundation for further study. Specifically, our contributions include the following.

First, we propose a novel randomized policy that can provide strong theoretical guarantees for the general adaptive stochastic maximization problem, which may be of independent interest. This new solution can be adopted in many other settings apart from adaptive IM, e.g., active learning~\cite{Cuong_active_2013}, active inspection~\cite{Hollinger_active_2013}, optimal information gathering \cite{Chen_sequential_2015}, which are special cases of adaptive stochastic maximization. In particular, our policy imposes far fewer constraints than the existing solutions \cite{Golovin_adaptive_2011}, which are more applicable. The derivation of approximation results requires a non-trivial extension of the existing theoretical results on adaptive algorithms \cite{Golovin_adaptive_2011}, and some new techniques like Azuma-Hoeffding inequality~\cite{Mitzenmacher_Martingales_2005}. In addition, we propose a framework \AG for adaptive IM that enables us to construct strong approximation solutions using existing non-adaptive IM methods as building blocks. In particular, we prove that \AG achieves a {\it worst-case} approximation guarantee of $1-\e^{\rho_b(\varepsilon-1)}$ with high probability when the number of adaptive rounds is reasonably large, where $\varepsilon\in (0,1)$ is a user-specified parameter and $\rho_b=1-(1-1/b)^b$ is set by the batch size $b$. Moreover, we show that \AG can also provide an {\it expected} approximation guarantee of $1-\e^{\rho_b(\varepsilon-1)}$. Meanwhile, our analyses uncover some potential gaps in two recent works \cite{Han_AIM_2018,Sun_MRIM_2018} and shed light on the future work of the adaptive IM problem.

Second, we conduct an in-depth analysis on how \AG could be instantiated with the state-of-the-art non-adaptive IM algorithms. The overall approximation guarantee of \AG relies on the {\it expected} approximation guarantee of the non-adaptive IM algorithm used by \AG. However, existing non-adaptive IM algorithms do not benefit \AG in this regard, as there is no known result on their expected approximation guarantees. Motivated by this fact, we develop a new non-adaptive IM method, \EP, that provides an attractive expected approximation ratio by utilizing martingale stopping theorem \cite{Mitzenmacher_Martingales_2005}. We establish \AG's performance guarantee instantiated with \EP.

Third, we conduct extensive experiments to test the performance of \AG and \EP, and the experimental results strongly corroborate the effectiveness and efficiency of our approach.

%% file: 2-prelim.tex
\section{Preliminaries} \label{sec:prelim}

\begin{figure*}[!t]
	\centering
	\subfloat[A social network]{\includegraphics[width=0.23\linewidth]{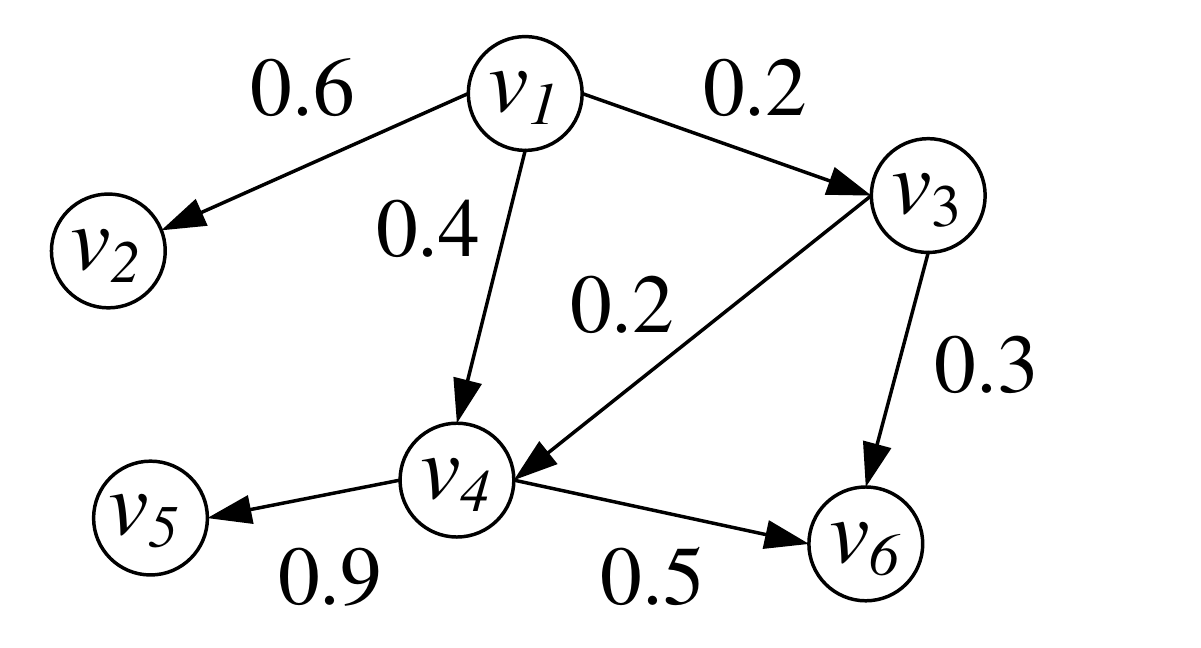}\label{fig:social-network}}\hfill
	\subfloat[Realization $\phi_1$]{\includegraphics[width=0.23\linewidth]{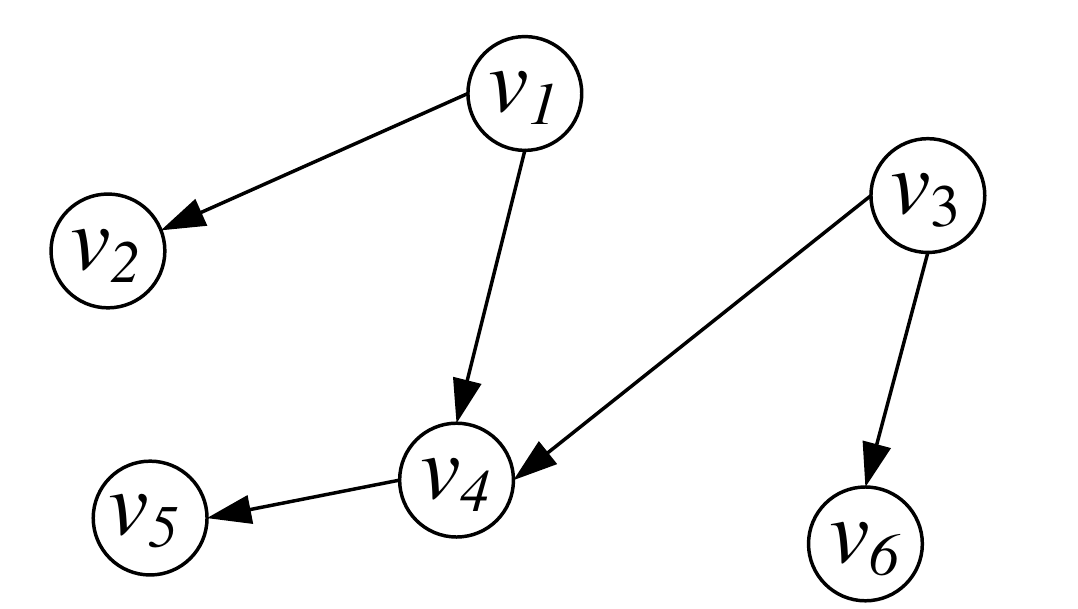}\label{fig:p1}}\hfill
	\subfloat[Realization $\phi_2$]{\includegraphics[width=0.23\linewidth]{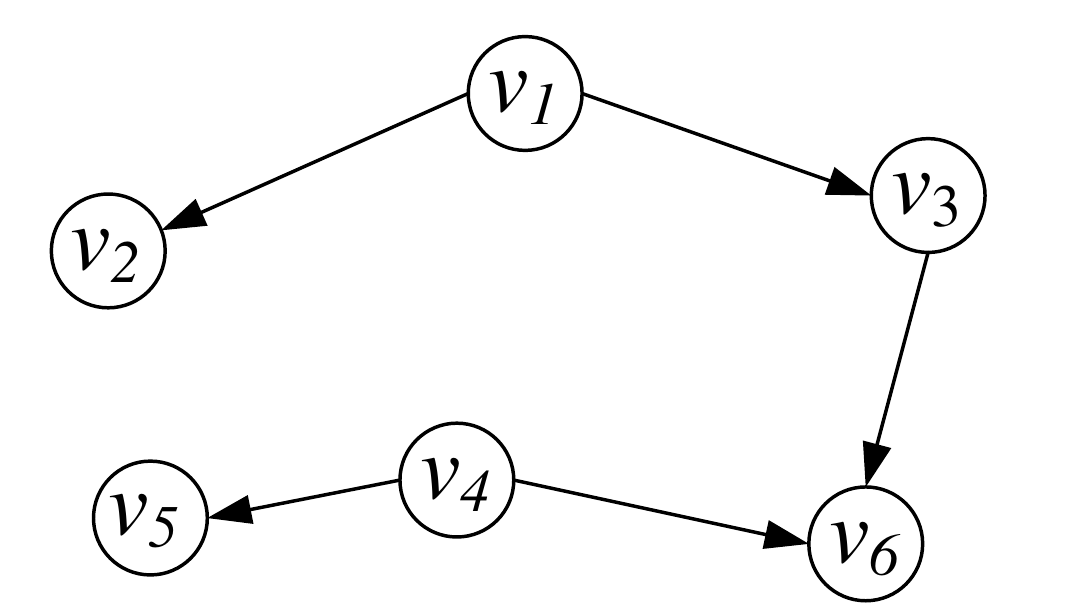}\label{fig:p2}}\hfill
	\subfloat[Realization $\phi_3$]{\includegraphics[width=0.23\linewidth]{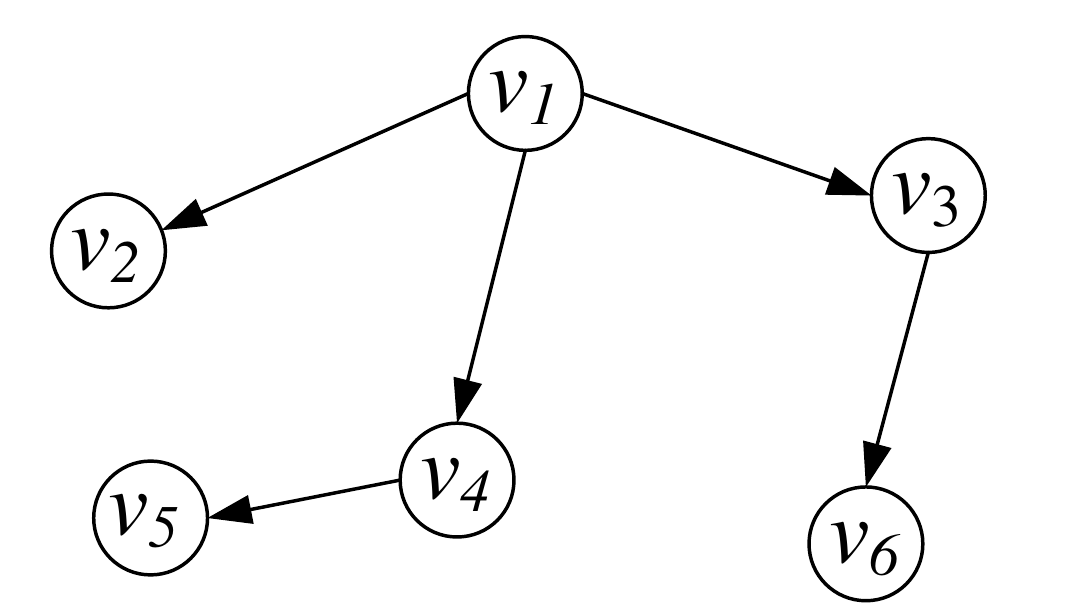}\label{fig:p3}}
	\caption{A social network and three of its realizations.}\label{fig:prelim-world}
\end{figure*}

\begin{figure*}[!t]
	\centering
	\vspace{-0.2in}
	\subfloat[$v_1$ as the first seed]{\includegraphics[width=0.23\linewidth]{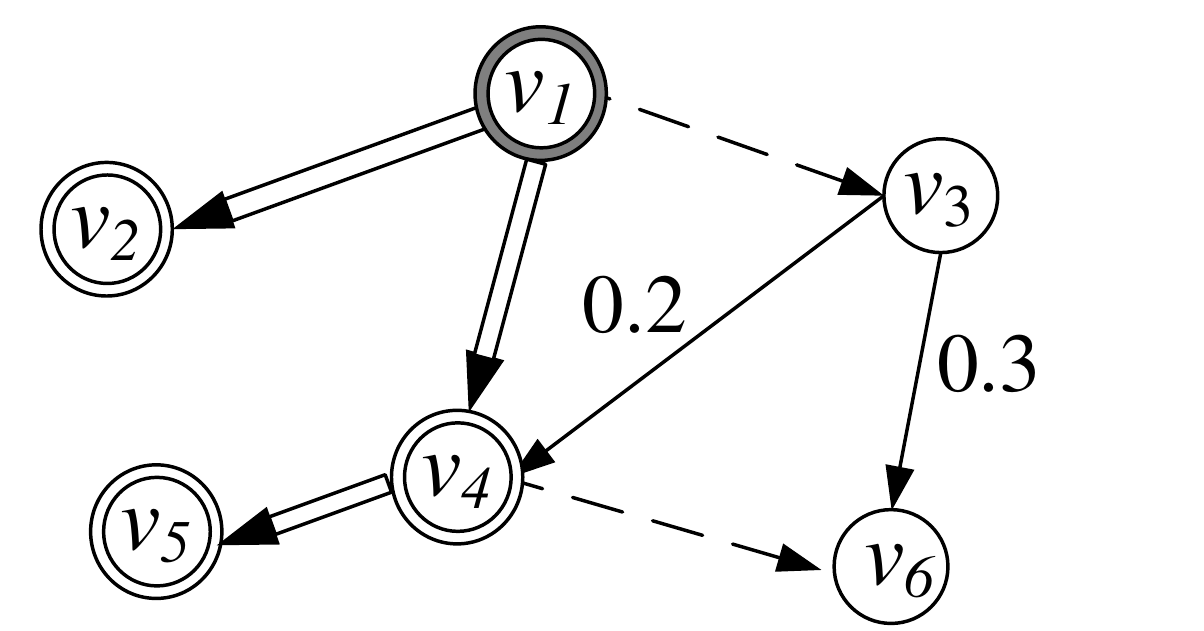}\label{fig:Adaptive1}}\hfill
	\subfloat[Second residual graph]{\includegraphics[width=0.2\linewidth]{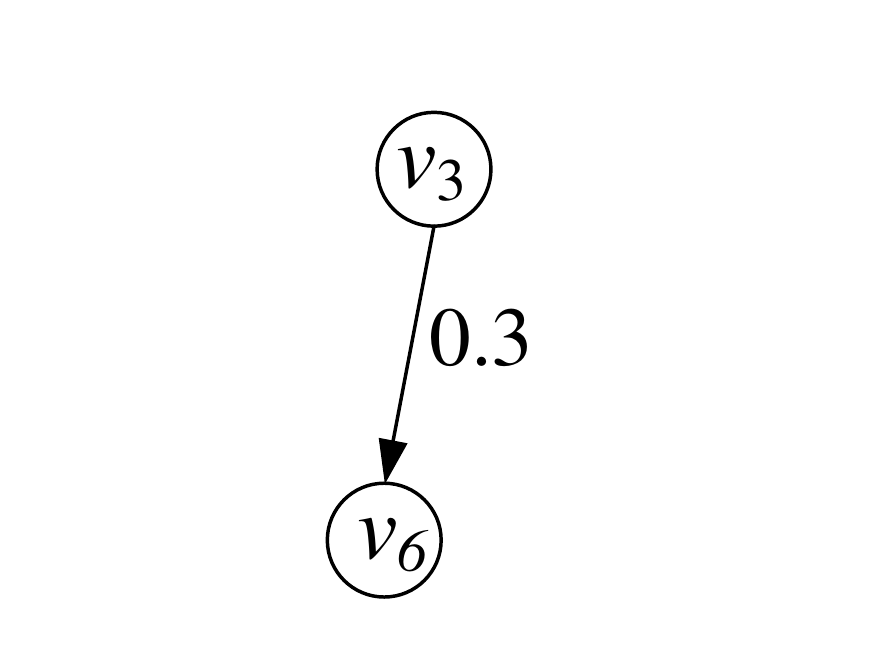}\label{fig:Adaptive2}}\hfill
	\subfloat[$v_3$ as the second seed]{\includegraphics[width=0.22\linewidth]{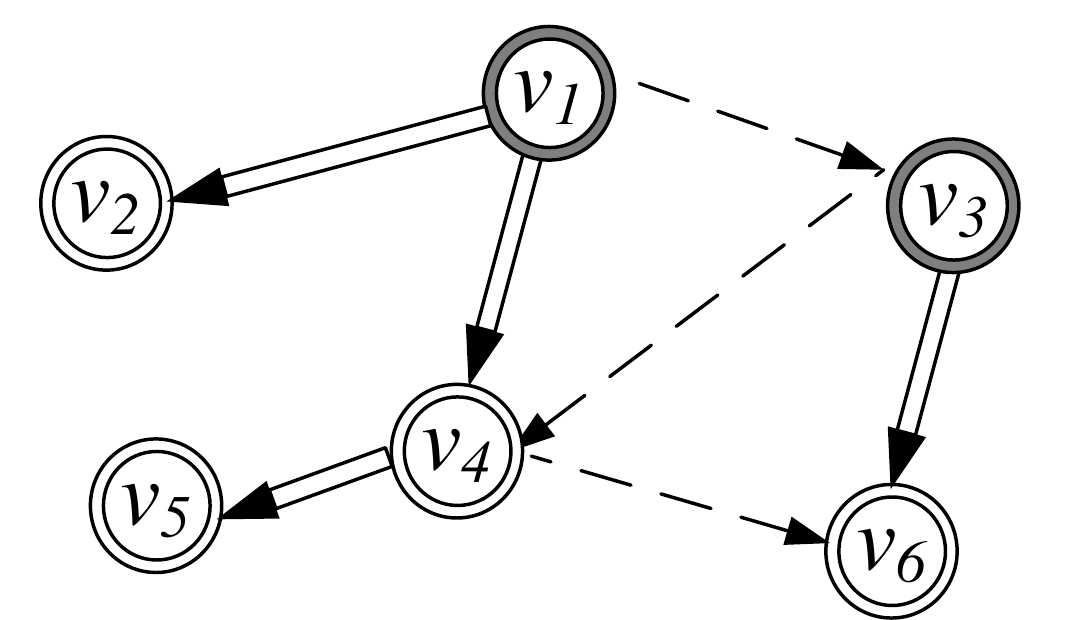}\label{fig:Adaptive3}}\hfill
	\subfloat[Non-adaptive IM]{\includegraphics[width=0.23\linewidth]{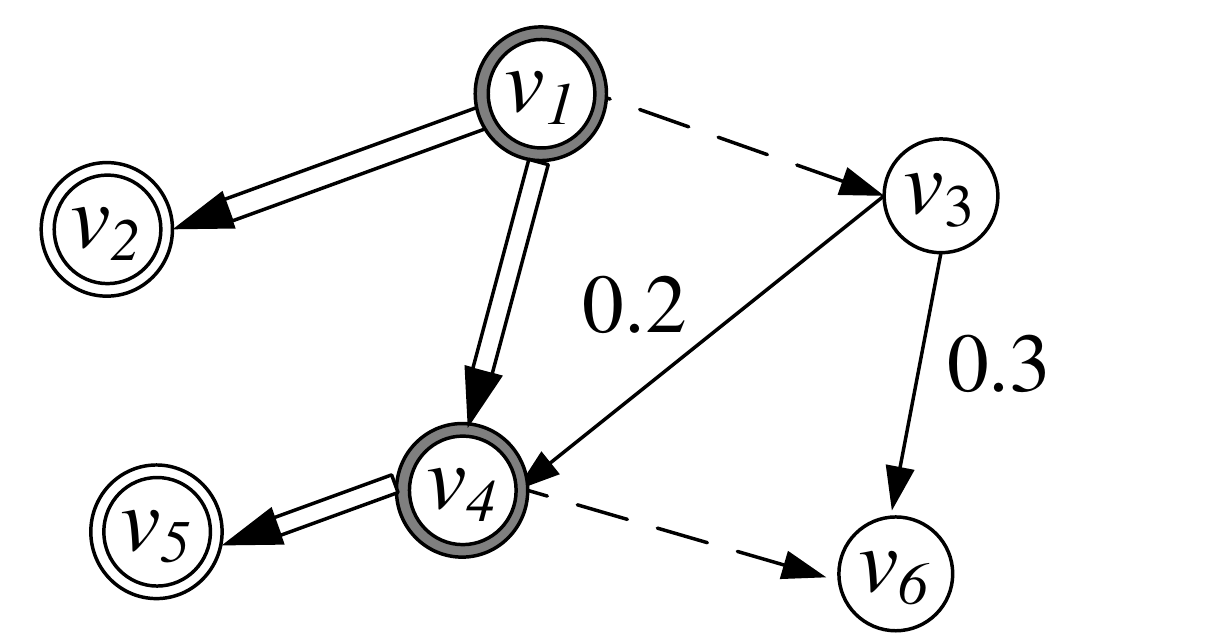}\label{fig:non-adaptive}}
	\caption{Adaptive vs. non-adaptive seed selection with $k=2$.}\label{fig:prelim-IM}
\end{figure*}

\subsection{IM and Realization} \label{sec:prelim-IM}

Let $G = (V, E)$ be a social network with a node set $V$ and an edge set $E$, such that $|V| = n$ and $|E| = m$. We assume that the propagation of influence on $G$ follows the {\it independent cascade (IC)} model \cite{Kempe_maxInfluence_2003}, in which each edge $(u, v)$ in $G$ is associated with a probability $p(u, v)$, and the influence propagation process is defined as a discrete-time stochastic process as follows. At timestamp $0$, we activate a set $S$ of {\it seed nodes}. Then, at each subsequent timestamp $t$, each node $u$ that is newly activated at timestamp $t-1$ has a chance to activate each of its neighbors $v$, such that the probability of activation equals $p(u, v)$. After that, $u$ stays active, but cannot activate any other nodes. The propagation process terminates when no node is newly activated at a certain timestamp, and the total number of nodes activated then is defined as the {\it influence spread} of $S$, denoted as $I_G(S)$. The {\it vanilla} influence maximization (IM) problem asks for a seed set $S$ of $k$ nodes that maximizes the expected value of influence spread $\E[I_G(S)]$.

As demonstrated in \cite{Kempe_maxInfluence_2003}, the IC model also has an interpretation based on {\it realization}. Specifically, a realization $\phi$ represents a \textit{live-edge} graph \cite{Kempe_maxInfluence_2003} generated by removing each edge $(u, v)$ in $G$ independently with $1 - p(u, v)$ probability. For example, \figurename~\ref{fig:prelim-world} shows a social network and three of its realizations. We use $\Phi$ to denote a random realization. For any seed set $S$, let $I_\phi(S)$ be the number of nodes in $\phi$ (including those in $S$) that can be reached from $S$ via a directed path starting from $S$, and $\E_{\Phi}[I_\Phi(S)]$ be the expectation over all realizations. It is shown in \cite{Kempe_maxInfluence_2003} that 
\begin{equation*}
	\E_{\Phi}[I_\Phi(S)] = \E[I_G(S)].
\end{equation*}
In other words, if we are to address the vanilla IM problem, it suffices to identify a seed set $S$ whose expected spread over all realizations is the largest.



\subsection{Adaptive IM} \label{sec:prelim-adaptive}

Suppose that the influence propagation on $G$ conforms to a \textit{realization} $\phi$, i.e., for any seed set $S$, the nodes that it can influence are exactly the nodes that it can reach in $\phi$. The {\em adaptive influence maximization (IM)} problem \cite{Golovin_adaptive_2011} considers that $\phi$ is unknown in advance, but can be partially revealed after we choose some nodes as seeds. For example, consider the social network in \figurename~\ref{fig:social-network}, and suppose that the realization is $\phi_1$, as shown in \figurename~\ref{fig:p1}. Assume that we choose $v_1$ as the first seed node. In that case, we can observe $v_1$'s influence on $v_2$ and $v_4$, since $v_1$ has two outgoing edges $(v_1, v_2)$ and $(v_1, v_4)$ in $\phi_1$. Similarly, we can observe $v_4$'s influence on $v_5$. In addition, we can also observe that $v_1$ (resp.\ $v_4$) cannot influence $v_3$ (resp.\ $v_6$), as $\phi_1$ does not contain an edge from $v_1$ to $v_3$ (resp.\ $v_4$ to $v_6$). \figurename~\ref{fig:Adaptive1} shows the results of the influence propagation from $v_1$, with each double-line (dashed-line) arrow denoting a successful (resp.\ failed) step of influence.

In general, after choosing a partial set $S'$ of seed nodes, we can learn all nodes that $S'$ can reach in $\phi$, as well as the out-edges of those nodes in $\phi$. 
This is referred to as the full-adoption feedback model in~\cite{Golovin_adaptive_2011}.
This enables us to optimize the choices of the remaining seed nodes since we can focus on the nodes that have not been influenced by $S'$. For instance, consider that selecting another seed node based on the result in \figurename~\ref{fig:Adaptive1}. In that case, we can omit the nodes that have been influenced (i.e., $v_1$, $v_2$, $v_4$, and $v_5$), and focus on the subgraph induced by the remaining nodes, as shown in \figurename~\ref{fig:Adaptive2}. Based on this, we can choose $v_3$ as the second seed node, which yields the result in \figurename~\ref{fig:Adaptive3}, where we have $6$ nodes influenced in total. In contrast, if we are to non-adaptively choose two seed nodes from the social network in \figurename~\ref{fig:social-network}, we may end up choosing $v_1$ and $v_4$, in which case we would obtain the result in \figurename~\ref{fig:non-adaptive} when the realization is $\phi_1$ in \figurename~\ref{fig:p1}. In other words, we can only influence $4$ nodes instead of $6$ nodes.

Assume that we are to choose $k$ seed nodes in $r$ batches of equal size $b = k/r$, and that we are allowed to observe the influence propagation in $\phi$ for $r$ times in total, once after the selection of each batch. The adaptive IM problem asks for a seed selection policy that could generate the next seed set given the feedback of previous seed sets to maximize the expected influence spread over all realizations. Observe that when $b = k$ (i.e., $r = 1$), the problem degenerates to the vanilla IM problem.

We aim to develop algorithms for adaptive IM that provide non-trivial guarantees in terms of both accuracy (i.e., the expected influence of $\bigcup_i S_i$) and efficiency (i.e., the time required to identify $S_i$). We do not consider the ``waiting time'' required to observe the influence of a seed node batch $S_i$ before the selection of the next batch $S_{i+1}$, since it is independent of the algorithms used. That is, we target at helping the advertiser to identify $S_{i+1}$ as quickly as possible after the effects of $S_i$ have been observed.

\begin{table}[!t]
	\centering
    \caption{Frequently used notations}
    \label{tbl:prelim-notations}
    \vspace{2mm}
    \setlength{\tabcolsep}{0.5em} 
    \renewcommand{\arraystretch}{1.2}
	\begin{tabular}{|c|m{6.1cm}|}\hline
        \textbf{Notation} & \multicolumn{1}{c|}{\textbf{Description}} \\ \hline
        $G=(V,E)$ & a social network with node set $V$ and edge set $E$\\ \hline
        $n,m$ & the numbers of nodes and edges in $G$, respectively\\ \hline
        $k$ & the total number of selected seed nodes \\ \hline
        $b$ & the number of nodes selected in each batch \\ \hline
        $G_i$ & the $i$-th residual graph \\ \hline
        $n_i,m_i$ & the numbers of nodes and edges in $G_i$, respectively \\ \hline
        $S_i$ & the seed set selected from $G_i$  \\ \hline
        $S_i^o$ & the optimal seed set in $G_i$  \\ \hline
        $\rho_b$ & approximation guarantee for \MC with $\rho_b = 1-(1-1/b)^b$. \\ \hline
        $\OPT_{k,b}$ & the optimal expected influence spread of $k$ seed nodes under the setting of selecting $b$ nodes in each batch\\ \hline
        $\OPT_b(G_i)$ & the optimal expected influence spread of $b$ seed nodes in $G_i$\\ \hline
        $I_G(S)$ & the number of nodes activated by $S$ in $G$\\ \hline
        $\Cov_{\R}(S)$ & the number of RR-sets in $\R$ that overlap $S$ \\ \hline
        $F_{\R}(S)$ & the fraction of RR-sets in $\R$ that overlap $S$ \\ \hline
        $\E[I(S)]$ & the expected spread of seed set $S$ \\ \hline 
    \end{tabular}
    \vspace{-0.1in}
\end{table}

Table~\ref{tbl:prelim-notations} lists the notations that are frequently used in the remainder of the paper.

\subsection{Existing Solutions} \label{sec:prelim-existing}

The first solution to adaptive IM is by \cite{Golovin_adaptive_2011}. It assumes that $b = 1$ (i.e., each batch consists of only one seed node), and adopts a greedy approach as follows. Given $G$, it first identifies the node $v_1$ whose expected spread $\E[I_G(\{v_1\})]$ on $G$ is the largest, and selects it as the first seed. Then, it observes the nodes that are influenced by $v_1$ (which are in accordance to the realization $\phi_0$), and removes them from $G$. Let $G_2$ denote the subgraph of $G$ induced by the remaining nodes. After that, for the $i$-th ($i > 1$) batch, it (i) selects the node $v_i$ with the maximum expected spread $\E[I_{G_i}(\{v_i\})]$ on $G_i$, (ii) observes the influence of $v_i$ on $G_i$, and then (iii) generates a new graph $G_{i+1}$ by removing from $G_i$ those nodes that are influenced by $v_i$. For convenience, we refer to $G_i$ as the $i$-th {\it residual graph}, and let $G_1 = G$.

Let $\OPT_{k, b}$ denote the expected spread of the optimal solution to the adaptive IM problem parameterized with $k$ and $b$. Golovin et al.~\cite{Golovin_adaptive_2011} show that the above greedy approach returns a solution whose expected spread is at least $(1-1/\e) \cdot \OPT_{k, 1}$. This approximation guarantee, however, cannot be achieved in polynomial time because (i) in the $i$-th batch, it requires identifying a node $v_i$ with the maximum largest expected spread $\E[I_{G_i}(\{v_i\})]$ on $G_i$, but (ii) computing the exact expected spread of a node in the IC model is \sharpP-hard in general \cite{Chen_MIA_2010}.

To remedy the above deficiency, Vaswani~and Lakshmanan \cite{Vaswani_adapIM_2016} propose a relaxed approach that allows errors in the estimation of expected spreads. In particular, they assume that for any node set $S$ and any residual graph $G_i$, we can derive an estimation $\tilde{\E}[I_{G_i}(S)]$ of $\E[I_{G_i}(S)]$, such that
\begin{equation} \label{eqn:prelim-alpha}
\alpha^\bot \cdot \E[I_{G_i}(S)] \le \tilde{\E}[I_{G_i}(S)] \le \alpha^\top \cdot \E[I_{G_i}(S)],
\end{equation}
with $\alpha^\top/\alpha^\bot$ bounded from above by a parameter $\eta$. They show that, by feeding such estimated expected spreads to the greedy approach in \cite{Golovin_adaptive_2011}, it can achieve an approximation guarantee of $1 - \e^{-1/\eta}$. In addition, they show that the greedy approach can be extended to the case when $b > 1$, with one simple change: in the $i$-th batch, instead of selecting only one node, we select a size-$b$ seed set $S_i$ whose estimated expected spread on $G_i$ is at least $1 - 1/\e$ fraction of the largest estimated expected spread on $G_i$. In that case, they show that the resulting approximation guarantee is ${1-\e^{-{(1-1/\e)^2}/{\eta}}}$.

Unfortunately, the accuracy requirement in Equation~\eqref{eqn:prelim-alpha} is still impractical as no existing algorithm for evaluating expected spread can meet the requirement. Indeed, as computing $\E[I_{G_i}(S)]$ is \sharpP-hard, the existing algorithms can only derive $\tilde{\E}[I_{G_i}(S)]$ in a probabilistic manner, which implies that both $\alpha^\bot$ and $\alpha^\top$ are random numbers depending on $G_i$. As $G_i$ is also random, it is hard to derive a meaningful fixed upper bound $\eta$ for $\alpha^\top/\alpha^\bot$. Therefore, we think that the approximation ratio proposed in \cite{Vaswani_adapIM_2016} only has theoretical value and cannot be implemented in practice.

Motivated by those defects of previous work, recently, \adapim \cite{Han_AIM_2018} and \adaimm \cite{Sun_MRIM_2018} algorithms are proposed for the adaptive IM problem. These two algorithms are claimed to provide an approximation guarantee of $1-\e^{(1-1/\e)(\varepsilon-1)}$ with $1-\delta$ probability where $\varepsilon, \delta \in (0,1)$. Unfortunately, both of the theoretical analyses contain some gaps which make their claims invalid. The detailed analyses are presented in Section~\ref{sec:revisited}.

%% file: 3-frame.tex
\section{Our Solution} \label{sec:solu}
Fundamentally, adaptive IM is based on adaptive submodular optimization \cite{Golovin_adaptive_2011}. In this section, we first present a \textit{randomized} adaptive greedy policy to address the general optimization problem, and analyze the corresponding theoretical guarantees. Our solutions generalize the results of Golovin and Krause \cite{Golovin_adaptive_2011}, and thus it may be of independent interest. Finally, we propose a general framework \AG upon which we can build specific algorithms with seed selection algorithms to address the adaptive IM problem.

\subsection{Notations and Definitions}
Let $\Ele$ be a finite set of \textit{items} (e.g., a set of node sets), and $O$ be a set of possible \textit{states} (e.g., the activation statuses of nodes). A \textit{realization} is a function $\phi\colon \Ele \mapsto O$ mapping every item $e$ to a state $o$. We use $\Phi$ to denote a random realization. Let $p(\phi):=\Pr[\Phi=\phi]$ be the probability distribution over all realizations. We sequentially select an item $e$, and then observe its state $\Phi(e)$. Based on the observation, we would choose the next item and get to see its state, and so on. We use $\psi$, referred to as \textit{partial realization}, to represent the relation such that $\psi:= \{(e,o)\colon \psi(e)=o\}$ for any $\psi \subseteq \Ele \times O$. Let $\dom(\psi)$ denote the domain of $\psi$ such that $\dom(\psi):= \{e\colon \exists o, (e,o)\in \psi\}$. A partial realization $\psi$ is \textit{consistent with} a realization $\phi$, referred to as $\phi\sim \psi$, if for every $e\in \dom(\psi)$, $\psi(e)=\phi(e)$. Furthermore, we say $\psi\subseteq \psi^\prime$, \ie~$\psi$ is a \textit{subrealization} of $\psi^\prime$, if there exists some $\phi$ such that $\phi\sim\psi$ and $\phi\sim\psi^\prime$, and $\dom(\psi)\subseteq \dom(\psi^\prime)$.

A \textit{policy} $\pi$ is an adaptive strategy for selecting items in $\Ele$ based on current partial realization $\psi$. In this paper, we consider a \textit{randomized} policy that selects items following certain distribution. To explicitly reveal the randomness of a randomized policy, we denote $\pi(\omega)$ as a random policy chosen from a set of all possible deterministic policies with respect to a random variable $\omega$. Intuitively, $\omega$ represents all random source of the randomized policy. In addition, let $\pi(\omega,\psi)$ be the item picked by policy $\pi(\omega)$ under partial realization $\psi$. We denote $\Ele(\pi(\omega),\phi)$ as the set of items selected by $\pi(\omega)$ under realization $\phi$. We consider a utility function $f\colon 2^{\Ele}\times O^{\Ele}\mapsto \mathbb{R}_{\geq 0}$ depending on the picked items and their states. Then, the expected utility of a policy $\pi(\omega)$ is $f_{\avg}(\pi(\omega)):= \E_{\Phi}[f(\Ele(\pi(\omega),\Phi),\Phi)]$. The goal of the \textit{adaptive stochastic maximization} problem is to find a randomized policy $\pi^\ast$ such that
\begin{align*}
	&\pi^\ast \in \argmax_{\pi} \E_\omega[f_{\avg}(\pi(\omega))] \\
	&\text{s.t. } \abs{\Ele(\pi(\omega),\phi)}\leq r \text{ for all $\omega$ and all $\phi$},
\end{align*}


In addition, for any partial realization $\psi$, let $\Delta(e\mid \psi)$ and $\Delta(\pi(\omega)\mid \psi)$ denote the \textit{conditional marginal benefit} of an item $e$ and a policy $\pi(\omega)$ conditioned on observing partial realization $\psi$, defined as
\begin{align}
\Delta(e\mid \psi)
&:=\E_{\Phi}\big[f(\dom(\psi)\cup\{e\},\Phi)\mid \Phi\sim \psi\big]\nonumber\\
&\mathrel{\phantom{:=}}\mathop{-}\E_{\Phi}\big[f(\dom(\psi),\Phi)\mid \Phi\sim \psi\big],\label{eqn:Delta}\\
\Delta(\pi(\omega)\mid \psi)
&:=\E_{\Phi}\big[f(\dom(\psi)\cup\Ele(\pi(\omega),\Phi),\Phi)\mid \Phi\sim \psi\big]\nonumber\\
&\mathrel{\phantom{:=}}\mathop{-}\E_{\Phi}\big[f(\dom(\psi),\Phi)\mid \Phi\sim \psi\big].\label{eqn:Delta_pi}
\end{align}
We are now ready to introduce the notations of monotonicity and submodularity to the adaptive setting:
\begin{definition}[Adaptive Monotonicity]\label{def:monotonicity}
	A function $f$ is adaptive monotone with respect to the realization distribution $p(\phi)$ if for all $\psi$ with $\Pr[\Phi\sim\psi]>0$ and all $e\in \Ele$, we have 
	\begin{equation*}
		\Delta(e\mid \psi)\geq 0.
	\end{equation*}
\end{definition}

\begin{definition}[Adaptive Submodularity]\label{def:submodularity}
	A function $f$ is adaptive submodular with respect to the realization distribution $p(\phi)$ if for all $\psi\subseteq \psi^\prime$ and $e\in \Ele\setminus \dom(\psi^\prime)$, we have 
	\begin{equation*}
	\Delta(e\mid \psi)\geq \Delta(e\mid \psi^\prime).
	\end{equation*}
\end{definition}
\spara{Remark} Note that a deterministic policy is a special randomized policy. Meanwhile, the solution of any randomized policy is a convex combination of solutions of deterministic policies. Thus, the optimal solution of any randomized policy can always be achieved by some deterministic policy. As a consequence, for any randomized policy $\pi$ and any deterministic policy $\pi^\prime$, it holds that $\max_{\pi} \EO[f_{\avg}(\pi(\omega))] = \max_{\pi^\prime} f_{\avg}(\pi^\prime)$.

\subsection{Adaptive Greedy Policy}\label{sec:AGP}
A policy $\pi(\omega)$ is called an \textit{$\alpha$-approximate greedy policy} if for all $\psi$, it always picks an item such that
\begin{equation*}
	\Delta(\pi(\omega,\psi)\mid \psi)\geq\alpha\max_{e}\Delta(e\mid \psi).
\end{equation*}
Golovin and Krause \cite{Golovin_adaptive_2011} show that when the utility function $f$ is adaptive monotone and adaptive submodular, an $\alpha$-approximate greedy policy $\pi$ can achieve an approximation ratio of $(1-\e^{-\alpha})$ for the adaptive stochastic maximization problem, \ie~$\EO[f_{\avg}(\pi(\omega))]\geq (1-\e^{-\alpha})\EO[f_{\avg}(\pi^\ast(\omega))]$ for all policies $\pi^\ast$. However, in some applications, we would construct a randomized policy that may perform arbitrary worse (with low probability). For example, if the true value of $\Delta(e\mid \psi)$ is difficult to obtain, a policy maximizes an estimate of $\Delta(e\mid \psi)$ using sampling method may perform arbitrary worse in terms of maximizing $\Delta(e\mid \psi)$ (e.g.,~with some probability, even though very small, all the state-of-the-art IM algorithms may perform arbitrary worse). Such a randomized policy is not an $\alpha$-approximate greedy policy, for which Golovin and Krause's theoretical results \cite{Golovin_adaptive_2011} are not applicable.

Inspired by Golovin and Krause's work \cite{Golovin_adaptive_2011}, we call a randomized policy $\pi^{\ag}$ an \textit{expected} $\alpha$-approximate greedy policy if it selects an item with $\alpha$-approximation to the best greedy selection in expectation, \ie~ 
\begin{equation*}
\EO[\Delta(\pi^{\ag}(\omega,\psi)\mid \psi)]\geq\alpha\max_{e}\Delta(e\mid \psi),
\end{equation*}
where the expectation is taken over the internal randomness of policy. For convenience, let $\xi(\pi^{\ag}(\omega),\psi)$ denote the random approximation ratio obtained by the policy $\pi^{\ag}(\omega)$ on $\psi$, \ie
\begin{equation*}
	\xi(\pi^{\ag}(\omega),\psi):=\frac{\Delta(\pi^{\ag}(\omega,\psi)\mid \psi)}{\max_{e}\Delta(e\mid \psi)}.
\end{equation*}
Then, an expected $\alpha$-approximate greedy policy $\pi^{\ag}$ can be described as
\begin{equation*}
\EO[\xi(\pi^{\ag}(\omega),\psi)]\geq\alpha.
\end{equation*}
In the following, we show that such an expected $\alpha$-approximate greedy policy have strong theoretical guarantees. 

\subsection{Approximation Guarantees}
We consider a general version of randomized policy $\pi^{\ag}$ that can return an expected $\alpha_i$-approximate solution for the $i$-th item selection under every partial realization $\psi_{i-1}$, where $\psi_{i-1}$ represents a partial realization after we pick the first $(i-1)$ items, \ie~for every $\psi_{i-1}$, $\EO[\xi(\pi^{\ag}(\omega),\psi_{i-1})]\geq\alpha_i$. For a conventional version of $\pi^{\ag}$, one may set $\alpha_i=\alpha$ for every $i$, while for the general version of $\pi^{\ag}$, $\alpha_i$'s can be distinct.

Let $\Psi_{i}(\pi(\omega),\phi)$ represent a random partial realization after the policy $\pi(\omega)$ picks the first $i$ items under the realization $\phi$. For simplicity, we omit $\pi^{\ag}$ in $\Psi_{i}(\pi^{\ag}(\omega),\phi)$ when policy $\pi^{\ag}$ is used, \ie~$\Psi_{i}(\omega,\phi):=\Psi_{i}(\pi^{\ag}(\omega),\phi)$. Then, given any realization $\phi$ and any partial realization $\psi_{i-1}$ such that $\phi\sim \psi_{i-1}$, policy $\pi^{\ag}$ satisfies
\begin{equation}\label{eqn:greedy_policy}
	\EO\big[\xi(\pi^{\ag}(\omega),\psi_{i-1})\mid \Psi_{i-1}(\omega,\phi)=\psi_{i-1}\big]\geq \alpha_i,
\end{equation}
which describes that $\pi^{\ag}$ always returns an expected $\alpha_i$-approximate solution for the $i$-th item selection under every partial realization $\psi_{i-1}$ no matter what items are chosen by $\pi^{\ag}$ in the first $(i-1)$ rounds. 

To facilitate the analysis that follows, we define the notions of ``policy truncation'' and ``policy concatenation'', which are conceptual operations performed by a policy.

\begin{definition}[Policy Truncation]\label{def:policytruncation}
	For any adaptive policy $\pi$, the policy truncation $\pi_{i}$ denotes an adaptive policy that performs exactly the same as $\pi$, except that $\pi_{i}$ only selects the first $i$ items for any $i\leq r$.
\end{definition}

\begin{definition}[Policy Concatenation]\label{def:policyconcat}
	For any two adaptive policy $\pi$ and $\pi'$, the policy concatenation $\pi\oplus\pi'$ denotes an adaptive policy that first executes the policy $\pi$, and then executes $\pi'$ from a fresh start as if any knowledge on the feedback obtained while running $\pi$ is ignored.
\end{definition}

\subsubsection{Expected Approximation Guarantee}\label{sec:expected-approximation}
The following theorem shows a concept of \textit{expected} approximation guarantee for policy $\pi^{\ag}$. 

\begin{theorem} \label{thm:greedy-ept-approx}
	If $f$ is adaptive monotone and adaptive submodular, and $\pi^{\ag}$ returns an expected $\alpha_i$-approximate solution for the $i$-th item selection under every partial realization $\psi_{i-1}$, then the policy achieves an expected approximation guarantee of $1-\e^{-\alpha}$, where $\alpha=\frac{1}{r}\sum_{i=1}^r \alpha_i$ and $r$ is the total number of items selected, \ie~for all policies $\pi^\ast$, we have
	\begin{equation}
		\EO[f_{\avg}(\pi^{\ag}(\omega))]\geq (1-\e^{-\alpha})\EO[f_{\avg}(\pi^\ast(\omega))].
	\end{equation} 
\end{theorem}


Note that if a policy is an $\alpha$-approximate greedy policy, it must also be an expected $\alpha$-approximate greedy policy. Thus, our results generalize those given by Golovin and Krause \cite{Golovin_adaptive_2011}. The proof of Theorem~\ref{thm:greedy-ept-approx} requires extensions of the theoretical results developed for adaptive stochastic maximization \cite{Golovin_adaptive_2011}. In the following, we first introduce some lemmas that are useful for proving Theorem~\ref{thm:greedy-ept-approx}.
\begin{lemma}\label{lemma:marginal_max}
	For any deterministic adaptive policy $\pi$ and any $i\geq 0$, we have
	\begin{equation*}
	f_{\avg}(\pi_{i+1})-f_{\avg}(\pi_{i})\leq \E_{\Phi}\big[\max_{e}\Delta(e\mid \Psi_{i}(\pi,\Phi))\big].
	\end{equation*}
\end{lemma}
\begin{proof}[Lemma~\ref{lemma:marginal_max}]
	Let $p(\psi_{i}^{\pi}):=\Pr[\Psi_{i}(\pi,\Phi)=\psi_{i}^{\pi}]$ be the probability of partial realization $\psi_{i}^{\pi}$ being observed after $\pi$ picks $i$ items over all realizations. We use $\Psi_{i}^{\pi}$ to denote such a random partial realization with respect to the probability distribution $p(\psi_{i}^{\pi})$.
	Then,
	\begin{align*}
	&f_{\avg}(\pi_{i+1})-f_{\avg}(\pi_{i})\\
	&=\E_{\Phi}[f(\Ele(\pi_{i+1},\Phi),\Phi)-f(\Ele(\pi_{i},\Phi),\Phi)]\\
	&=\E_{\Psi_{i}^{\pi}}\big[\E_{\Phi}[f(\Ele(\pi_{i+1},\Phi),\Phi)-f(\Ele(\pi_{i},\Phi),\Phi)\mid \Phi\sim \Psi_{i}^{\pi}]\big]\\
	&=\E_{\Psi_{i}^{\pi}}\big[\Delta(\pi(\Psi_{i}^{\pi})\mid \Psi_{i}^{\pi})\big]\\
	&\leq \E_{\Psi_{i}^{\pi}}\big[\max_{e}\Delta(e\mid \Psi_{i}^{\pi})\big]\\
	&=\E_{\Phi}\big[\max_{e}\Delta(e\mid \Psi_{i}(\pi,\Phi))\big],
	\end{align*}
	where the inequality is because $\Delta(\pi(\Psi_{i}^{\pi})\mid \Psi_{i}^{\pi})\leq \max_{e}\Delta(e\mid \Psi_{i}^{\pi})$ for every $\Psi_{i}^\pi$.
	\qed
\end{proof}

\begin{lemma}\label{lemma:step_submodular}
	Given any deterministic adaptive policy $\pi$ and any $i\leq j$, we have
	\begin{equation*}
	\E_{\Phi}\big[\max_{e}\Delta(e\mid \Psi_{i}(\pi,\Phi))\big]\geq \E_{\Phi}\big[\max_{e}\Delta(e\mid \Psi_{j}(\pi,\Phi))\big].
	\end{equation*}
\end{lemma}
\begin{proof}[Lemma~\ref{lemma:step_submodular}]
	Again, let $\Psi_{i}^{\pi}$ and $\Psi_{j}^{\pi}$ denote random partial realizations with respect to the probability distribution $p(\psi_{i}^{\pi})$ and $p(\psi_{j}^{\pi})$, respectively. In addition, For every realization $\phi$ and any $i\leq j$, according to the nature of policy $\pi$, we have $\Psi_{i}(\pi,\phi)\subseteq \Psi_{j}(\pi,\phi)$. Thus, we can partition $\psi_{j}^{\pi}$ based on $\psi_{i}^{\pi}$. Then,
	\begin{align*}
		&\E_{\Phi}\big[\max_{e}\Delta(e\mid \Psi_{i}(\pi,\Phi))\big]\\
		&=\E_{\Psi_{j}^{\pi}}\big[\max_{e}\Delta(e\mid \Psi_{j}^{\pi})\big]\\
		&=\E_{\Psi_{i}^{\pi}}\Big[\E_{\Psi_{j}^{\pi}}\big[\max_{e}\Delta(e\mid \Psi_{j}^{\pi})\mid \Psi_{i}^{\pi}\subseteq \Psi_{j}^{\pi}\big]\Big]\\
		&=\E_{\Psi_{i}^{\pi}}\Big[\E_{\Psi_{j}^{\pi}}\big[\Delta(e^\ast(\Psi_{j}^{\pi})\mid \Psi_{j}^{\pi})\mid \Psi_{i}^{\pi}\subseteq \Psi_{j}^{\pi}\big]\Big]\\
		&\leq \E_{\Psi_{i}^{\pi}}\Big[\E_{\Psi_{j}^{\pi}}\big[\Delta(e^\ast(\Psi_{j}^{\pi})\mid \Psi_{i}^{\pi})\mid \Psi_{i}^{\pi}\subseteq \Psi_{j}^{\pi}\big]\Big]\\
		&\leq \E_{\Psi_{i}^{\pi}}\Big[\E_{\Psi_{j}^{\pi}}\big[\max_{e}\Delta(e\mid \Psi_{i}^{\pi})\mid \Psi_{i}^{\pi}\subseteq \Psi_{j}^{\pi}\big]\Big]\\
		&=\E_{\Psi_{i}^{\pi}}\big[\max_{e}\Delta(e\mid \Psi_{i}^{\pi})\big]\\
		&=\E_{\Phi}\big[\max_{e}\Delta(e\mid \Psi_{j}(\pi,\Phi))\big],
	\end{align*}
	where $e^\ast(\Psi_{j}^{\pi}):=\argmax_{e}\Delta(e\mid \Psi_{j}^{\pi})$ for each $\Psi_{j}^{\pi}$. The first inequality is due to the adaptive submodularity of $f$, and the second inequality is because $\Delta(e^\ast(\Psi_{j}^{\pi})\mid \Psi_{i}^{\pi})\leq \max_{e}\Delta(e\mid \Psi_{i}^{\pi})$ for each  $\Psi_{i}^{\pi}$.
	\qed
\end{proof}

Using Lemma~\ref{lemma:marginal_max} and Lemma~\ref{lemma:step_submodular}, we can build a quantitative relationship between any policy $\pi$ and the optimal adaptive policy, as shown by Lemma~\ref{lemma:opt_max}.

\begin{lemma}\label{lemma:opt_max}
	For any deterministic adaptive policy $\pi$, any deterministic policy $\pi^\ast$ selecting $r$ items, and any $0\leq i\leq r$, we have
	\begin{equation*}
	f_{\avg}(\pi^{\ast})-f_{\avg}(\pi_{i})\leq r\cdot \E_{\Phi}\big[\max_{e}\Delta(e\mid \Psi_{i}(\pi,\Phi))\big].
	\end{equation*}
\end{lemma}
\begin{proof}[Lemma~\ref{lemma:opt_max}]
	Each deterministic policy $\pi$ can be associated with a decision tree $T^{\pi}$ in a natural way. Each node in the decision tree is a partial realization $\psi$ such that the policy picks item $\pi(\psi)$ and the children of $\psi$ will be observed under respective realizations. Furthermore, each node $\psi$ is associated with a reward $r(\psi):=\Pr[\Phi\sim\psi]\cdot \Delta(\pi(\psi)\mid \psi)$ which is nonnegative due to the adaptive monotonicity of $f$, \ie~$\Delta(e\mid \psi)\geq 0$ for every $e$ and every $\psi$. Then, we can get that $f_{\avg}(\pi)=\sum_{\psi\in T^{\pi}}r(\psi)$. In addition, it is easy to see that $T^{\pi}\subseteq T^{\pi\oplus\pi^\prime}$. Thus, $f_{\avg}(\pi\oplus\pi^\prime)-f_{\avg}(\pi)=\sum_{\psi\in (T^{\pi\oplus\pi^\prime}\setminus T^{\pi})}r(\psi)\geq 0$. Meanwhile, it is easy to verify that $f_{\avg}(\pi\oplus\pi^\prime)=f_{\avg}(\pi^\prime\oplus\pi)$, since $\pi\oplus\pi^\prime$ and $\pi^\prime\oplus\pi$ pick the same items under every realization.
	
	For rotational convenience, let $\hat{\pi}:=\pi_{i}\oplus\pi^{\ast}$ and $\hat{\pi}_{i+j}:=\pi_{i}\oplus\pi^{\ast}_j$ for any $j\geq 0$. Thus, we have
	\begin{align*}
	f_{\avg}(\pi^{\ast})-f_{\avg}(\pi_{i})
	&\leq f_{\avg}(\hat{\pi})-f_{\avg}(\pi_{i})\\
	&=\sum_{j=1}^{r} \Big(f_{\avg}(\hat{\pi}_{i+j})-f_{\avg}(\hat{\pi}_{i+j-1})\Big)\\
	&\leq\sum_{j=1}^{r} \E_{\Phi}\big[\max_{e}\Delta(e\mid \Psi_{i+j-1}({\hat{\pi}},\Phi))\big]\\
	&\leq\sum_{j=1}^{r} \E_{\Phi}\big[\max_{e}\Delta(e\mid \Psi_{i}({\hat{\pi}},\Phi))\big]\\
	&=r\cdot \E_{\Phi}\big[\max_{e}\Delta(e\mid \Psi_{i}({{\pi}},\Phi))\big].
	\end{align*}
	The first inequality is due to the adaptive monotonicity of $f$ as discussed above. The second inequality is by Lemma~\ref{lemma:marginal_max} while the third inequality is by Lemma~\ref{lemma:step_submodular}. The final equality is because $\hat{\pi}_i=\pi_i$.\qed
\end{proof}

Then, we are able to establish a relationship between our proposed randomized policy $\pi^\ag$ and any randomized policy in the following lemma.

\begin{lemma}\label{lemma:opt_greedy}
	Let $\pi^{\ag}$ be a randomized policy that returns an expected $\alpha_i$-approximate solution for the $i$-th item selection under every partial realization $\psi_{i-1}$. For any $0\leq i< r$ and any randomized policy $\pi^{\ast}$, we have
	\begin{equation}\label{eqn:toprove1}
	\begin{split}
		&\EO\big[f_{\avg}(\pi^{\ast}(\omega))-f_{\avg}(\pi^{\ag}_{i}(\omega))\big]\\
		&\leq (1-\tfrac{\alpha_i}{r})\cdot \EO\big[f_{\avg}(\pi^{\ast}(\omega))-f_{\avg}(\pi^{\ag}_{i-1}(\omega))\big],
	\end{split}
	\end{equation}
	where the expectation is over the randomness of policy.
\end{lemma}
\begin{proof}[Lemma~\ref{lemma:opt_greedy}]
	We first fix the randomness of $\omega$. Then, $\pi^{\ast}(\omega)$ and $\pi^{\ag}(\omega)$ are deterministic policies. By Lemma~\ref{lemma:opt_max}, we have
	\begin{align*}
	&f_{\avg}(\pi^{\ast}(\omega))-f_{\avg}(\pi^{\ag}_{i-1}(\omega))\\
	&\leq r\cdot \E_{\Phi}\big[\max_{e}\Delta(e\mid \Psi_{i-1}(\omega,\Phi))\big].
	\end{align*}
	Taking the expectation over the randomness of $\omega$ gives
	\begin{equation}\label{eqn:right_relation}
	\begin{split}
	&\EO\big[f_{\avg}(\pi^{\ast}(\omega))-f_{\avg}(\pi^{\ag}_{i-1}(\omega))\big]\\
	&\leq r\cdot \EO\big[\E_{\Phi}\big[\max_{e}\Delta(e\mid \Psi_{i-1}(\omega,\Phi))\big]\big].
	\end{split}
	\end{equation}
	
	On the other hand, by definition, we have
	\begin{align*}
	&\EO\big[f_{\avg}(\pi^{\ag}_{i}(\omega))-f_{\avg}(\pi^{\ag}_{i-1}(\omega))\big]\\
	&=\EO\big[\E_{\Phi}\big[\Delta(\pi^{\ag}_i(\omega)\mid \Psi_{i-1}(\omega,\Phi))\big]\big]\\
	&=\E_{\Phi}\big[\EO\big[\Delta(\pi^{\ag}_i(\omega)\mid \Psi_{i-1}(\omega,\Phi))\big]\big]
	\end{align*}
	In addition, for any given realization $\phi$, let $\Psi_{i-1}^{\phi}$ be the random partial realization following the probability distribution $p(\psi_{i-1}^{\phi}):=\Pr[\Psi_{i-1}(\omega,\phi)=\psi_{i-1}^{\phi}]$ over the randomness of $\omega$. Then,
	\begin{align*}
	&\EO\big[\Delta(\pi^{\ag}_i(\omega)\mid \Psi_{i-1}(\omega,\phi))\big]\\
	&=\E_{\Psi_{i-1}^{\phi}}\big[\EO\big[\Delta(\pi^{\ag}_i(\omega)\mid \Psi_{i-1}^{\phi})\mid \Psi_{i-1}(\omega,\phi)=\Psi_{i-1}^{\phi}\big]\big]\\
	&\geq \E_{\Psi_{i-1}^{\phi}}\big[\alpha_i\cdot \max_{e}\Delta(e\mid \Psi_{i-1}^{\phi})\big]\\
	&=\alpha_i\cdot \E_{\Psi_{i-1}^{\phi}}\big[\max_{e}\Delta(e\mid \Psi_{i-1}^{\phi})\big]\\
	&=\alpha_i\cdot\EO\big[\max_{e}\Delta(e\mid \Psi_{i-1}(\omega,\phi))\big].
	\end{align*}
	Therefore,
	\begin{equation}\label{eqn:left_relation}
	\begin{split}
		&\EO\big[f_{\avg}(\pi^{\ag}_{i}(\omega))-f_{\avg}(\pi^{\ag}_{i-1}(\omega))\big]\\
		&\geq\alpha_i\cdot\E_{\Phi}\big[\EO\big[\max_{e}\Delta(e\mid \Psi_{i-1}(\omega,\Phi))\big]\big].
	\end{split}
	\end{equation}
	
	Combing \eqref{eqn:right_relation} and \eqref{eqn:left_relation} yields
	\begin{align*}
		&\EO\big[f_{\avg}(\pi^{\ast}(\omega))-f_{\avg}(\pi^{\ag}_{i-1}(\omega))\big]\\
		&\leq \tfrac{r}{\alpha_i}\cdot \EO\big[f_{\avg}(\pi^{\ag}_{i}(\omega))-f_{\avg}(\pi^{\ag}_{i-1}(\omega))\big].
	\end{align*}
	Rearranging it completes the proof.
	\qed
\end{proof}

Finally, we are ready to prove Theorem~\ref{thm:greedy-ept-approx}.
\begin{proof}[Theorem~\ref{thm:greedy-ept-approx}]
	Note that for any $x$ such that $0\leq x\leq 1$, we have $1-x\leq \e^{-x}$. Therefore, recursively applying Lemma~\ref{lemma:opt_greedy} gives
	\begin{align*}
	&\EO\big[f_{\avg}(\pi^{\ast}(\omega))-f_{\avg}(\pi^{\ag}_{r}(\omega))\big]\\
	&\leq \e^{-\alpha_i/r}\cdot \EO\big[f_{\avg}(\pi^{\ast}(\omega))-f_{\avg}(\pi^{\ag}_{r-1}(\omega))\big]\\
	&\leq \cdots\\
	&\leq \e^{-\sum_{i=1}^r\alpha_i/r}\cdot \EO\big[f_{\avg}(\pi^{\ast}(\omega))-f_{\avg}(\pi^{\ag}_{0}(\omega))\big]\\
	&= \e^{-\alpha}\cdot \EO\big[f_{\avg}(\pi^{\ast}(\omega))\big]
	\end{align*}
	Rearranging it completes the proof.
	\qed
\end{proof}

\subsubsection{Worst-case Approximation Guarantee}
In what follows, we derive another concept of \textit{worst-case} approximation guarantee for policy $\pi^{\ag}$. To begin with, we first provide a \textit{random} approximation guarantee for $\pi^{\ag}$ as follows.

\begin{lemma}\label{lemma:frame-greedy-random}
Let $X_i(\omega)$ be the overall random approximation for the $i$-th item selection achieved by $\pi^{\ag}(\omega)$ with respect to $\omega$, \ie 
\begin{equation*}
X_i(\omega):=\frac{\E_{\Phi}\big[\Delta(\pi^{\ag}_i(\omega)\mid \Psi_{i-1}(\omega,\Phi))\big]}{\E_{\Phi}\big[\max_{e}\Delta(e\mid \Psi_{i-1}(\omega,\Phi))\big]}.
\end{equation*}
Then, $\pi^{\ag}(\omega)$ achieves a random approximation guarantee of $1-\e^{-X(\omega)}$, where ${X(\omega)}=\frac{1}{r}\sum_{i=1}^{r}X_i(\omega)$.
\end{lemma}
\begin{proof}[Lemma~\ref{lemma:frame-greedy-random}]
	By definition, we have
	\begin{align*}
	&f_{\avg}(\pi^{\ag}_{i}(\omega))-f_{\avg}(\pi^{\ag}_{i-1}(\omega))\\
	&=\E_{\Phi}\big[\Delta(\pi^{\ag}_i(\omega)\mid \Psi_{i-1}(\omega,\Phi))\big]\\
	&=X_i(\omega)\cdot \E_{\Phi}\big[\max_{e}\Delta(e\mid \Psi_{i-1}(\omega,\Phi))\big].
	\end{align*}
	Together with Lemma~\ref{lemma:opt_max}, we have
	\begin{align*}
	&f_{\avg}(\pi^{\ast}(\omega))-f_{\avg}(\pi^{\ag}_{i}(\omega))\\
	&\leq \Big(1-\tfrac{X_i(\omega)}{r}\Big)\cdot \Big(f_{\avg}(\pi^{\ast}(\omega))-f_{\avg}(\pi^{\ag}_{i-1}(\omega))\Big).
	\end{align*}
	Using a similar argument for the proof of Theorem~\ref{thm:greedy-ept-approx} immediately concludes Lemma~\ref{lemma:frame-greedy-random}.
	\qed
\end{proof}

Note that in Lemma~\ref{lemma:frame-greedy-random}, for any given realization $\phi$, the conditional expected marginal benefit of policy $\pi^{\ag}_i(\omega)$ based on $\Psi_{i-1}(\omega,\phi)$ equals to that of item $\pi^{\ag}(\omega,\Psi_{i-1}(\omega,\phi))$. The random approximation guarantee in Lemma~\ref{lemma:frame-greedy-random} is crucial in providing worst-case theoretical guarantees. 

To this end, a simple and intuitive idea is to show that $X_i(\omega)\geq \alpha_i$ with high probability for every $i$. Then, by a union bound for $r$ rounds of item selection, we can obtain a worst-case approximation guarantee. However, it is hard to derive such a non-trivial worst-case approximation guarantee, as the probability that $X_i(\omega)< \alpha_i$ could be large even though the probability of $\xi(\pi^{\ag}(\omega),\psi_{i-1})<\alpha_i$ on any given $\psi_{i-1}$ is small, where $\xi(\pi^{\ag}(\omega),\psi_{i-1})=\frac{\Delta(\pi^{\ag}_i(\omega)\mid \psi_{i-1})}{\max_{e}\Delta(e\mid \psi_{i-1})}$. To explain, the number of possible $\psi_{i-1}$'s can be as large as an exponential scale size, e.g.,~$O(2^m)$ realizations of influence propagation where $m$ is the number of edges in $G$. Once there exists one instance of $\psi_{i-1}$ such that $\xi(\pi^{\ag}(\omega),\psi_{i-1})<\alpha_i$, it is possible that $X_i(\omega)< \alpha_i$. In other words, to ensure that $X_i(\omega)\geq \alpha_i$, one sufficient way is to guarantee that $\xi(\pi^{\ag}(\omega),\psi_{i-1})\geq\alpha_i$ for every $\psi_{i-1}$. However, such a requirement is too stringent to satisfy. Unfortunately, two recent papers \cite{Han_AIM_2018} and \cite{Sun_MRIM_2018} claim that if $\Pr[\xi(\pi^{\ag}(\omega),\psi_{i-1})<\alpha_i]\leq \delta_i$ for every $\psi_i$, then $\Pr[X_i(\omega)<\alpha_i]\leq \delta_i$. This misclaim makes their approximation guarantees invalid. More details are presented in Section~\ref{sec:revisited}. 

On the other hand, the strategy that demands every $X_i(\omega) \geq \alpha_i$ with high probability is also overly conservative. For example, suppose that there exists one $X_j(\omega)$ satisfying $X_j(\omega) < \alpha_j$, \ie~it fails to achieve the overall $\alpha_j$-approximation in the $j$-th item selection. Even in that case, the overall approximation ratio of $\pi^{\ag}(\omega)$ could still be better than $1-\e^{-\alpha}$, as long as there exists another $X_i(\omega)$ satisfying $X_i(\omega) \geq \alpha_i + \alpha_j - X_j(\omega)$. In other words, the deficiency of one round can be compensated, as long as there exists other rounds whose quality is above the bar by a sufficient margin.

Formally, as the approximation ratio $X_i(\omega)$ in each round of $\pi^{\ag}(\omega)$ is a random variable, the overall approximation guarantee of $\pi^{\ag}(\omega)$, namely, $1 - \e^{-X(\omega)}$, depends on the mean of all $r$ variables. Intuitively, when $r$ is sizable, $X(\omega)=\frac{1}{r}\sum_{i=1}^r X_i(\omega)$ should be concentrated to its expectation, i.e., $\EO[X(\omega)]$. Note that  $\EO[X_i(\omega)] \geq \alpha_i$ holds if $\EO[\xi(\pi^{\ag}(\omega),\psi_{i-1})] \geq \alpha_i$ holds for every $\psi_{i-1}$, which is exactly the requirement of $\pi^{\ag}(\omega)$ for each round of item selection. That is, instead of formulating the approximation ratio of $\pi^{\ag}(\omega)$ based on the worst-case guarantee of each selected item, we might derive it based on each selected item's expected approximation ratio.

To make the above idea work, the distance between $X(\omega)$ and its expectation $\EO[X(\omega)]$ is to be bounded with high probability. However, there is a challenge that we need to address. As the selection of the $i$-th item is dependent on the results of the first $(i-1)$ items, the random variables $X_1(\omega), X_2(\omega), \dotsc, X_r(\omega)$ are {\it correlated}, making it rather non-trivial to derive concentration results for $\frac{1}{r}\sum_{i=1}^r X_i(\omega)$. We circumvent this issue with a theoretical analysis by leveraging Azuma-Hoeffding inequality for martingales \cite{Mitzenmacher_Martingales_2005}.

\begin{definition}[Martingale~\cite{Mitzenmacher_Martingales_2005}]
	A sequence of random variables $Z_0,Z_1,\dots$ is a martingale with respect to the sequence $Y_0,Y_1,\dots$ if, for all $i\geq 0$, the following conditions hold:
	\begin{itemize}
		\item $Z_i$ is a function of $Y_0,Y_1,\dots,Y_i$;
		\item $\E[\abs{Z_i}]<\infty$;
		\item $\E[Z_{i+1}\mid Y_0,\dotsc,Y_{i}]=Z_{i}$.
	\end{itemize}
\end{definition}

\begin{lemma}[Azuma-Hoeffding Inequality~\cite{Mitzenmacher_Martingales_2005}]\label{lma:azuma}
	Let $Z_0,Z_1,\dots$ be a martingale with respect to the sequence of random variables $Y_0,Y_1,\dots$ such that
	\begin{equation*}
		B_i\leq Z_i-Z_{i-1}\leq B_i+c_i
	\end{equation*}
	for some constants $c_i$ and for some random variables $B_i$ that may be functions of $Y_0,\dots,Y_{i-1}$. Then, for any $r\geq 0$ and any $\lambda>0$,
	\begin{equation*}
	\Pr\big[Z_r \geq Z_0+\lambda \big]\leq \e^{-2{\lambda^2}/{\sum_{i=1}^{r}c_i^2}}.
	\end{equation*}
\end{lemma}

Based on the above Azuma-Hoeffding inequality, we provide a concentration bound for possibly correlated random variables as follows.
\begin{corollary}\label{corollary:azuma}
	Let $Y_1,\dotsc, Y_r$ be any sequence of random variables and $Y_i^\prime$ be a function of $Y_1,\dotsc, Y_i$ satisfying $\abs{Y_i^\prime}\leq \beta$ and $\E[Y_i^\prime\mid Y_1,\cdots,Y_{i-1}]\leq \gamma_i$ for every $i=1,\dots,r$. Then, we have
	\begin{equation}
	\Pr\left[\sum\nolimits_{i=1}^r Y_i^\prime \ge \sum\nolimits_{i=1}^r\gamma_i+\sqrt{r}\beta\lambda \right]\leq \e^{-\lambda^2/2}.
	\end{equation}
\end{corollary}
\begin{proof}[Corollary~\ref{corollary:azuma}]
	Let $Z_0=Y_0^\prime=0$ and for every $i=1,\dots,r$,
	\begin{equation*}
		Z_i:=\sum\nolimits_{j=1}^i(Y_j^\prime-\E[Y_j^\prime\mid Y_1,\dots,Y_{j-1}]).
	\end{equation*}
	Then, it is easy to verify that $\E[\abs{Z_i}]<\infty$ and $\E[Z_i\mid Y_1,\dots,Y_{i-1}]=Z_{i-1}$, which indicates that $Z_i$ is a martingale. In addition, let $B_i:=-\beta-\E[Y_i^\prime\mid Y_1,\dots,Y_{i-1}]$. As $\abs{Y_i^\prime}\leq \beta$, we can get that $Z_i-Z_{i-1}=Y_i^\prime-\E[Y_i^\prime\mid Y_1,\dots,Y_{i-1}]$ is in the range of $[B_i,B_i+2\beta]$. Thus, according to Lemma~\ref{lma:azuma},
	\begin{equation}\label{eqn:inequality1}
		\Pr\Big[Z_r\geq\sqrt{r}\beta\lambda\Big]\leq \e^{-2r\beta^2{\lambda^2}/{\sum_{i=1}^{r}(2\beta)^2}}=\e^{-\lambda^2/2}.
	\end{equation}
	On the other hand, as $\E[Y_i^\prime\mid Y_1,\cdots,Y_{i-1}]\leq \gamma_i$ for every  $i=1,\dots,r$, we have $Z_r\geq \sum_{i=1}^r(Y_i^\prime-\gamma_i)$. As a consequence
	\begin{equation}\label{eqn:inequality2}
		\Pr\left[\sum\nolimits_{i=1}^r(Y_i^\prime-\gamma_i)\geq \sqrt{r}\beta\lambda\right]\leq \Pr\Big[Z_r\geq \sqrt{r}\beta\lambda\Big].
	\end{equation}
	Combining \eqref{eqn:inequality1} and \eqref{eqn:inequality2} completes the proof. \qed
\end{proof}

Recall that we have shown in Lemma~\ref{lemma:frame-greedy-random} that the approximation guarantee of $\pi^{\ag}$ is determined by the summation of the approximation factors, \ie $X(\omega)=\frac{1}{r}\sum_{i=1}^{r}X_i(\omega)$, and these approximation factors could be correlated. According to Corollary~\ref{corollary:azuma}, we can get a bound on their summation, based on which we can derive the overall approximation guarantee as follows. 

\begin{theorem}\label{thm:greedy-wst-approx}
Without loss of generality, suppose that policy $\pi^{\ag}$ returns $\xi(\pi^{\ag}(\omega),\psi)$-approximate solution satisfying $c_1\leq \xi(\pi^{\ag}(\omega),\psi)\leq c_2$ for every $\omega$ and $\psi$.\footnote{Note that $c_1=0$ and $c_2=1$ can always satisfy the requirement. Thus, we can always find some $c_1$ and $c_2$ such that $0\leq c_1\leq c_2\leq 1$.} For any given $\delta \in (0,1)$, let $\alpha^\prime=\frac{1}{r} \sum_{i=1}^r \alpha_i - (c_2-c_1)\cdot\sqrt{1/(2r)\cdot\ln (1/\delta)}$. If $f$ is adaptive monotone and adaptive submodular, and $\pi^{\ag}$ returns an expected $\alpha_i$-approximate solution for the $i$-th item selection under every partial realization $\psi_{i-1}$, then $\pi^{\ag}$ achieves the worst-case approximation ratio $1-\e^{-\alpha^\prime}$ with a probability of at least $1-\delta$.
\end{theorem}

\begin{proof}[Theorem~\ref{thm:greedy-wst-approx}]
For every algorithm randomness $\omega$ and every realization $\phi$, as defined before, $\Psi_{i}(\omega, \phi)$ represents the
	partial realization corresponding to the first $i$ steps of running the adaptive greedy policy $\pi^{\ag}(\omega)$ on realization $\phi$,
	for $i=1, 2, \ldots, r$.
Then $\Psi_i(\omega, \cdot)$ is a mapping from all realizations to partial realizations that has $i$ items in the domain.
Let $A_i$ be any fixed mapping from all realizations to partial realizations that has $i$ items in the domain and
	is consistent with the corresponding realization, i.e., $\phi\sim A_i(\phi)$ for all $\phi$ and $\abs{\dom(A_i(\phi))} = i$.
Let $\Omega_{i-1}$ be the distribution of $\omega$ conditional on $\Psi_j(\omega, \cdot)=A_{j}$ for every $j<i$. 
That is, $\Omega_{i-1}$ is the probability subspace in which the adaptive greedy policy $\pi^{\ag}$ generates the first $i-1$ steps
	exactly according to $A_1, A_2, \ldots, A_{i-1}$.
For any $\omega$ sampled from $\Omega_{i-1}$, by the above definition, we have that for every realization $\phi$, 
	$\Psi_j(\omega, \phi)=A_{j}(\phi)$.
Note that, to be precise, we would include $A_1, A_2, \ldots, A_{i-1}$ in the notation $\Omega_{i-1}$, such as
	$\Omega_{i-1}(A_1, A_2, \ldots, A_{i-1})$, but for simplicity we choose the shorter notation.
Then we have
	\begin{align}
&\EO[X_i(\omega)\mid \omega\sim \Omega_{i-1}] \nonumber \\
&=\EO\bigg[\frac{\E_{\Phi}\big[\Delta(\pi^{\ag}_i(\omega)\mid \Psi_{i-1}(\omega,\Phi))\big]}{\E_{\Phi}\big[\max_{e}\Delta(e\mid \Psi_{i-1}(\omega,\Phi))\big]}\,\mathrel{\bigg|}\,\omega\sim \Omega_{i-1}\bigg] \nonumber \\
&=\EO\bigg[\frac{\E_{\Phi}\big[\Delta(\pi^{\ag}_i(\omega)\mid A_{i-1}(\Phi))\big]}{\E_{\Phi}\big[\max_{e}\Delta(e\mid  A_{i-1}(\Phi)\big]}\,\mathrel{\bigg|}\,\omega\sim \Omega_{i-1}\bigg] \nonumber \\
&=\frac{\EO\Big[\E_{\Phi}\big[\Delta(\pi^{\ag}_i(\omega)\mid A_{i-1}(\Phi))\big]\,\mathrel{\Big|}\,\omega\sim \Omega_{i-1}\Big]}{\E_{\Phi}\big[\max_{e}\Delta(e\mid A_{i-1}(\Phi))\big]} \nonumber \\
&=\frac{\E_{\Phi}\Big[\EO\big[\Delta(\pi^{\ag}_i(\omega)\mid A_{i-1}(\Phi))\,\mathrel{\big|}\,\omega\sim \Omega_{i-1}\big]\Big]}{\E_{\Phi}\big[\max_{e}\Delta(e\mid A_{i-1}(\Phi))\big]} \nonumber \\
&\geq\frac{\E_{\Phi}\big[\alpha_i\cdot\max_{e}\Delta(e\mid A_{i-1}(\Phi))\big]}{\E_{\Phi}\big[\max_{e}\Delta(e\mid A_{i-1}(\Phi))\big]}
=\alpha_i,  \label{eq:alphai}
\end{align}
	where the inequality is by the requirement of $\pi^{\ag}$ in \eqref{eqn:greedy_policy}. 
	
Note that by definition, $\EO[X_i(\omega)\mid \omega\sim \Omega_{i-1}]$ represents $\EO[X_i(\omega)\mid \Psi_j(\omega, \cdot)=A_{j},\forall j < i]$.
Inequality~\eqref{eq:alphai} means that for any fixed mappings $A_1, \ldots, A_{i-1}$, 
	$\EO[X_i(\omega)\mid \Psi_j(\omega, \cdot)=A_{j},\forall j < i] \ge \alpha_i$.
Omitting $A_1, \ldots, A_{i-1}$, we have $\EO[X_i(\omega)\mid \Psi_j(\omega, \cdot),\forall j < i] \ge \alpha_i$.
	
Next, by letting $Y_i^\prime(\omega)=1/2-(X_i(\omega)-c_1)/(c_2-c_1)$, we have $\E[Y_i^\prime(\omega)\mid \Psi_j(\omega, \cdot),\forall j < i]\leq 1/2-(\alpha_i-c_1)/(c_2-c_1)$. Meanwhile, it is easy to obtain that $c_1\leq X_i(\omega)\leq c_2$ as $c_1\leq \xi(\pi^{\ag}(\omega),\psi)\leq c_2$ for every $\omega$ and $\psi$. Thus, we also have $|Y_i^\prime(\omega)|\leq 1/2$.

	Hence, by treating $\Psi_j(\omega, \cdot)$ as the random variable $Y_j$ in Corollary~\ref{corollary:azuma}, 
	we can apply Corollary~\ref{corollary:azuma} and obtain
\begin{align*}
&\Pr[X(\omega)\leq\alpha^\prime]\\
&=\Pr\bigg[X(\omega)\leq \frac{1}{r}\sum_{i=1}^{r}\alpha_i - (c_2-c_1)\cdot\sqrt{\frac{1}{2r}\ln \frac{1}{\delta}}\bigg]\\
&=\Pr\bigg[\sum_{i=1}^r Y_i^\prime(\omega)\geq \sum_{i=1}^{r}\big(\frac{1}{2}-\frac{\alpha_i-c_1}{c_2-c_1}\big) +\sqrt{r} \frac{1}{2}\sqrt{2\ln \frac{1}{\delta}}\bigg]\\
&\leq \e^{-2\ln\frac{1}{\delta}/2}=\delta,
\end{align*}
which completes the proof.\qed
\end{proof}

Note that the worst-case approximation ratio of $\pi^{\ag}$ (Theorem~\ref{thm:greedy-wst-approx}) is worse than its expected approximation guarantee (Theorem~\ref{thm:greedy-ept-approx}), where the overall approximation factor of the latter is larger by an additive factor of $(c_2-c_1)\cdot\sqrt{1/(2r)\cdot\ln (1/\delta)}$ than the former. Furthermore, as $\pi^{\ag}$ is a $c_1$-approximate greedy policy, according to Golovin and Krause \cite{Golovin_adaptive_2011}, $\pi^{\ag}$ achieves an approximation ratio of $(1-\e^{-c_1})$. Thus, the worst-case approximation ratio is meaningful only when $\alpha^\prime\geq c_1$, \ie~$\delta \geq \e^{-\frac{2\left( \sum_{i=1}^r (\alpha_i-c_1)\right)^2}{r(c_2-c_1)^2}}$.

\subsection{Solution Framework for Adaptive IM}\label{sec:frame}

\begin{algorithm}[t]\label{alg:frame-greedy}
	\setlength{\hsize}{0.94\linewidth}
	\caption{\AG}
	\KwIn{social network $G$, seed set size $k$, batch number $r$, approximation error $\varepsilon_1, \varepsilon_2, \cdots,\varepsilon_r$}
	\KwOut{adaptively selected seed sets $S_1,\cdots,S_r$}
	$b \leftarrow k/r$\;
	$G_1 \leftarrow G$\;
	$\rho_b \leftarrow 1-(1-1/b)^b$\; \label{line:alpha}
	\For{$i\leftarrow1$ \KwTo $r$ \label{ln:iterationstart1}}{
		Identify a size-$b$ seed set $S_i$ from $G_i$, such that $S_i$ achieves an expected approximation ratio of at least $\rho_b(1-\varepsilon_i)$ on $G_i$\;\label{line:seed-identify}
		Observe the influence of $S_i$ in $G_i$\;
		Remove all nodes in $G_i$ that are influenced by $S_i$, and denote the resulting graph as $G_{i+1}$\;
	}
	\Return{$S_1,\cdots,S_r$}
\end{algorithm}

The adaptive IM under the IC model satisfies the adaptive monotonicity and adaptive submodularity \cite{Golovin_adaptive_2011}. Based on the expected $\alpha$-approximate greedy policy $\pi^{\ag}$, we propose a general framework \AG (\ie~Algorithm~\ref{alg:frame-greedy}) upon which we can build specific algorithms with seed selection algorithms to address the adaptive IM problem. At the first glance, \AG may seem similar to Vaswani and Lakshmanan's method~\cite{Vaswani_adapIM_2016}, since both techniques (i) adaptively select seed nodes in $r$ batches and (ii) do not require exact computation of expected spreads. However, there is a crucial difference between the two: Vaswani and Lakshmanan's method requires that the expected spread of {\it every} node set should be estimated with a small {\it fixed} relative error with respect to its own expectation, whereas \AG just requires an expected approximation ratio of $\rho_b(1-\varepsilon_i)$ with respect to $\OPT_b(G_i)$ (Line~\ref{line:seed-identify} in Algorithm~\ref{alg:frame-greedy}), where $\OPT_b(G_i)$ denotes the maximum expected spread of any size-$b$ seed set on $G_i$ and $\rho_b=1-(1-1/b)^b$. Note that $\rho_b$ is the approximation ratio achieved by a greedy algorithm for \MC (which is a building block for IM), and this factor cannot be further improved by any polynomial time algorithm unless \PequalNP \cite{Feige_setcover_1998}. The error requirement of \AG is much more lenient than that of Vaswani and Lakshmanan's method, and it can be achieved by several state-of-the-art solutions \cite{Tang_TIM_2014,Tang_IMM_2015,Nguyen_DSSA_2016,Tang_OPIM_2018} for vanilla influence maximization, i.e., it admits practical implementations. 

In addition, \AG is {\it flexible} in that it allows each batch of seed nodes $S_i$ to be selected with different approximation guarantee $\rho_b(1-\varepsilon_i)$, whereas the existing solutions (e.g., \cite{Golovin_adaptive_2011}) for adaptive IM require that all seed sets $S_1, \ldots, S_r$ should be processed with identical accuracy assurance. Therefore, \AG is a general framework for the adaptive IM problem. According to Theorem~\ref{thm:greedy-ept-approx} and Theorem~\ref{thm:greedy-wst-approx}, \AG can provide the following theoretical guarantees.

\begin{theorem} \label{thm:frame-greedy-approx}
	If \AG returns an expected $\rho_b(1-\varepsilon_i)$-approximate solution in the $i$-th batch of seed selection, then it achieves an expected approximation guarantee of $1-\e^{\rho_b({\varepsilon}-1)}$, where ${\varepsilon}=\frac{1}{r}\sum_{i=1}^r \varepsilon_i$, $\rho_b = 1-(1-1/b)^b$, $r$ is the number of batches and $b$ is the batch size.
	
	Meanwhile, \AG  also achieves a worst-case approximation guarantee of $1-\e^{\rho_b({\varepsilon}^\prime-1)}$ with a probability of at least $1-\delta$, where $\varepsilon^\prime=\frac{1}{r} \sum_{i=1}^r \varepsilon_i + \sqrt{1/(2r)\cdot\ln (1/\delta)}$.
\end{theorem}

%% file: 4-batch.tex
\section{Instantiations of \AG} \label{sec:batch}

In this section, we first present a naive instantiation of \AG using the state-of-the-art non-adaptive IM algorithms. To utilize the notion of expected approximation ratio for each batch of seed selection, we then design a new non-adaptive IM algorithm \EP. Finally, we analyze the approximation guarantees and time complexity of \EP and \AG instantiated with \EP respectively.

\subsection{Instantiation using Existing Algorithms} \label{sec:batch-basic}

As shown in Algorithm~\ref{alg:frame-greedy}, \AG requires identifying a random size-$b$ seed set $S_i(\omega)$ with respect to the randomness\footnote{Usually, the random source $\omega$ indicates sampling for IM, e.g.,~reverse influence sampling \cite{Borgs_RIS_2014}.} of $\pi^{\ag}(\omega)$ from the $i$-th residual graph $G_i$, such that
\begin{equation}\label{eqn:approbatch}
\EO[\E[I_{G_i}(S_i(\omega))]]\geq \rho_b(1-\varepsilon_i){\OPT}_b(G_i),
\end{equation}
where $\E[I_{G_i}(S_i(\omega))]$ is the expected spread of $S_i(\omega)$ on $G_i$ and its expectation $\EO[\cdot]$ is over the internal randomness of the algorithm, ${\OPT}_b(G_i)$ is the maximum expected spread of any size-$b$ seed set on $G_i$. For brevity, in the rest of the paper, we use $S_i$ to represent a random set $S_i(\omega)$ obtained by a randomized policy $\pi^{\ag}(\omega)$ from the $i$-th residual graph $G_i$.

We observe that such a seed set $S_i$ could be obtained by applying the state-of-the-art algorithms (e.g., \cite{Tang_TIM_2014,Tang_IMM_2015,Nguyen_DSSA_2016,Tang_OPIM_2018}) for vanilla influence maximization (IM) on $G_i$. In particular, these algorithms are randomized, and they provide a \textit{worst-case} approximation guarantee as follows: given a seed set size $b$, a relative error threshold $\varepsilon_i^\prime$ and a failure probability $\delta_i$, they output a size-$b$ seed set $S_i$ in $G_i$ whose expected spread is $\rho_b (1- \zeta(\omega,G_i))$ times the maximum expected spread of any size-$b$ seed set on $G_i$, such that $\zeta(\omega,G_i) \le \varepsilon_i^\prime$ with at least $1 - \delta_i$ probability. Thus, we obtain that
\begin{equation}
\EO[\zeta(\omega,G_i)]\leq\varepsilon_i^\prime \cdot (1-\delta_i)+\delta_i.
\end{equation}
To ensure \eqref{eqn:approbatch}, for each pair of $(\varepsilon_i^\prime,\delta_i)$, let $\delta_i$ be a sufficient small value and $\varepsilon_i^\prime:=(\varepsilon_i-\delta_i)/(1-\delta_i)\approx \varepsilon_i$. According to Theorem~\ref{thm:frame-greedy-approx}, such an instantiation of \AG yields an expected (resp.\ worst-case) approximation ratio of $1 - \e^{\rho_b(\varepsilon - 1)}$ (resp.\ $1 - \e^{\rho_b(\varepsilon^\prime - 1)}$) where $\varepsilon=\frac{1}{r}\sum_{i=1}^r \varepsilon_i$ (resp.\ $\varepsilon^\prime=\frac{1}{r}\sum_{i=1}^r \varepsilon_i+\sqrt{{1}/{2r}\cdot\ln ({1}/{\delta})}$).

But how efficient is the above instantiation? To answer this question, we need to investigate the time complexity of the vanilla IM algorithms in \cite{Tang_IMM_2015}. The theoretical analysis in \cite{Tang_IMM_2015} shows that if we are to achieve $(\rho_b - \varepsilon_i^\prime)$-approximation on $G_i$ with at least $1 - \delta_i$ probability, then the expected computation cost is $O((b\log n_i+\log\frac{1}{\delta_i})(m_i+n_i)/{\varepsilon_i^\prime}^2)$, where $n_i$ and $m_i$ denote the number of nodes and edges in $G_i$ respectively. Since $n_i \le n$ and $m_i \le m$, the expected time required to process $G_i$ is $O((b\log n+\log\frac{1}{\delta_i})(m+n)/{\varepsilon_i^\prime}^2)$. As such, all $r$ batches of seed nodes can be identified in $O(\sum_{i=1}^r(b\log n+\log\frac{1}{\delta_i})(m+n)/{\varepsilon_i^\prime}^2)$ expected time. By setting a pair of parameters $(\varepsilon_i^\prime,\delta_i)$ in the vanilla IM algorithms as $\varepsilon_i^\prime=(\varepsilon_i-\delta_i)/(1-\delta_i)$, we can achieve an expected approximation ratio of $\rho_b(1-\varepsilon_i)$ in the $i$-th batch. This shows that the total expected time complexity for achieving the final expected (resp.\ worst-case) approximation ratio of $1-\e^{\rho_b(\frac{1}{r}\sum_{i=1}^r{\varepsilon}_i-1)}$ (resp.\ $1-\e^{\rho_b(\frac{1}{r}\sum_{i=1}^r{\varepsilon}_i+\sqrt{{1}/{2r}\cdot\ln ({1}/{\delta})}-1)}$) is $O(\sum_{i=1}^r(b\log n+\log\frac{1}{\delta_i})(m+n)/{\varepsilon_i^\prime}^2)$.

\spara{Rationale for an Improved Approach} The aforementioned instantiation of \AG is straightforward and intuitive, but is far from optimized in terms of its approximation guarantee. To explain, recall that it requires each seed set $S_i$ to achieve $\rho_b(1-\varepsilon_i^\prime)$-approximation on $G_i$ with probability at least $1 - \delta_i$, based on which it provides an overall expected approximation ratio of $1 - \e^{\rho_b(\varepsilon-1)}$ with $\varepsilon=\frac{1}{r}\sum_{i=1}^r \varepsilon_i$ where $\varepsilon_i= \varepsilon_i^\prime \cdot (1-\delta_i)+\delta_i\approx \varepsilon_i^\prime$. In other words, it imposes a stringent worst-case approximation guarantee on each seed set $S_i$. This, however, might be overly conservative. Intuitively, the expected approximation error factor $\EO[\zeta(\omega,G_i)]$ should be much smaller than the naive upper bound ${\varepsilon}_i$ deduced from the worst-case approximation. To the best of our knowledge, there is no known result for vanilla IM with tight expected approximation guarantees. This motivates us to develop a vanilla IM algorithm tailored for \AG, as we show in the following section.

\subsection{IM Algorithm with Expected Approximation} \label{sec:batch-epic}

As discussed in Section~\ref{sec:batch-basic}, the existing IM algorithms provide only a worst-case approximation guarantee, i.e., the relative error factor $\zeta(\omega,G_i)$ is no more than the input threshold $\varepsilon_i^\prime$ with high probability. To optimize the performance of \AG, we are in need of one non-adaptive IM algorithm with expected approximation guarantee $\rho_b(1-\EO[\zeta(\omega,G_i)])$ such that $\EO[\zeta(\omega,G_i)]$ has a tighter bound. In what follows, we present a new non-adaptive IM algorithm, referred to as \EP\footnote{\underline{E}xpected a\underline{p}proximation for \underline{i}nfluen\underline{c}e maximization.}, that returns a solution with expected approximation guarantee. To this end, we first introduce the concept of {\it reverse reachable sets (RR-sets)} \cite{Borgs_RIS_2014}, which is the basis of our algorithm.

\spara{RR-Sets} In a nutshell, RR-sets are subgraph samples of $G$ that can be used to efficiently estimate the expected spreads of any given seed sets. Specifically, a random RR-set of $G$ is generated by first selecting a node $v\in V$ uniformly at random, and then taking the nodes that can reach $v$ in a random graph generated by independently removing each edge $(u,v)\in E$ with probability $1-p(u,v)$. If a seed node set $S$ has large expected influence spread, then the probability that $S$ intersects with a random RR-set is high, as shown in the following equation~\cite{Borgs_RIS_2014}:
\begin{equation}\label{eqn:borgsequation11}
\E[I_G(S)]=n\cdot \Pr[R\cap S\neq\emptyset],
\end{equation}
where $R$ is a random RR-set. This result suggests a simple method for estimating the expected influence spread of any node set $S$: we can use a set $\R$ of random RR-sets to estimate the value of $\Pr[R \cap S\neq\emptyset]$ and hence $\E[I_G(S)]$. In particular, let $\Cov_{\R}(S)$ denote the number of RR-sets in $\R$ that overlap $S$. Then the value of $\E[I_G(S)]$ can be unbiasedly estimated by $n\cdot F_{\R}(S)$, where
\begin{equation}\label{eqn:estimation}
F_{\R}(S)= \Cov_{\R}(S)/|\R|.
\end{equation}
By the law of large numbers,  $n\cdot F_{\R}(S)$ should converge to $\E[I_{G}(S)]$ when $|\R|$ is sufficiently large, which provides a way to estimate $\E[I_{G}(S)]$ to any desired accuracy level. However, due to the cost of generating RR-sets, there is a tradeoff between accuracy and efficiency in any algorithms using RR-set sampling.

\spara{The \EP Algorithm} Algorithm~\ref{alg:OPIM-E} shows the pseudo-code of our \EP algorithm, which borrows the idea from the \OPIMC algorithm \cite{Tang_OPIM_2018} via (i) starting from a small number of RR-sets and (ii) iteratively increasing the RR-set number until a satisfactory solution is identified. The key difference between the two algorithms lies in the way that they compute the upper bound $\ub(S_i^o)$, \ie the fraction of RR-sets in $\R_1$ covered by $S_i^o$ in each iteration where $S_i^o$ is an optimal seed set in $G_i$. In particular, in \OPIMC, the upper bound $\ub(S_i^o)$ is ensured to be no smaller than the {\it expected} fraction of RR-sets in $\R_1$ covered by $S_i^o$ with high probability. This needs \OPIMC to provide the worst-case approximation guarantee with high probability. In contrast, in each iteration of \EP, the upper bound $\ub(S_i^o)$ is only required to be no smaller than the {\it true} fraction of RR-sets in $\R_1$ covered by $S_i^o$, based on which rigorous bounds on its expected approximation guarantee can be derived. In what follows, we discuss the details of \EP and its subroutine \MC (in Algorithm~\ref{alg:max-cover}).
\begin{algorithm}[!t]
	\setlength{\hsize}{0.94\linewidth}
	\caption{{\EP}$(G_i,b,\varepsilon_i)$}\label{alg:OPIM-E}
	\KwIn{Graph $G_i$, seed set size $b$, and error threshold $\varepsilon_i$.}
	\KwOut{A $b$-size seed set $S_i$ that provides an expected approximation guarantee of at least $\rho_b(1-\varepsilon_i)$.}
	let $\delta_i\leftarrow 0.01\varepsilon_ib/n_i$\;
	let $\varepsilon_i^\prime\leftarrow (b\varepsilon_i-\delta_in_i)/(b-\delta_in_i)$\; 
	let $\varepsilon_a\leftarrow{\varepsilon_i^\prime}/{(1-\varepsilon_i^\prime)}$\;
	let $i_{\max}\leftarrow\lceil\log_2\frac{(2+2\varepsilon_a/3)n_i}{\varepsilon_a^2}\rceil+1$ and $a_i\leftarrow \ln\frac{2i_{\max}}{\delta_i}$\;\label{alg:magim_maxiteration}
	initialize $\theta_{0}\leftarrow \frac{1}{b}\big(\ln{\frac{2}{\delta_i}}+\ln{\binom{n_i}{b}}\big)$\;\label{alg:magim_maxRRsets}
	generate two sets $\mathcal{R}_1$ and $\mathcal{R}_2$ of random RR-sets, with $\abs{\mathcal{R}_1} = \abs{\mathcal{R}_2} = \theta_0$\;\label{alg:magim_RRsets2}
	\For{$i\leftarrow1$ \KwTo $i_{\max}$\label{alg:magim_loop}}{
		$\langle S_i,\ub(S_i^o) \rangle\leftarrow \MC(\R_1,b)$\; \label{alg:magim_greedy}
		$\lb(S_i)\leftarrow\big(\sqrt{F_{\R_2}(S_i)+\frac{2a_i}{9\abs{\mathcal{R}_2}}} - \sqrt{\frac{a_i}{2\abs{\mathcal{R}_2}}}\big)^2-\frac{a_i}{18\abs{\mathcal{R}_2}}$\;\label{alg:magim_lower}
		\lIf{$\frac{\lb(S_i)}{\ub(S_i^o)}\geq \rho_b(1-\varepsilon_i^\prime)$ {\bf or} $i=i_{\max}$\label{alg:magim_stop}}{\Return{$S_i$}\label{alg:magim_return}}
		double the sizes of $\mathcal{R}_1$ and $\mathcal{R}_2$ with new RR-sets\;\label{alg:magim_doubleRRsets}
	}
\end{algorithm}

\begin{algorithm}[!t]
	\setlength{\hsize}{0.94\linewidth}
	\caption{{\MC}$(\R,b)$}
	\label{alg:max-cover}
	\KwIn{A set $\R$ of random RR-sets, and seed set size $b$.}
	\KwOut{A node set $S_i$, and an upper bound $\ub(S_i^o)$ on the fraction of RR-sets in $\R$ covered by $S_i^o$.}
	$S_i\leftarrow\emptyset$\;
	$\ub(S_i^o)\leftarrow \ub(S_i^o\mid S_i)$ which is computed by \eqref{eqn:upperbound}\;
	\For{$i\leftarrow1$ \KwTo $b$}
	{
		$u \leftarrow \arg\max_{v\in V_i}\left(\Cov_{\R}(S_i \cup \{v\} ) - \Cov_{\R}(S_i)\right)$\;
		insert $u$ into $S_i$\;
		compute $\ub(S_i^o\mid S_i)$ by \eqref{eqn:upperbound} based on the new $S_i$\;
		update $\ub(S_i^o)\leftarrow\min\{\ub(S_i^o), \ub(S_i^o\mid S_i)\}$\;
	}
	\Return $\langle S_i,\ub(S_i^o) \rangle$\;
\end{algorithm}

Based on the RR-set sampling method described previously, a simple approach for selecting $S_i$ with a large expected influence spread is to first generate a set $\R$ of RR-sets, and then invoke the \MC algorithm on $\R$. In particular, \MC uses a simple greedy approach to identify $S_i\subseteq V_i$ such that $S_i$ overlaps as many RR-sets in $\R$ as possible. Since $F_{\R}(\cdot)$ is a submodular function for any set $\R$ of RR-sets~\cite{Borgs_RIS_2014}, given any node set $S\subseteq V$ with $|S|\le b$, we know that
\begin{equation}\label{eqn:upperbound}
\ub(S_i^o\mid S)=F_{\R}(S)+\!\!\!\!\!\!\!\!\!\!\!\!\sum_{v\in\maxMC(S,b)}\!\!\!\!\!\!\!\!\!\!\!\!\left(F_{\R}(S \cup \{v\} ) - F_{\R}(S)\right)
\end{equation}
is an upper bound on $F_{\R}(S_i^o)$, where $S_i^o$ is an optimal seed set in $G_i$ and $\maxMC(S,b)$ is the set of $b$ nodes with the top-$b$ largest marginal coverage in $\R$ with respect to $S$. As a consequence, the smallest one $\ub(S_i^o)=\min_{S_i}\{\ub(S_i^o\mid S_i)\}$ during the greedy procedure ensures that
\begin{equation} \label{eqn:ubound1}
\ub(S_i^o)\geq F_{\R}(S_i^o).
\end{equation}
In addition, according to \cite{Tang_OPIM_2018}, we also have
\begin{equation} \label{eqn:ubound2}
F_{\R}(S_i)\geq \rho_b\ub(S_i^o)
\end{equation}
where $\rho_b$ is as defined in Algorithm~\ref{alg:frame-greedy}.
Putting \eqref{eqn:ubound1} and \eqref{eqn:ubound2} together yields
\begin{equation} \label{eqn:naiveapproach}
F_{\R}(S_i)\geq \rho_b\ub(S_i^o)\geq \rho_b F_{\R}(S_i^o).
\end{equation}
Thus, when $|\R|$ is large, the approximation guarantee of $S_i$ converges to $\rho_b$ according to Equation~\eqref{eqn:naiveapproach}.

To strike a balance between the quality of $S_i$ and the number of RR-sets used to derive $S_i$, \EP iterates in a careful manner as follows. In each iteration, it maintains two sets of random RR-sets $\R_1$ and $\R_2$ with $|\R_1|=|\R_2|$. It invokes \MC on $\R_1$ to identify a seed set $S_i$, and then utilizes $\R_2$ to test whether $S_i$ provides a good approximation guarantee. Initially, the cardinalities of $\R_1$ and $\R_2$ are small constants determined by the parameter $\theta_0$ in Line~\ref{alg:magim_maxRRsets} in the first iteration of \EP. Then, whenever \EP finds that the quality of the seed set $S_i$ generated in an iteration is not satisfactory, it doubles the sizes of $\R_1$ and $\R_2$. This process repeats until that a qualified solution is identified or the sizes of $\R_1$ and $\R_2$ reach $2^{i_{\max}-1}\theta_0$ which exceeds the threshold $\frac{(2+2\varepsilon_a/3) n_i}{\varepsilon_a^2b }\big(\ln{\frac{2}{\delta_i}}+\ln{\binom{n_i}{b}}\big)$ (Line~\ref{alg:magim_stop}).

As explained before, one of the main designing goals for \EP is to achieve an expected approximation ratio of $\rho_b(1-\varepsilon_i)$. \EP achieves this goal by a series of operations in each iteration, whose implications are briefly explained as follows.

In each iteration, \EP first applies \MC on $\R_1$ (Line~\ref{alg:magim_greedy}), which returns a seed set $S_i$ and an upper bound $\ub(S_i^o)$ on $F_{\R_1}(S_i^o)$, i.e.,
\begin{equation}\label{eqn:ubound}
\ub(S_i^o)\geq F_{\R_1}(S_i^o).
\end{equation}
After that, \EP uses $\R_2$ to estimate the expected spread of $S_i$ (i.e., $\E[I_{G_i}(S_i)]$). Observe that $|\R_2|F_{\R_2}(S_i)$ is a binomial random variable due to Equation~\eqref{eqn:estimation}. Accordingly, \EP uses the Chernoff-like martingale concentration bound to set a threshold $\lb(S_i)$ (Line~\ref{alg:magim_lower}) such that
\begin{equation}\label{eqn:lbound}
\E[n_iF_{\R_2}(S_i)]\geq n_i\lb(S_i)
\end{equation}
should hold with high probability. Intuitively, Equation~\eqref{eqn:lbound} implies that $n_i\lb(S_i)$ gives a sufficiently accurate lower bound on $\E[I_{G_i}(S_i)]$. After that, \EP checks whether
\begin{equation}\label{eqn:stopcondition}
\lb(S_i)/\ub(S_i^o)\geq \rho_b(1-\varepsilon_i^\prime)
\end{equation}
holds in Line~\ref{alg:magim_stop}. Intuitively, if Equation~\eqref{eqn:stopcondition} is true, then we know that $\E[n_iF_{\R_2}(S_i)]$ is no smaller than $\rho_b(1-\varepsilon_i^\prime)n_iF_{\R_1}(S_i^o)$ and it suffices to conclude our result by taking the expectation. Specifically, combining Equations~\eqref{eqn:ubound}--\eqref{eqn:stopcondition} and taking the expectation with respect to the randomness of the algorithm, we can derive a quantitative relationship between $\EO[\E [I_{G_i}(S_i)]]$ and ${\OPT}_b(G_i)$ when a seed set $S_i$ is returned:
\begin{align*}
&\EO[\E [I_{G_i}(S_i)]] \\
&\geq \EO[n_i\lb(S_i)]-\delta_i\rho_b(1-\varepsilon_i^\prime)n_i\\
&\geq \EO[\rho_b(1-\varepsilon_i^\prime)n_i\ub(S_i^o)]-\delta_i\rho_b(1-\varepsilon_i^\prime)n_i\\
&\geq \EO[\rho_b(1-\varepsilon_i^\prime)n_i F_{\R_1}(S_i^o)]-\delta_i\rho_b(1-\varepsilon_i^\prime)n_i\\
&=\rho_b(1-\varepsilon_i^\prime){\OPT}_b(G_i)-\delta_i\rho_b(1-\varepsilon_i^\prime)n_i\\
&\geq\rho_b(1-\varepsilon_i^\prime)(1-\delta_in_i/b){\OPT}_b(G_i),
\end{align*}
where the first inequality is due to the fact that \eqref{eqn:lbound} holds with high probability and $\delta_i\rho_b(1-\varepsilon_i^\prime)n_i$ is used to offset the failed scenario, and the equality is due to the martingale stopping theorem \cite{Mitzenmacher_Martingales_2005} (see details in Section~\ref{sec:exp-approx}). This proves the $\rho_b(1-\varepsilon_i)$ expected approximation ratio of $\EO[\E[I_{G_i}(S_i)]]$ as $\varepsilon_i=\varepsilon_i^\prime+(1-\varepsilon_i^\prime)\delta_in_i/b$.

It is easy to see that the expected approximation guarantee of \EP is better than those of vanilla IM algorithms, and thus instantiating \AG using \EP can lead to performance improvement for adaptive IM. Note that \EP does not provide the worst-case approximation guarantee with high probability, as against the state-of-the-art IM algorithms. The reason behind is that $n_i\ub(S_i^o)$ is likely to be smaller than ${\OPT}_b(G_i)$ though $n_i\ub(S_i^o)$ is an upper bound on $n_iF_{\R_1}(S_i^o)$ as shown in \eqref{eqn:ubound}.

\subsection{Theoretical Analysis of \EP}\label{sec:exp-approx}

Based on the discussions in Section~\ref{sec:batch-epic}, we show the details of theoretical analysis of \EP. We prove our main results for the expected approximation guarantee and the time complexity of \EP as follows.

\spara{Expected Approximation Guarantee} We establish the expected approximation guarantee of \EP in the following theorem.
\begin{theorem}\label{thm:batch_guarantee}
For any $G_i$, \EP returns a seed set $S_i$ satisfying 
\begin{equation}\label{eqn:exp-approx}
\EO[\E[I_{G_i}(S_i)]]\geq \rho_b(1-\varepsilon_i){\OPT}_b(G_i).
\end{equation}
\end{theorem}

To prove Theorem~\ref{thm:batch_guarantee}, we first prove the following lemma.
\begin{lemma}\label{lemma:maxsamples}
	Let $\varepsilon_a\in(0,1)$, $\delta_a\in(0,1)$, and
	\begin{equation}\label{eqn:thetamax}
		\theta_{\max}=\frac{(2+2\varepsilon_a/3) n_i}{\varepsilon_a^2b}\big(\ln{\frac{1}{\delta_a}}+\ln{\tbinom{n_i}{b}}\big).
	\end{equation}
	If a set of random RR-sets $\R$ are generated such that $\abs{\R}\geq\theta_{\max}$, then with probability at least $1-\delta_a$, the greedy algorithm returns a solution $S_i$ satisfying
	\begin{equation}
	\E[I_{G_i}(S_i)]\geq \frac{\rho_b}{1+\varepsilon_a} n_i F_{\R}(S_i^o).
	\end{equation}
\end{lemma}
\begin{proof}[Lemma~\ref{lemma:maxsamples}]
	According to Chernoff-like martingale concentration bound \cite{Tang_IMM_2015}, for any $S$ that is independent of $\R$, we have
	\begin{align*}
		&\Pr\big[n_i F_{\R}(S)>(1+\varepsilon_a)\E[n_i F_{\R}(S)]\big]\\
		&=\exp\left(-\frac{\varepsilon_a^2}{2+2\varepsilon_a/3}\cdot\frac{\abs{\R}\cdot\E[n_i F_{\R}(S)]}{n_i}\right)\\
		&\leq \delta_a/\tbinom{n_i}{b}.
	\end{align*}
	By the union bound, $S_i$ returned by the greedy algorithm satisfies
	\begin{equation}\label{eqn:greedysolution1}
		\Pr\big[n_i F_{\R}(S_i)>(1+\varepsilon_a)\E[n_i F_{\R}(S_i)]\big]\leq \delta_a.
	\end{equation}
	On the other hand, due to the submodularity of $F_{\R}(\cdot)$, we have
	\begin{equation}\label{eqn:greedysolution2}
	 F_{\R}(S_i)\geq \rho_b F_{\R}(S_i^o).
	\end{equation}
	Combining \eqref{eqn:greedysolution1} and \eqref{eqn:greedysolution2} completes the proof. \qed
\end{proof}

Next, we use the following martingale stopping theorem \cite{Mitzenmacher_Martingales_2005} to prove Theorem~\ref{thm:batch_guarantee}.

\begin{definition}[Stopping Time~\cite{Mitzenmacher_Martingales_2005}]
A nonnegative, integer-valued random variable $T$ is a stopping time for the sequence $\{Z_n, n\geq 0\}$ if the event $T=n$ depends only on the value of the random variables $Z_0,Z_1,\dotsc,Z_n$.
\end{definition}

\begin{lemma}[Martingale Stopping Theorem~\cite{Mitzenmacher_Martingales_2005}]\label{lemma:OptionalStoppingTheorem}
	If $Z_0,Z_1,\dots$ is a martingale with respect to $Y_1,Y_2,\dots$ and if $T$ is a stopping time for $Y_1,Y_2,\dots$, then
	\begin{equation}
	\E[Z_T]=\E[Z_0]
	\end{equation}
	whenever one of the following holds:
	\begin{itemize}
		\item the $Z_i$ are bounded, so there is a constant $c$ such that, for all $i$, $\abs{Z_i}\leq c$;
		\item $T$ is bounded;
		\item $\E[T]<\infty$, and there is a constant $c$ such that $\E[\abs{Z_{i+1}-Z_{i}}\mid Y_1,\dots,Y_i]<c$.
	\end{itemize}
\end{lemma}

\begin{proof}[Theorem~\ref{thm:batch_guarantee}] 
	Let $\mathcal{E}_1$ and $\mathcal{E}_2$ denote the following events:
	\begin{align*}
	&\mathcal{E}_1(S_i)\colon \E[I_{G_i}(S_i)]\geq \rho_b(1-\varepsilon_i^\prime) n_i F_{\R_1}(S_i^o),\\
	&\mathcal{E}_2(S_i)\colon \E[I_{G_i}(S_i)]\geq n_i\lb(S_i).
	\end{align*}
	Let $T$ be the stopping time (i.e., the iteration in which \EP returns $S_i$), which is bounded by $i_{\max}$. Let $\varepsilon_a={\varepsilon_i^\prime}/{(1-\varepsilon_i^\prime)}$ and $\delta_a=\delta_i/2$ for $\theta_{\max}$ defined in \eqref{eqn:thetamax}. When $T=i_{\max}$, it is easy to verify that $\abs{\R_1}=2^{i_{\max}-1}\theta_0\geq \theta_{\max}$. Hence, by Lemma~\ref{lemma:maxsamples}, we have
	\begin{equation}\label{eqn:lastiteration}
	\Pr[(T=i_{\max})\wedge\neg\mathcal{E}_1(S_i)]\leq \delta_i/2.
	\end{equation}
	On the other hand, when $T=t<i_{\max}$, let $\mathcal{S}_t$ be the set of possible node sets selected by \EP (but not necessarily returned), where each $S\in \mathcal{S}_t$ has a probability $\Pr[S]$ such that $\sum_{S\in\mathcal{S}_t}\Pr[S]=1$. Then, we have
	\begin{align*}
	&\Pr[(T=t)\wedge\neg\mathcal{E}_1(S_i)]\\
	&\leq\Pr[(T=t)\wedge\neg\mathcal{E}_2(S_i)] \\
	&\leq\sum_{S\in\mathcal{S}_t}\Pr[(T=t)\wedge\neg\mathcal{E}_2(S)]\cdot \Pr[S_i=S]\\
	&\leq\sum_{S\in\mathcal{S}_t}\delta_i/(2i_{\max})\cdot \Pr[S_i=S]\\
	&\leq \delta_i/(2i_{\max}),
	\end{align*}
	where the first inequality is because if $\mathcal{E}_2$ happens then $\mathcal{E}_1$ must also happen, the second inequality is by the fact that only a subset of the node sets in $\mathcal{S}_t$ are returned, and the third inequality is obtained from \cite{Tang_OPIM_2018} for any node set $S$ that is independent of $\R_2$. As a consequence, by a union bound,
	\begin{equation}\label{eqn:beforelastiteration}
	\Pr\bigg[\bigvee_{t=1}^{i_{\max}-1} \Big((T=t)\wedge\neg\mathcal{E}_1(S_i)\Big)\bigg]\leq \delta_i/2.
	\end{equation}
	Combining Equations~\eqref{eqn:lastiteration} and \eqref{eqn:beforelastiteration} shows that the event $\mathcal{E}_1(S_i)$ does not happen with probability at most $\delta_i$ no matter when the algorithm stops. Therefore, \EP returns a random solution $S_i$ satisfying $\E [I_{G_i}(S_i)]\geq \rho_b(1-\varepsilon_i^\prime)n_i F_{\R_1}(S_i^o)$ with at least $1-\delta_i$ probability. Thus, adding an additive factor of $\delta_i\rho_b(1-\varepsilon_i^\prime)n_i$ ensures that 
	\begin{equation}\label{eqn:allept}
	\begin{split}
	&\EO[\E[I_{G_i}(S_i)]]+\delta_i\rho_b(1-\varepsilon_i^\prime)n_i\\
	&\geq \rho_b(1-\varepsilon_i^\prime) \EO[n_i F_{\R_1}(S_i^o)].
	\end{split}
	\end{equation}
	Subsequently, the main challenge lies in how we connect $n_i F_{\R_1}(S_i^o)$ with ${\OPT}_b(G_i)$. Note that $F_{\R_1}(S_i^o)$ is also a random variable with respect to $\R_1$. At the first glance, it seems that this analysis is difficult as the stopping time is a random variable. However, fortunately, by utilizing the martingale stopping theorem \cite{Mitzenmacher_Martingales_2005}, we can bridge the gap between $n_i F_{\R_1}(S_i^o)$ and ${\OPT}_b(G_i)$ as follows.
	
	Note that $T$ is bounded within $i_{\max}$ so that it satisfies the condition of martingale stopping theorem as shown in Lemma~\ref{lemma:OptionalStoppingTheorem}. Thus, we have
	\begin{equation}\label{eqn:ost}
		\EO[n_i F_{\R_1}(S_i^o)]-{\OPT}_b(G_i)=0.
	\end{equation}
	Combining Equations~\eqref{eqn:allept} and \eqref{eqn:ost} yields:
	\begin{equation}
		\EO[\E[I_{G_i}(S_i)]]\geq \rho_b(1-\varepsilon_i^\prime)(1-\delta_in_i/b){\OPT}_b(G_i).
	\end{equation}
	By replacing $\varepsilon_i^\prime=(b\varepsilon_i-\delta_in_i)/(b-\delta_in_i)$, we can immediately acquire Equation~\eqref{eqn:exp-approx}, by which we complete the proof of Theorem~\ref{thm:batch_guarantee}.\qed
\end{proof}	

\spara{Time Complexity} The expected time complexity of \EP is given in the following theorem.
\begin{theorem}\label{thm:time-complexity}
For any $G_i$, the expected time complexity of \EP is ${O((b\log n_i+\log\frac{1}{\varepsilon_i})(m_i+n_i)/\varepsilon_i^2)}$, where $n_i$ and $m_i$ are the number of nodes and edges of $G_i$, respectively.
\end{theorem}

\begin{proof}[Theorem~\ref{thm:time-complexity}]
When we set the parameters $(\varepsilon_i^\prime,\delta_i)$ for \OPIMC \cite{Tang_OPIM_2018}, the expected time complexity of \OPIMC is 
\begin{align*}
	&O\big(\frac{m_i+n_i}{{\varepsilon_i^\prime}^2}(b\log n_i+\log\frac{1}{\delta_i})\big) \\
	&=O\big(\frac{m_i+n_i}{\varepsilon_i^2}(b\log n_i+\log\frac{1}{\varepsilon_i})\big).
\end{align*}
On the other hand, for any given $\R_1$ and $\R_2$, if \OPIMC stops, then \EP must also stop. This implies that \EP always finishes earlier than \OPIMC, which completes the proof.\qed
\end{proof}

\subsection{\AG Instantiated with \EP}\label{sec:exp-wor-appro}

In the following, we derive the approximation guarantees and time complexity of \AG instantiated using \EP.

Theorem~\ref{thm:batch_guarantee} indicates that \AG instantiated using \EP with parameter $\varepsilon_i$ achieves an expected approximation guarantee of at least $\rho_b(1-\varepsilon_i)$ in the $i$-th batch of seed selection. Immediately following by Theorem~\ref{thm:frame-greedy-approx} and Theorem~\ref{thm:time-complexity}, we have the following theorem.
\begin{theorem}\label{thm:expected-approx}
Suppose that we instantiate \AG using \EP with parameter $\varepsilon_i$ for the $i$-th batch of seed selection, then \AG achieves the expected approximation ratio of $1-\e^{\rho_b(\varepsilon-1)}$ where $\varepsilon=\frac{1}{r}\sum_{i=1}^{r}\varepsilon_i$, and takes an expected time complexity of $O(\sum_{i=1}^r(b\log n+\log\frac{1}{\varepsilon_i})(m+n)/\varepsilon_i^2)$.
\end{theorem}

To achieve the expected approximation ratio of $1-\e^{\rho_b(\varepsilon-1)}$, instantiating \AG using \EP takes shorter running time compared with that of using the naive expected approximation guarantee of the existing IM algorithms. As discussed in Section~\ref{sec:batch-epic}, the intuition behind is that \EP avoids the additional estimation error on $\OPT_b(G_i)$ which is considered by all the existing IM algorithms.

In addition, Theorem~\ref{thm:frame-greedy-approx} indicates that \AG instantiated using \EP with parameter $\varepsilon_i$ achieves the worst-case approximation ratio of $1-\e^{\rho_b(\varepsilon^\prime-1)}$ with a probability of at least $1-\delta$, where $\varepsilon^\prime=\frac{1}{r} \sum_{i=1}^r \varepsilon_i + \sqrt{1/(2r)\cdot\ln (1/\delta)}$. Therefore, to achieve a predefined worst-case approximation ratio of $1-\e^{\rho_b(\varepsilon-1)}$ with a probability of at least $1-\delta$, we may decrease the parameter $\varepsilon_i$ in \EP by an additive factor of $\sqrt{1/(2r)\cdot\ln (1/\delta)}$ for every $i$.
\begin{theorem}\label{thm:worst-approx}
	Suppose that we instantiate \AG using \EP with the parameters $\varepsilon^\prime_i=\varepsilon_i-\sqrt{1/(2r)\cdot\ln (1/\delta)}$ in each batch where $\delta\in (0,1)$. Then, \AG achieves the approximation ratio $1-\e^{\rho_b(\varepsilon-1)}$ with a probability of at least $1-\delta$ where $\varepsilon=\frac{1}{r}\sum_{i=1}^{r}\varepsilon_i$, and takes an expected time complexity of $O(\sum_{i=1}^r(b\log n+\log\frac{1}{\varepsilon_i^\prime})(m+n)/{\varepsilon_i^\prime}^2)$.
\end{theorem}

Note that Theorem~\ref{thm:worst-approx} requires that $\varepsilon^\prime_i=\varepsilon_i-\sqrt{1/(2r)\cdot\ln (1/\delta)}>0$. This implies that only when the number of batches $r$ is sufficiently large, i.e., $r>\frac{\ln (1/\delta)}{2\varepsilon_i^2}$, there is a valid instantiation of \AG to achieve a predefined worst-case approximation guarantee of $1-\e^{\rho_b(\varepsilon-1)}$ with probability at least $1-\delta$.

%% file: 5-revisit.tex
\section{Misclaims in Previous Work~\cite{Sun_MRIM_2018,Han_AIM_2018}}\label{sec:revisited}

In this section, we revisit two of the latest work proposed to address the adaptive IM problem, \ie our preliminary work~\cite{Han_AIM_2018} and Sun~\etal's work~\cite{Sun_MRIM_2018}. We aim to discuss potential issues and clarify some common misunderstandings towards this problem. Specifically, their algorithms are claimed to return a worst-case approximation guarantee with high probability. However, there exist potential theoretical issues in the analysis of the failure probability, which is elaborated as follows.

In \cite{Han_AIM_2018} (Section~4.1), it is claimed that the overall failure probability of the $i$-th batch satisfies \[ \Pr\Big[\varepsilon_i < \sum_{G_1, \ldots, G_i} ( \zeta_i \cdot \Pr[\zeta_i, G_1, \ldots, G_{i-1}] )\Big] \le \delta_i. \] Then, the failure probability of all $r$ batches is bounded by a union bound of $\sum_{i=1}^{r}\delta_i$.

Similarly, in~\cite{Sun_MRIM_2018} (Theorem 5.5) it is claimed that the proposed algorithm \adaimm achieves the worst-case approximation with $1-\frac{1}{n^l}$ probability where $l$ is a constant. They first prove that the seed set $S$ selected by \adaimm returns an approximation with at least $1-1/(n^l\cdot r)$ probability for each batch. Sun~\etal~\cite{Sun_MRIM_2018} thus claim that \adaimm achieves the approximation ratio with at least $1-1/n^l$ probability by union bound.

The theoretical guarantees of these two papers are based on Theorem A.10 in \cite{Golovin_adaptive_2011}. Through a careful examination of the proof of Theorem A.10 in \cite{Golovin_adaptive_2011}, we find that the essence is to bound the overall approximation guarantee for each batch, \ie~$X_i(\omega)\geq \rho_b(1-\varepsilon_i)$, where $X_i(\omega)$ represents the overall random approximation for the $i$-th batch of seed selection over all realizations, \ie
\begin{equation*}
X_i(\omega):=\frac{\E_{G_{i}}\big[\rho_b(1-\zeta(\omega,G_i))\cdot \OPT_b(G_i)\big]}{\E_{G_{i}}\big[\OPT_b(G_i)\big]},
\end{equation*}
Since $X_i(\omega)$ is a random variable that is likely to be smaller than $\rho_b(1-\varepsilon_i)$, these two papers \cite{Han_AIM_2018,Sun_MRIM_2018} attempt to bound the probability of $X_i(\omega)<\rho_b(1-\varepsilon_i)$ as
\begin{equation*}
\Pr[X_i(\omega)<\rho_b(1-\varepsilon_i)]\leq \delta_i.
\end{equation*}
However, as long as there exists one realization such that the seed set $S_i$ returned in the $i$-th batch does not meet the approximation of $\rho_b(1-\varepsilon_i)$, \ie~$\zeta(\omega,G_i)>\varepsilon_i$, it is possible that $X_i(\omega)< \rho_b(1-\varepsilon_i)$. On the other hand, there are exponential number $O(2^m)$ of realizations, where $m$ is the number of edges in $G$. Thus, although it holds that $\Pr[\zeta(\omega,G_i)>\varepsilon_i]\leq \delta_i$ under a given $G_i$, the probability of $X_i(\omega)<\rho_b(1-\varepsilon_i)$ can be as large as $O(2^m\cdot \delta_i)$ by the union bound. Therefore, it is intricate to bound $X_i(\omega)$, which indicates that their claims on failure probability do not hold. In other words, Theorem 4 in \cite{Han_AIM_2018} and Theorem 5.5 in~\cite{Sun_MRIM_2018} are invalid.

In this paper, we rectify the theoretical analysis of the worst-case approximation guarantee utilizing Azuma-Hoeffding inequality \cite{Mitzenmacher_Martingales_2005}. In particular, instead of bounding the probability of each individual $X_i(\omega)<\rho_b(1-\varepsilon_i)$, we directly bound the probability of $X(\omega)<\rho_b(1-\varepsilon)$, where $X(\omega)=\frac{1}{r}\sum_{i=1}^{r}X_i(\omega)$ and $\varepsilon=\frac{1}{r}\sum_{i=1}^{r}\varepsilon_i$, as $X(\omega)$ should be concentrated to its expectation $\EO[X(\omega)]$ when $r$ is sizable and $\EO[X_i(\omega)]\geq \rho_b(1-\varepsilon_i)$ can be achieved by various non-adaptive IM algorithms. 

%% file: 6-related.tex
\section{Related Work} \label{sec:relatedwork}

\subsection{Comparison with Preliminary Version} \label{sec:newcontri}

Compared with our preliminary work~\cite{Han_AIM_2018}, the current paper includes two major new contributions as follows.

First, we propose a randomized greedy policy that can provide strong theoretical guarantees for the general adaptive stochastic maximization problem, which may be of independent interest. This new solution can be adopted in many other settings apart from adaptive IM, e.g., active learning~\cite{Cuong_active_2013}, active inspection~\cite{Hollinger_active_2013}, optimal information gathering \cite{Chen_sequential_2015}, which are special cases of adaptive stochastic maximization. In particular, our proposed policy imposes far few constraints than Golovin and Krause's policy \cite{Golovin_adaptive_2011}. In fact, in some applications (e.g., adaptive IM), the requirement of Golovin and Krause's policy \cite{Golovin_adaptive_2011} is too stringent to construct such a policy whereas our proposed policy is easy to obtain. Moreover, we show that our policy can achieve a worst-case approximation guarantee with high probability, which uncovers some potential gaps in two recent studies \cite{Han_AIM_2018,Sun_MRIM_2018} and shed light on the future work of the adaptive IM problem.

Second, we improve the efficiency of algorithm \EP (Section~\ref{sec:batch}). In~\cite{Han_AIM_2018}, \EP is designed based on an idea similar to that of the \ssa algorithm in~\cite{Nguyen_DSSA_2016}. However, \ssa is rather inefficient when the input error parameter $\varepsilon$ is small, as verified in~\cite{Tang_OPIM_2018}. Therefore, we redesigned \EP based on the state-of-the-art method \OPIMC~\cite{Tang_OPIM_2018}, which is far more efficient than \ssa. Moreover, we optimize the estimation of the upper bound of $\OPT$ in \EP based on {\it martingale stopping theorem} \cite{Mitzenmacher_Martingales_2005}, which boosts the performance of \AG noticeably.

\subsection{Non-Adaptive Influence Maximization}\label{sec: related-NIM}

The IM problem under the non-adaptive setting has been extensively studied. The seminal work of Kempe~\etal~\cite{Kempe_maxInfluence_2003} shows that there is a $1-1/\e-\varepsilon$ approximation guarantee for the non-adaptive IM problem, and it proposes a monte carlo simulation algorithm to achieve this approximation ratio with high time complexity. After that, a lot of studies have appeared to improve Kempe~\etal's work in terms of time efficiency, especially for some applications \cite{LiQYM_IC_2017,LinLWWX_TOPK_2018} that require efficient algorithms to identify the top-$k$ influential set in large graphs. Among these works, Borgs~\etal~\cite{Borgs_RIS_2014} propose the RR-set sampling method for influence spread estimation, and several later studies~\cite{Tang_TIM_2014,Tang_IMM_2015,Nguyen_DSSA_2016,Tang_OPIM_2018} use this method to find more efficient algorithms for the IM problem. Moreover, the RR-set sampling method is extensively adopted in other variants of IM, e.g.,~profit maximization~\cite{Tang_profitMax_2016,Tang_profitMaxUS_2018,Tang_profitMax_2018} that optimizes a profit metric naturally combining the benefit and cost of influence spread. However, all these studies concentrate on the non-adaptive IM problem (or its variants), and hence their approximation guarantees do not hold for the adaptive IM problem.

\subsection{Adaptive Influence Maximization}\label{sec: related-AIM}

Compared with the studies on non-adaptive IM, the studies on adaptive IM are relatively few. Golovin~\etal~\cite{Golovin_adaptive_2011} derive a $(1-1/\e)$-approximation ratio under the case that only one seed node can be selected in each batch. The feedback model they consider is the same as the one described in this paper, which they call the {\em full-adoption feedback} model. In their arXiv version, they also mention another feedback model called {\em myopic feedback} model, where the feedback of a selected seed node only includes the directed neighbors activated by the seed, but does not include further activated nodes in the cascade process. They show that under the IC model full-adoption feedback is adaptive submodular but myopic feedback is not adaptive submodular. In addition, Yuan and Tang~\cite{YuanT17} propose a generalized feedback model, called {\em partial feedback} model, under which the objective is not adaptive submodular either. Chen~\etal~\cite{Chen_activeLearning_2013}, Tang~\etal~\cite{Tang_ASM_2019}, Huang~\etal~\cite{Huang_ATPM_2020}, Vaswani and Lakshmanan~\cite{Vaswani_adapIM_2016} study adaptive seed selection under the case that more than one seed nodes can be selected in each batch. Nevertheless, Chen~\etal~\cite{Chen_activeLearning_2013} and Tang~\etal~\cite{Tang_ASM_2019} aim to minimize the cost of the selected seeds under the constraint that the influence spread is larger than a given threshold while Huang~\etal~\cite{Huang_ATPM_2020} target at maximizing the profit (i.e., revenue of influence spread less the cost of seed selection), which are different goals from ours. Vaswani and Lakshmanan~\cite{Vaswani_adapIM_2016} derive an approximation guarantee ${1-\exp\left(-\frac{(1-1/\e)^2}{\eta}\right)}$ for certain $\eta>1$. Unfortunately, none of the studies listed above provide a practical algorithm to achieve the claimed approximation ratios. More specifically, Golovin~\etal~\cite{Golovin_adaptive_2011} and Chen~\etal~\cite{Chen_activeLearning_2013} assume that the expected influence spread can be exactly computed in polynomial time (which is not true due to \cite{Chen_MIA_2010}), while Vaswani and Lakshmanan~\cite{Vaswani_adapIM_2016} did not provide a method to bound the key parameter $\eta$ appearing in their approximation ratio. Sun~\etal~\cite{Sun_MRIM_2018} study the Multi-Round Influence Maximization (MRIM) problem under the multi-round triggering model, where influence propagates in multiple rounds independently from possibly different seed sets. In our adaptive IM problem, we consider a natural diffusion model that the realization of influence propagation is identical for all batches/rounds. Meanwhile, as we discussed, our analyses of approximation guarantees uncover some potential gaps in \cite{Sun_MRIM_2018}.

More recently, there are a few studies on the {\em adaptivity gap}, the	ratio between the optimal adaptive solution versus the optimal non-adaptive solution, in the context of adaptive influence maximization. Peng and Chen~\cite{Peng_adaptive_2019} show a constant adaptivity gap for adaptive influence maximization under the IC model with myopic feedback, and using this result to further show that the adaptive greedy algorithm achieves a constant approximation even though the model is not adaptive submodular. They also show in another paper~\cite{Chen_adaptivity_2019} the constant upper and lower bounds for the adaptivity gap in the IC model with full-adoption feedback for several special classes of graphs, but the adaptivity gap for the general graphs remains open. Chen et al.~\cite{Chen_greedyadaptive_2020} define the greedy adaptivity gap as the ratio between the adaptive greedy solution versus the non-adaptive greedy solution, and provide upper/lower bounds for the greedy adaptivity gap under certain influence propagation models. These studies on the adaptivity gap demonstrate the power and limitation of adaptivity in influence maximization and are complementary to our study on efficient algorithms for adaptive influence maximization. 

We also note that Seeman~\etal~\cite{Seeman_adapSeeding_2013}, Horel~\etal~\cite{Horel_adapSeeding_2015} and Badanidiyuru~\etal~\cite{Badanidiyuru_adapSeeding_2016} consider an influence maximization problem called ``adaptive seeding'', but with totally different implications from ours. More specifically, they assume that the seed nodes can be selected in two stages. In the first stage, a set $S$ can be selected from a given node set $S^\prime\subseteq V$. In the second stage, another seed set $S^+$ can be selected from the influenced neighboring nodes of $S$. The goal of their problem is to maximize the expected influence spread of $S^+$, under the constraint that the total number of nodes in $S\cup S^+$ is no more than $k$. However, both the problem model and the optimization goal of these studies are very different from ours, and hence their methods cannot be applied to our problem.

%% file: 7-exp.tex
\section{Performance Evaluation} \label{sec:exp}

In this section, we evaluate the performance of our proposed approach with extensive experiments. The goal of our experiments is to measure the efficiency and effectiveness of \AG using real social networks. All of our experiments are conducted on a Linux machine with an Intel Xeon 2.6GHz CPU and 64GB RAM. 

\subsection{Experimental Setting}\label{sec:exp-setting}

\spara{Datasets} We use five real datasets in our experiments, \ie~{NetHEPT}, {Epinions}, {DBLP}, {LiveJournal}, and {Orkut} as summarized by Table~\ref{tbl:dataset}. {NetHEPT} is obtained from~\cite{Chen_NewGreedy_2009}, representing the academic collaboration networks in ``High Energy Physics-Theory'' area. The rest four datasets are available from~\cite{Leskovec_SNAP_2014}. Among them, Orkut contains millions of nodes and edges. We randomly generate 20 realizations for each dataset, and then report the average performance for each algorithm on those 20 realizations.
\begin{table}[!t]
	\centering
	\caption{Dataset details. ($\mathrm{K}=10^3, \mathrm{M}=10^6$)} \label{tbl:dataset}
	\vspace{2mm}
	\setlength{\tabcolsep}{0.5em} 
	\renewcommand{\arraystretch}{1.2}
	\begin{tabular} {l|rrrc}
		\hline
		{\bf Dataset} & \multicolumn{1}{c}{$\boldsymbol{n}$} & \multicolumn{1}{c}{$\boldsymbol{m}$} & \multicolumn{1}{c}{\bf Type}  & {\bf Avg.\ deg}  \\ \hline
		{NetHEPT}       & 15.2K			&  31.4K		& 	undirected		 &	4.18       \\ 
		{Epinions}		 & 132K			&  841K			&  	directed		 &	13.4       \\ 
		{DBLP}			 & 655K			&  1.99M		&  	undirected		 &	6.08       \\ 
		{LiveJournal}   & 4.85M			&  69.0M		&  	directed		 &	28.5        \\ 
		{Orkut}		 	& 3.07M			&  117M			&  	undirected		 &	76.2      \\ \hline 			
	\end{tabular}
\end{table}

\begin{figure*}[!t]
	\centering
	\includegraphics[height=10pt]{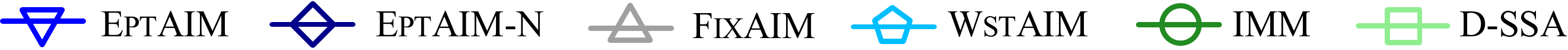}\vspace{-0.1in}\\
	\subfloat[NetHEPT]{\includegraphics[width=0.205\linewidth]{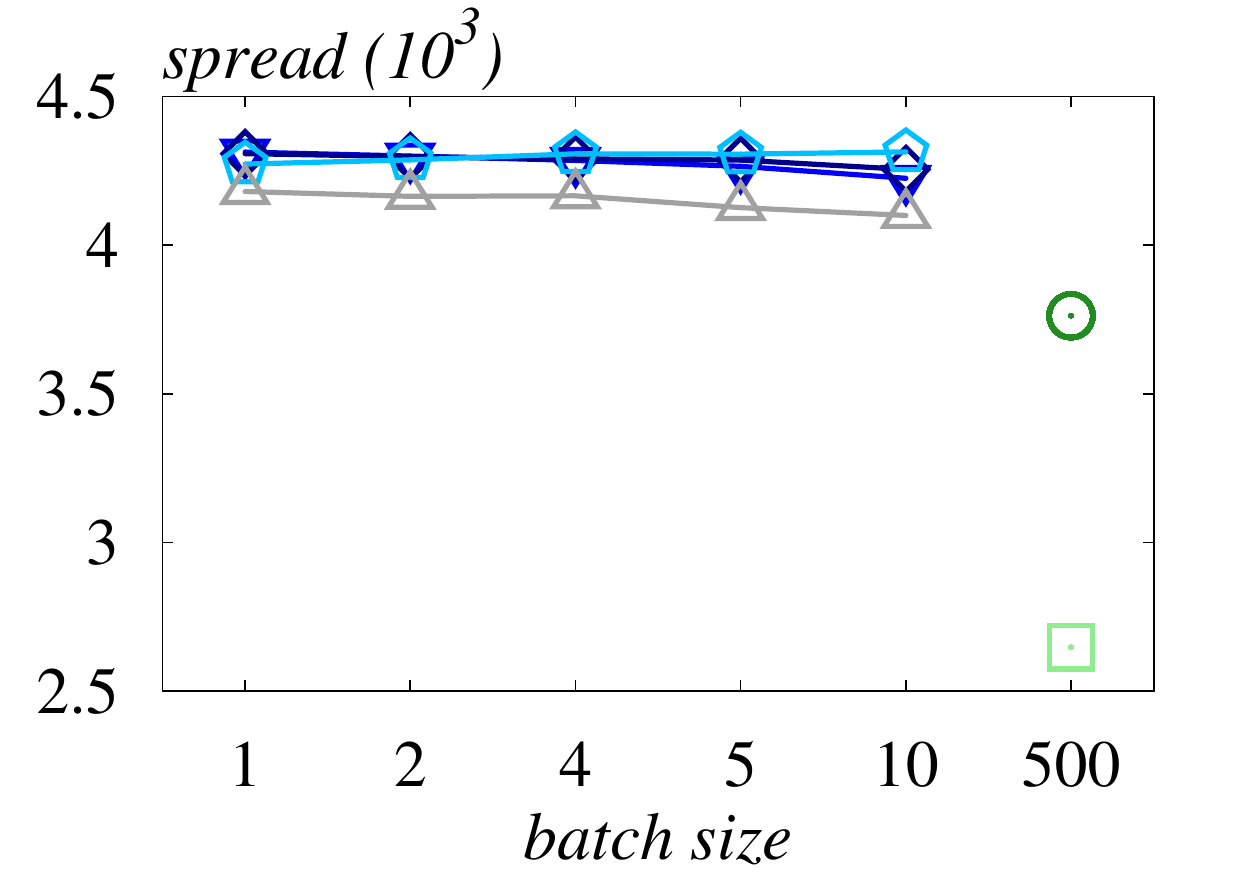}\label{fig:NetHEPT_batch_spread}}
	\subfloat[Epinions]{\includegraphics[width=0.205\linewidth]{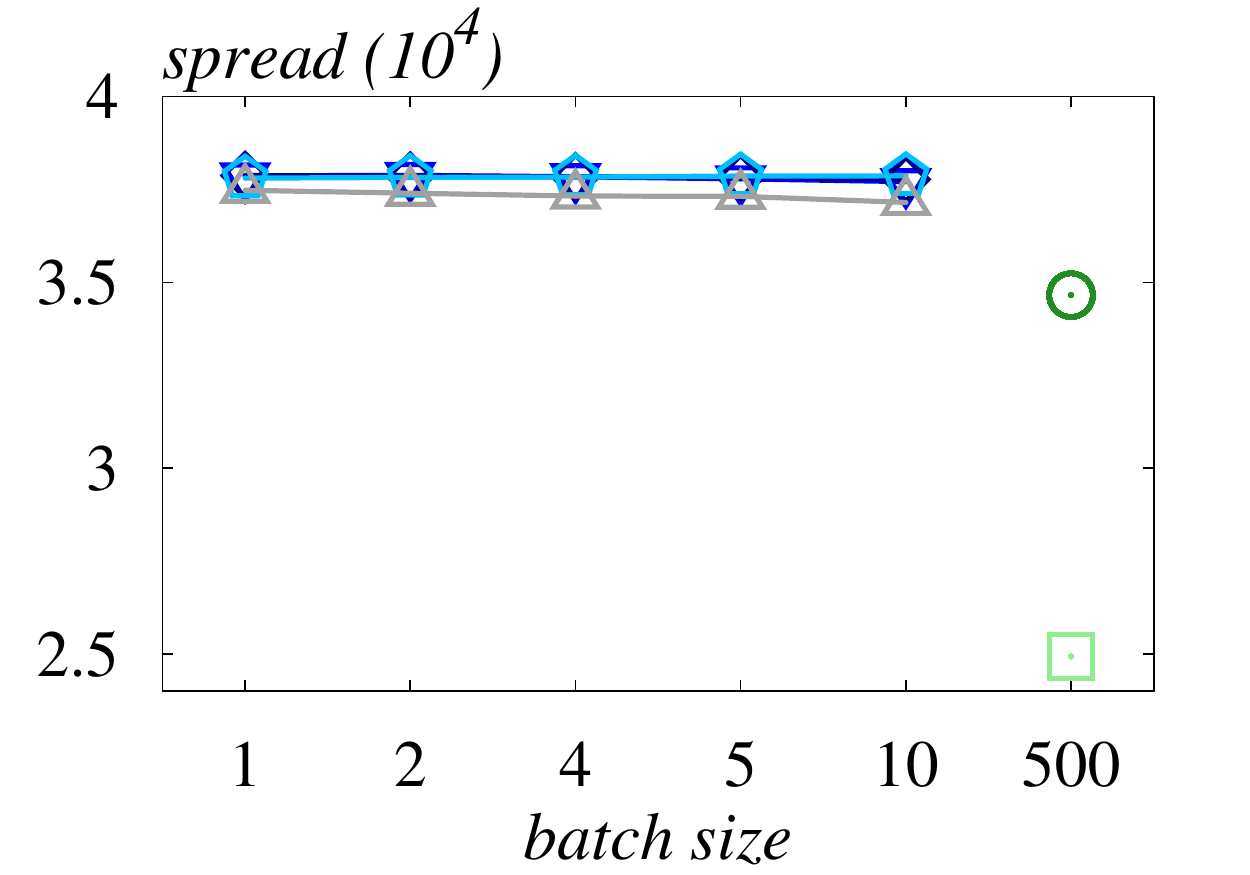}\label{fig:Epinions_batch_spread}}
	\subfloat[DBLP]{\includegraphics[width=0.205\linewidth]{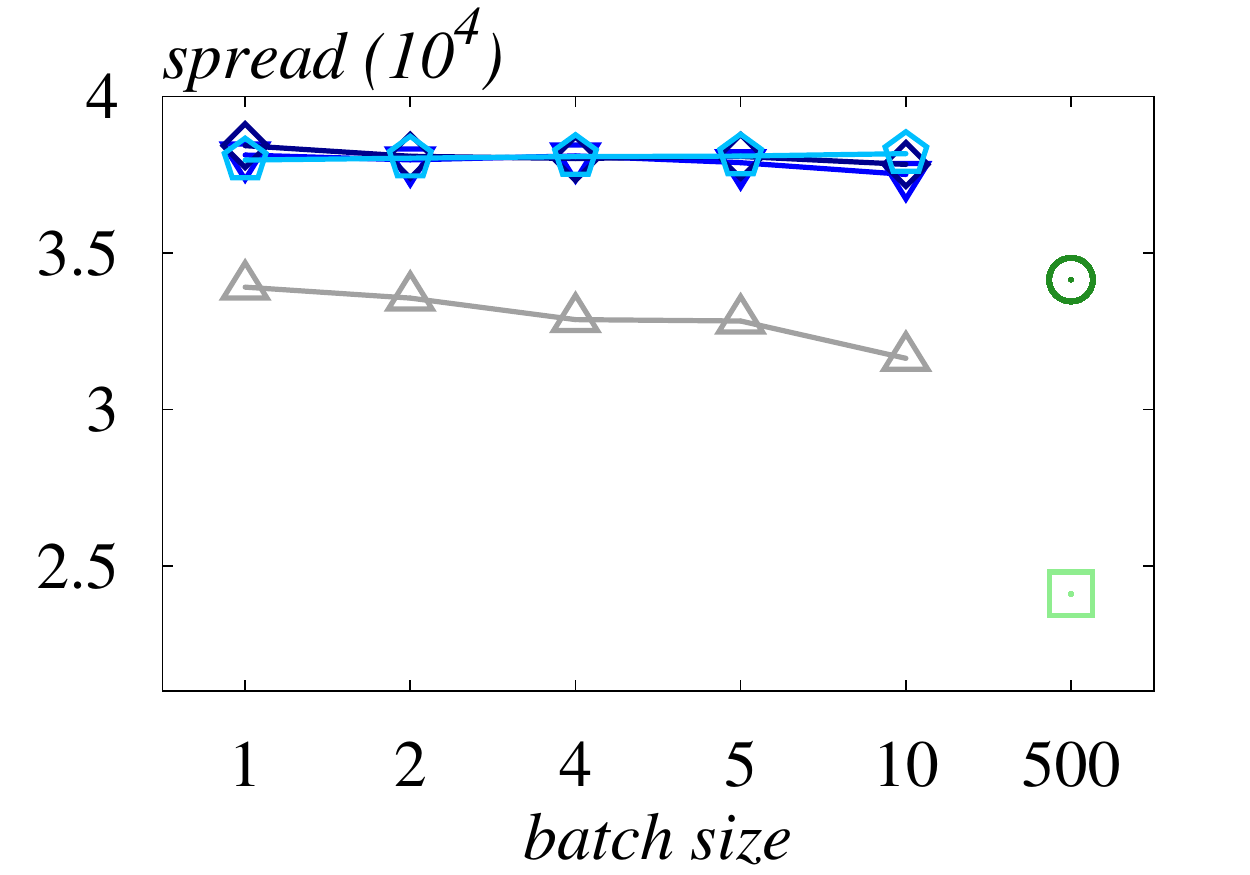}\label{fig:DBLP_batch_spread}}
	\subfloat[LiveJournal]{\includegraphics[width=0.205\linewidth]{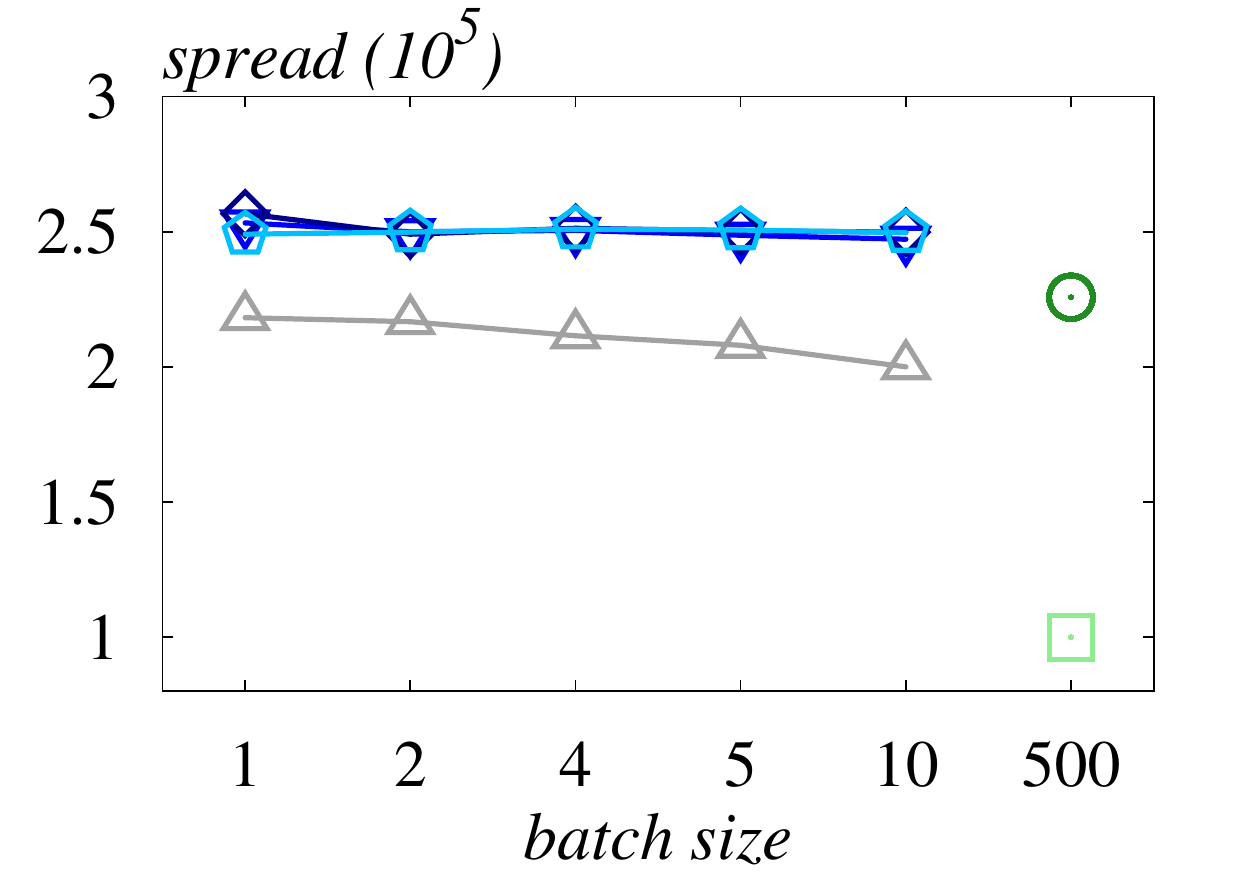}\label{fig:LiveJournal_batch_spread}}
	\subfloat[Orkut]{\includegraphics[width=0.205\linewidth]{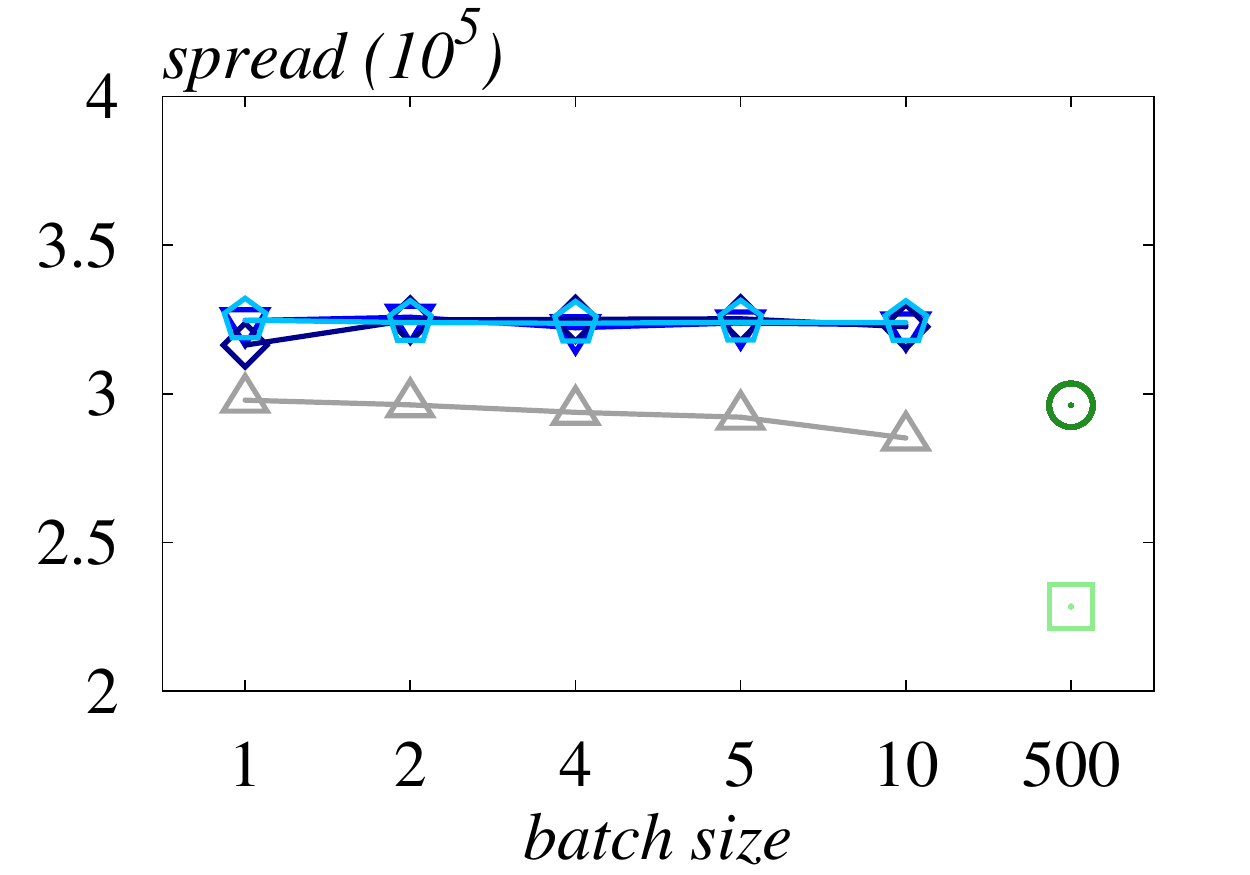}\label{fig:Orkut_batch_spread}}
	\caption{Spread vs. batch size.}\label{fig:spread-batch}
\end{figure*}

\begin{figure*}[!t]
	\centering
    \includegraphics[height=10pt]{legend}\vspace{-0.1in}\\
	\subfloat[NetHEPT]{\includegraphics[width=0.205\linewidth]{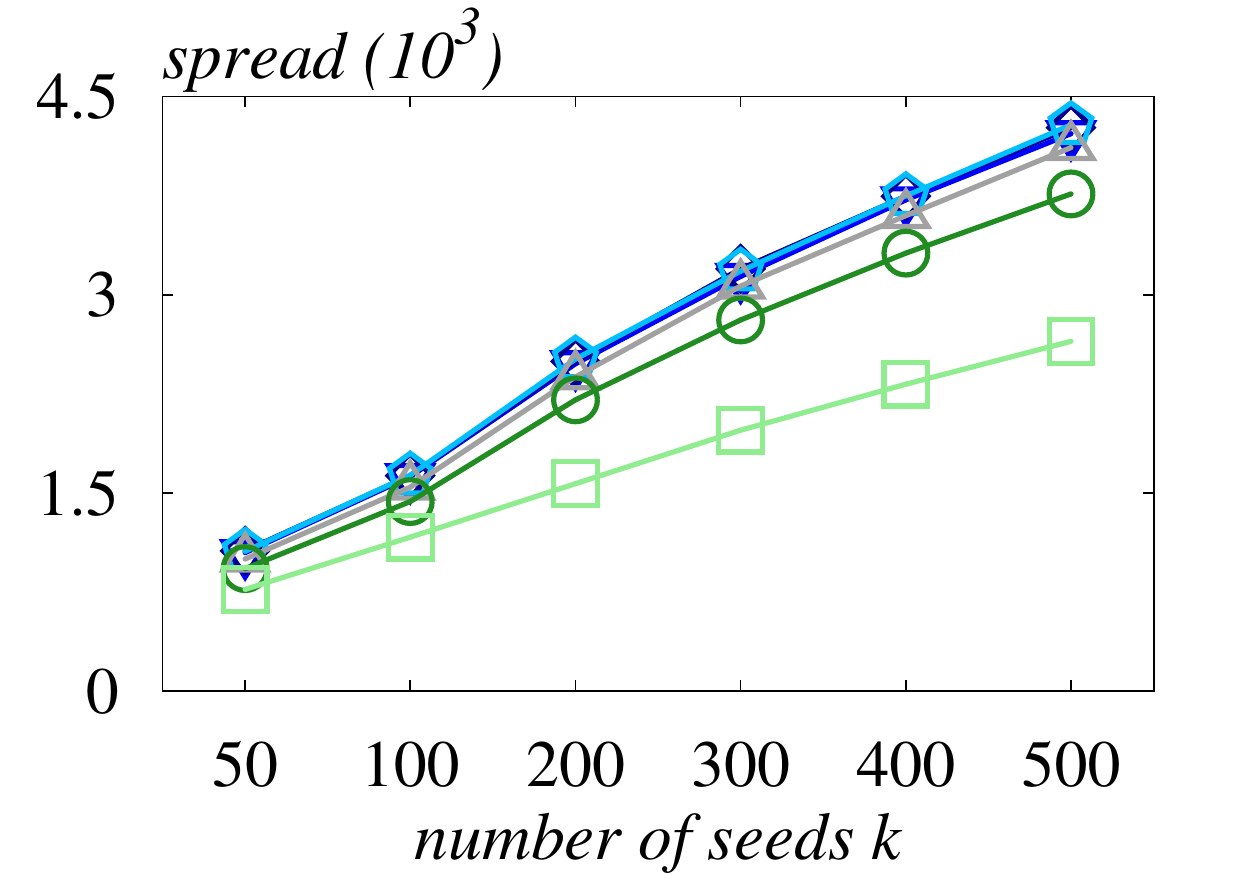}\label{fig:NetHEPT_k_spread}}
	\subfloat[Epinions]{\includegraphics[width=0.205\linewidth]{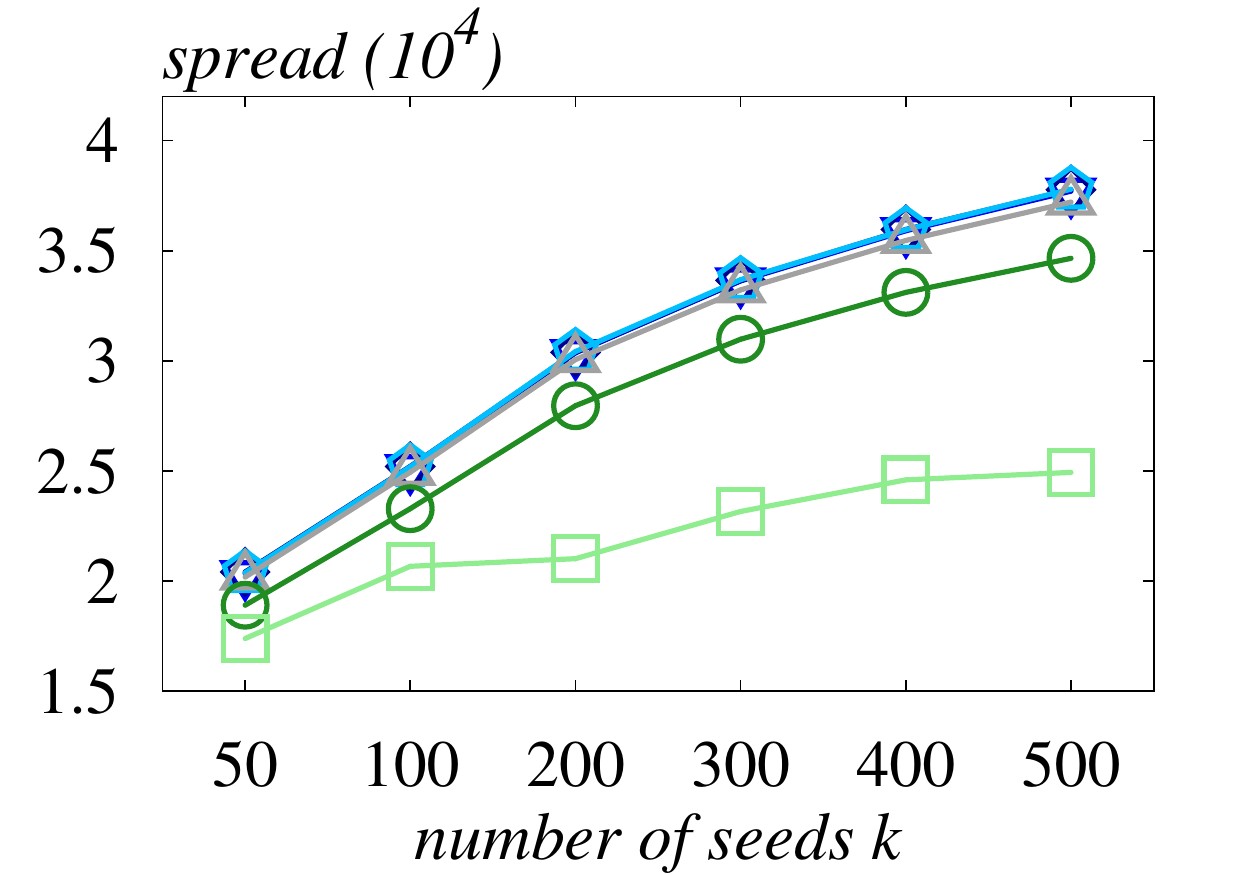}\label{fig:Epinions_k_spread}}
	\subfloat[DBLP]{\includegraphics[width=0.205\linewidth]{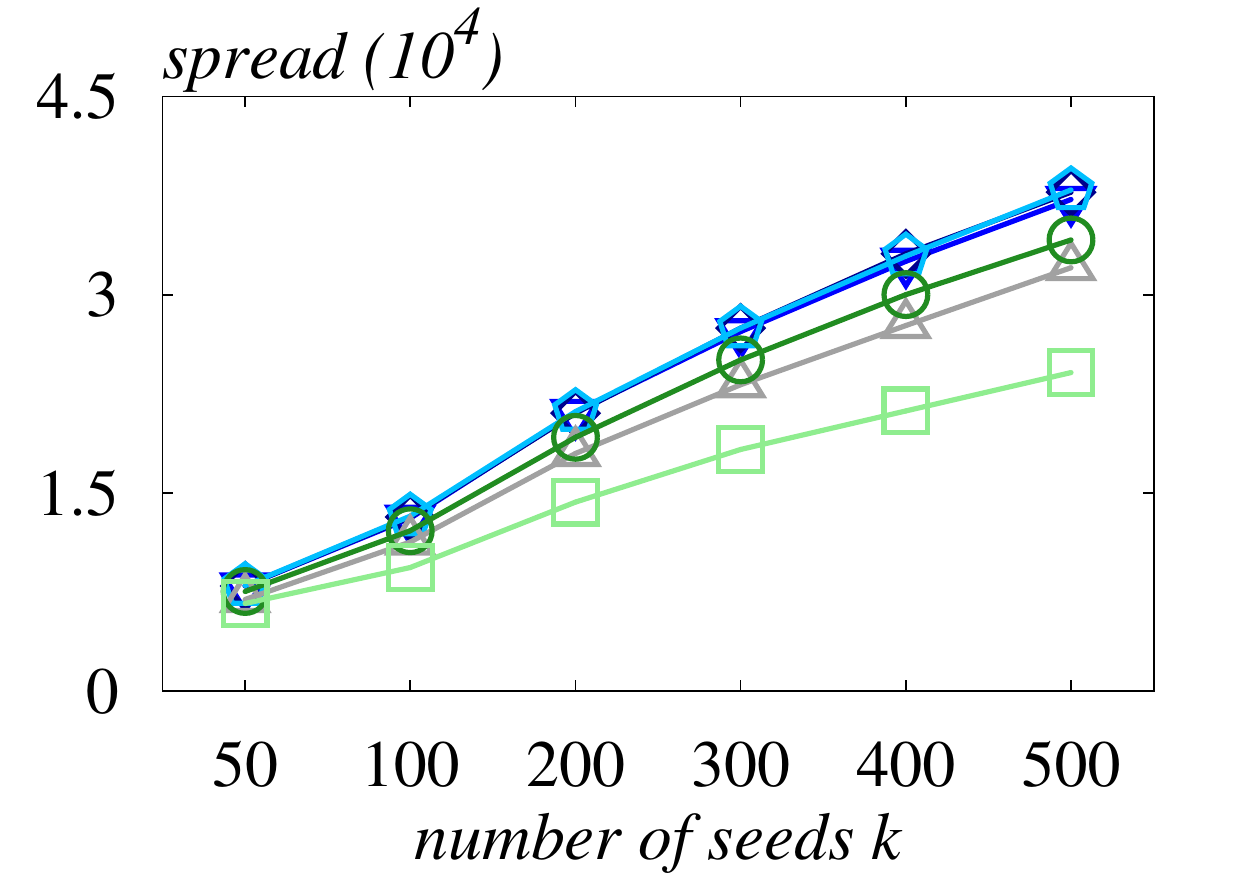}\label{fig:DBLP_k_spread}}
	\subfloat[LiveJournal]{\includegraphics[width=0.205\linewidth]{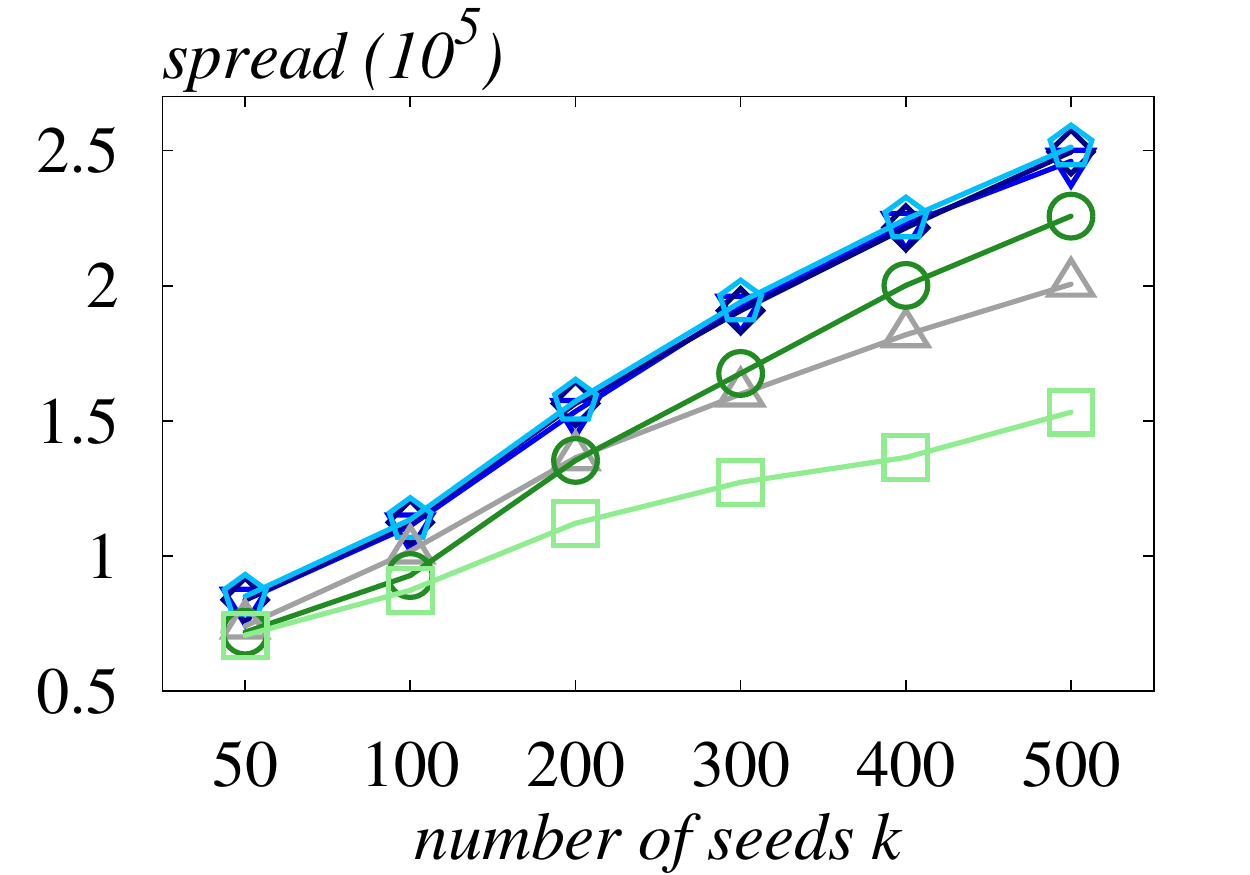}\label{fig:LiveJournal_k_spread}}
	\subfloat[Orkut]{\includegraphics[width=0.205\linewidth]{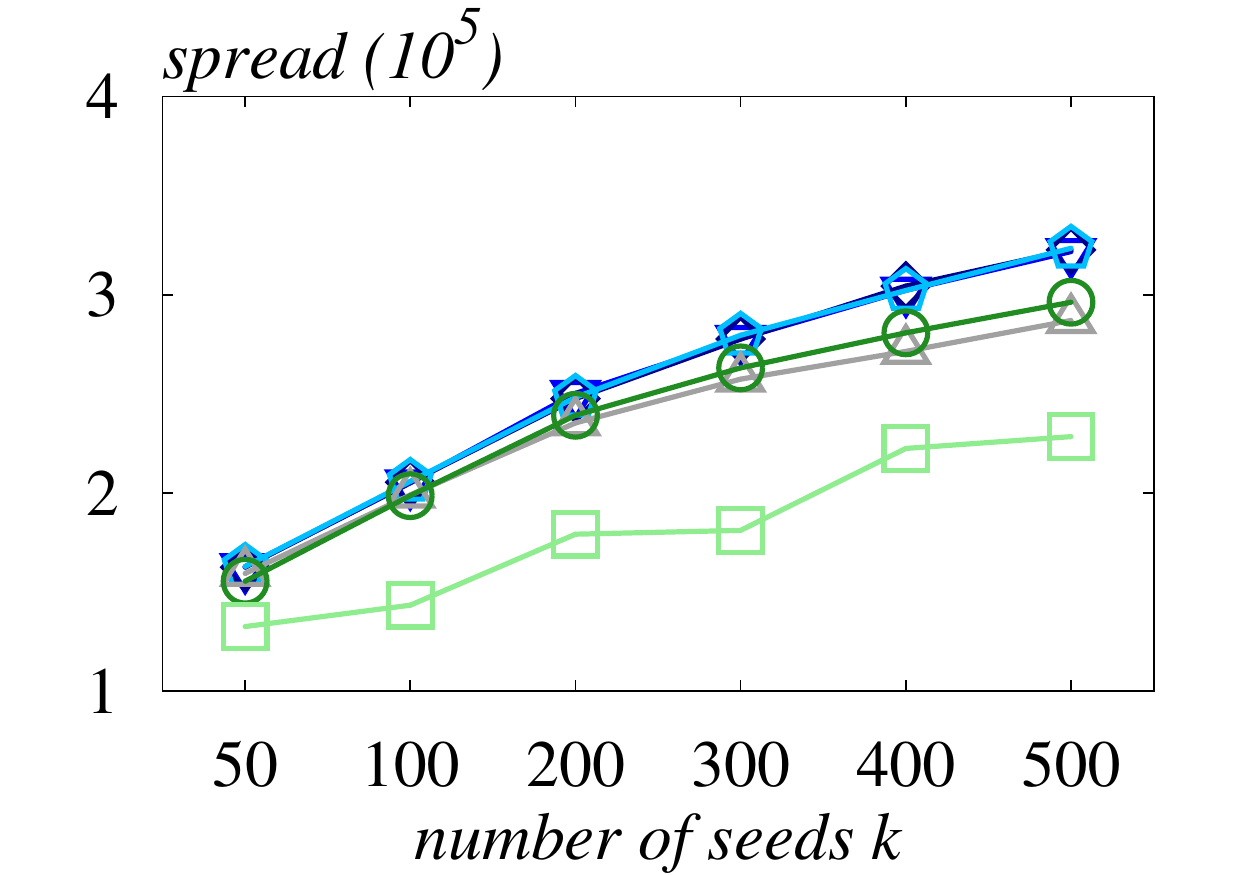}\label{fig:Orkut_k_spread}}
	\caption{Spread vs. seed size.}\label{fig:spread-seed}
\end{figure*}

\spara{Algorithms} We evaluate four adaptive algorithms, \ie \expepic, \worepic, \expopim, and \fixaim and two state-of-the-art non-adaptive algorithms, \ie \imm~\cite{Tang_IMM_2015} and \dssa~\cite{Nguyen_DSSA_2016}. \expepic is the algorithm we instantiate \AG with \EP to achieve an expected approximation ratio of $1-\e^{\rho_b(\varepsilon-1)}$ where $\rho_b=1-(1-1/b)^b$ and $b$ is the batch size. \worepic is the same implementation as \expepic but with well-calibrated parameters to acquire a worst-case approximation ratio of $1-\e^{\rho_b(\varepsilon-1)}$ with high probability. Recall that obtaining the worst-case approximation needs a more demanding requirement than the expected approximation, pointed out in Theorem~\ref{thm:worst-approx}. \expopim is a naive instantiation of \AG instantiated using the existing non-adaptive algorithm \OPIMC~\cite{Tang_OPIM_2018} directly, as introduced in Section~\ref{sec:batch-basic}. In addition, \expopim fixes the issues of \adapone \cite{Han_AIM_2018} so that it provides the correct expected approximation ratio of $1-\e^{\rho_b(\varepsilon-1)}$. By including \expopim, we could evaluate the performance improvement of \expepic against \expopim. \fixaim is a variant of \expepic that uses a fixed number of samples for each batch of seed selection. Note that \fixaim is a heuristic algorithm, which does not provide any theoretical guarantees. The purpose of \fixaim is to provide some insights on the effect of sample size on the performance of adaptive algorithms.

We also test two state-of-the-art non-adaptive IM algorithms (i.e., \imm~\cite{Tang_IMM_2015} and \dssa~\cite{Nguyen_DSSA_2016}) in our experiments. The purpose of using \dssa and \imm in our experiments is to measure the influence spread increase achieved by the adaptive IM algorithms compared with the non-adaptive IM algorithms.

\spara{Parameter Settings} We use the popular independent cascade (IC) model~\cite{Kempe_maxInfluence_2003} in our experiments. Following a large body of existing work on influence maximization~\cite{Tang_TIM_2014,Tang_IMM_2015,Nguyen_DSSA_2016,Tang_OPIM_2018,Kempe_maxInfluence_2003}, we set the propagation probability of each edge $(u,v)$ to $\frac{1}{d_{\mathrm{in}}(v)}$, where $d_{\mathrm{in}}(v)$ is the in-degree of node $v$.

We set $\varepsilon=0.5$ for the three adaptive algorithms and two non-adaptive algorithms for fair comparison and approximation errors $\varepsilon_1=\cdots=\varepsilon_r=\varepsilon$ for the three adaptive algorithms. Meanwhile, we set the failure probability of $\delta=1/n$ for \worepic, \imm, and \dssa. For \fixaim, we generate $10\mathrm{K}$ RR-sets for each batch of seed selection.

Recall that we need to select $k$ nodes in $r$ batches in adaptive IM, where $b=k/r$ nodes are selected in each batch. To see how the performance of our algorithms is affected by input parameters $k, b$ and $r$, we set these parameters according to the {\em $b$-setting} and {\em $k$-setting} explained as follows. Under the {\em $b$-setting}, we fix $k=500$ and vary $b$ such that $b \in \{1,2,4,5,10,500\}$. Under the {\em $k$-setting}, we fix $r=50$ and vary $k$ such that $k \in \{ 50, 100, 200, \cdots, 500\}$.

\subsection{Comparison of Influence Spread}\label{sec:influence}

In this section, we study the influence spread for all tested algorithms, as shown in \figurename~\ref{fig:spread-batch} and \figurename~\ref{fig:spread-seed}. In order to gain a comprehensive understanding about the efficacy of the tested algorithms, we measure their influence spreads achieved by varying the number of seed nodes $k$ and the batch size $b$. 

\figurename~\ref{fig:spread-batch} reports the influence spread obtained with $500$ seed nodes selected through different numbers of batches on the four datasets. In general, the spreads acquired by \expepic, \worepic, and  \expopim are comparable to each other but notably larger than the spreads of the baselines, including the heuristic adaptive algorithm, \ie~\fixaim, and two non-adaptive algorithms, \ie \imm and \dssa. In particular, \worepic, \imm and \dssa achieve the worst-case approximation guarantee, while \worepic obtains around $12\%$ and $60\%$ more spread than \imm and \dssa do in average, respectively. On the one hand, this can be explained by the advantage of adaptivity over non-adaptivity that adaptive algorithms could make smarter decisions based on the feedback from previous batches. On the other hand, the considerable discrepancy on the spread of \dssa exposes that \dssa sacrifices its effectiveness badly for the sake of high efficiency (referring to its running time, as shown in \figurename~\ref{fig:time-batch} and \figurename~\ref{fig:time-seed}). As with \fixaim, it achieves the smallest spreads among the four adaptive algorithms, with around $10\%$ less in average. In particular, \fixaim obtains even smaller spreads than the non-adaptive algorithm \imm on the three largest datasets, as shown in \figurename~\ref{fig:spread-batch}. This fact suggests that $10\mathrm{K}$ samples are insufficient to provide good performance.

\figurename~\ref{fig:spread-seed} shows the results of influence spreads with various number of seed nodes. We observe that (i) the spreads grow with the number of seed nodes $k$ as expected, (ii) our three adaptive algorithms achieve similar amount of spreads and outperform the heuristic adaptive algorithm \fixaim under the same $k$ and $b$ setting, which is consistent with the results in \figurename~\ref{fig:spread-batch}, and (iii) the percentage increase of spreads obtained by the adaptive algorithms over the spreads of \imm is around $10\%$ in average. This spread improvement is quite promising considering the large number of users in social networks. Meanwhile, it further confirms the superiority of adaptive algorithms on influence maximization.

\begin{figure}[!t]
	\centering
	\subfloat[Spread vs. batch size]{\includegraphics[width=0.53\linewidth]{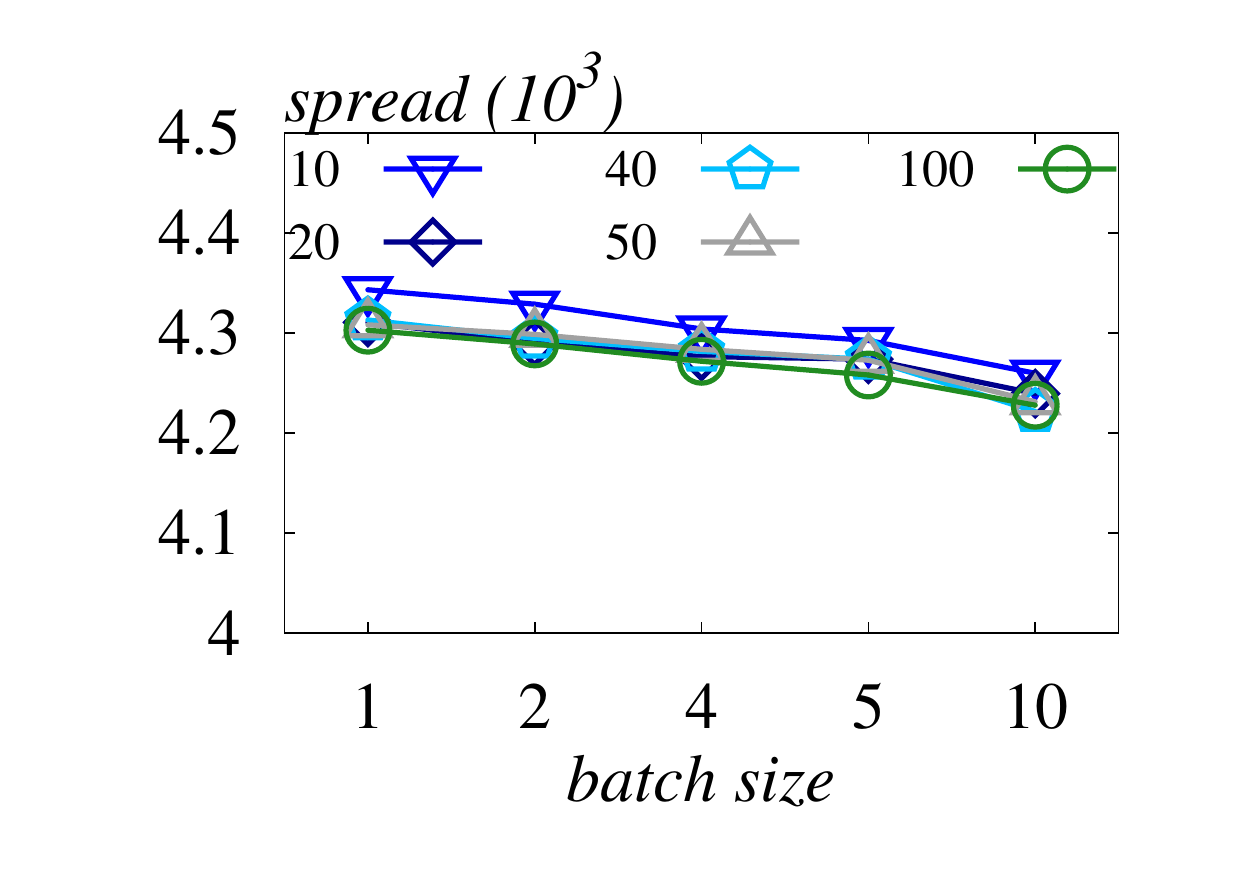}\label{fig:exp_b_spread}}
	\subfloat[Spread vs. seed size]{\includegraphics[width=0.5\linewidth]{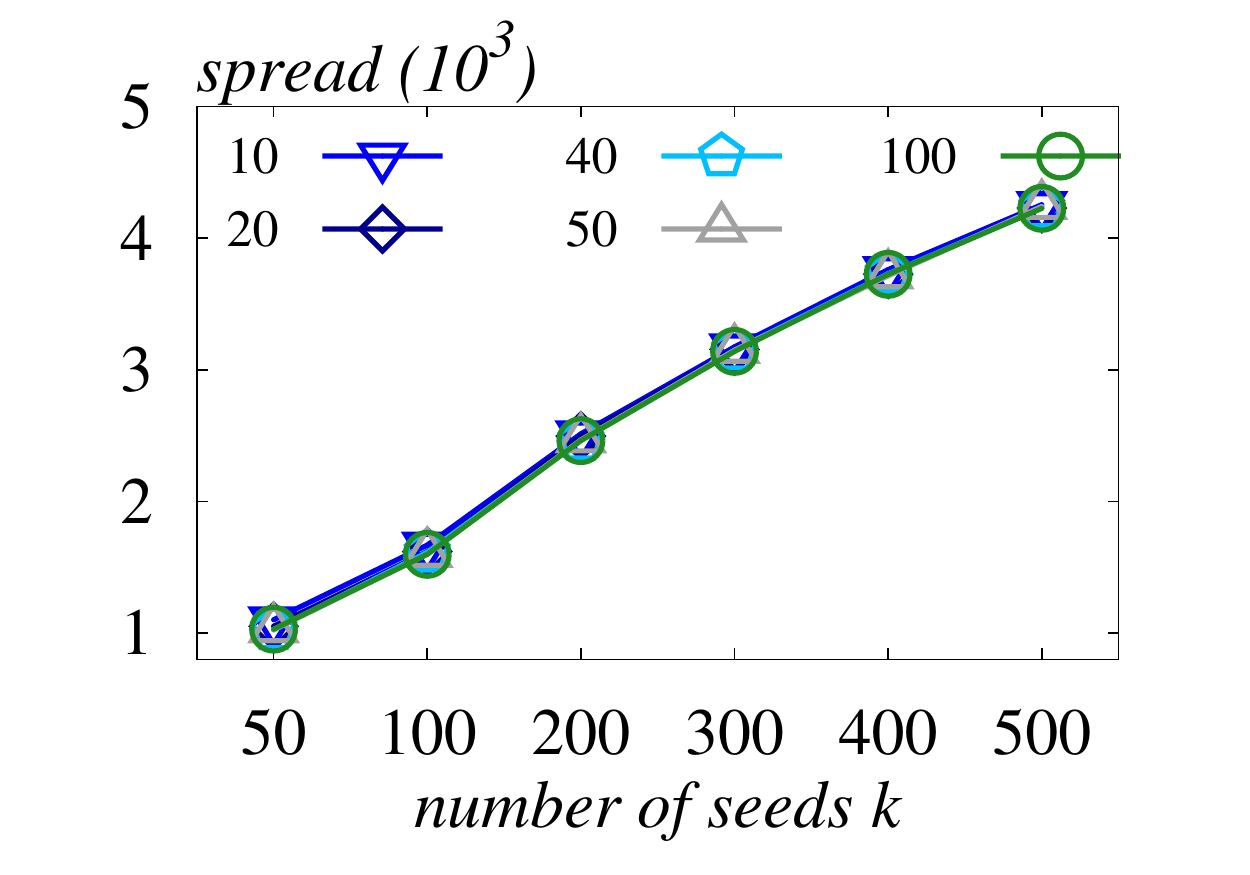}\label{fig:exp_k_spread}}
	\caption{Spread in different realization number on NetHEPT.}\label{fig:spread-realization}
\end{figure}
\figurename~\ref{fig:spread-realization} reports the average spread of our proposed \expepic algorithm under different number of realizations, including $\{10, 20, 40, 50, 100\}$, on the {NetHEPT} dataset. As shown, the average spreads of \expepic in various numbers of realizations are well-converged, especially under the {\em $k$-setting}. These results support the reliability of the results obtained through $20$ random realizations.

\begin{figure*}[!t]
	\centering
	\includegraphics[height=10pt]{legend}\vspace{-0.1in}\\
	\subfloat[NetHEPT]{\includegraphics[width=0.205\linewidth]{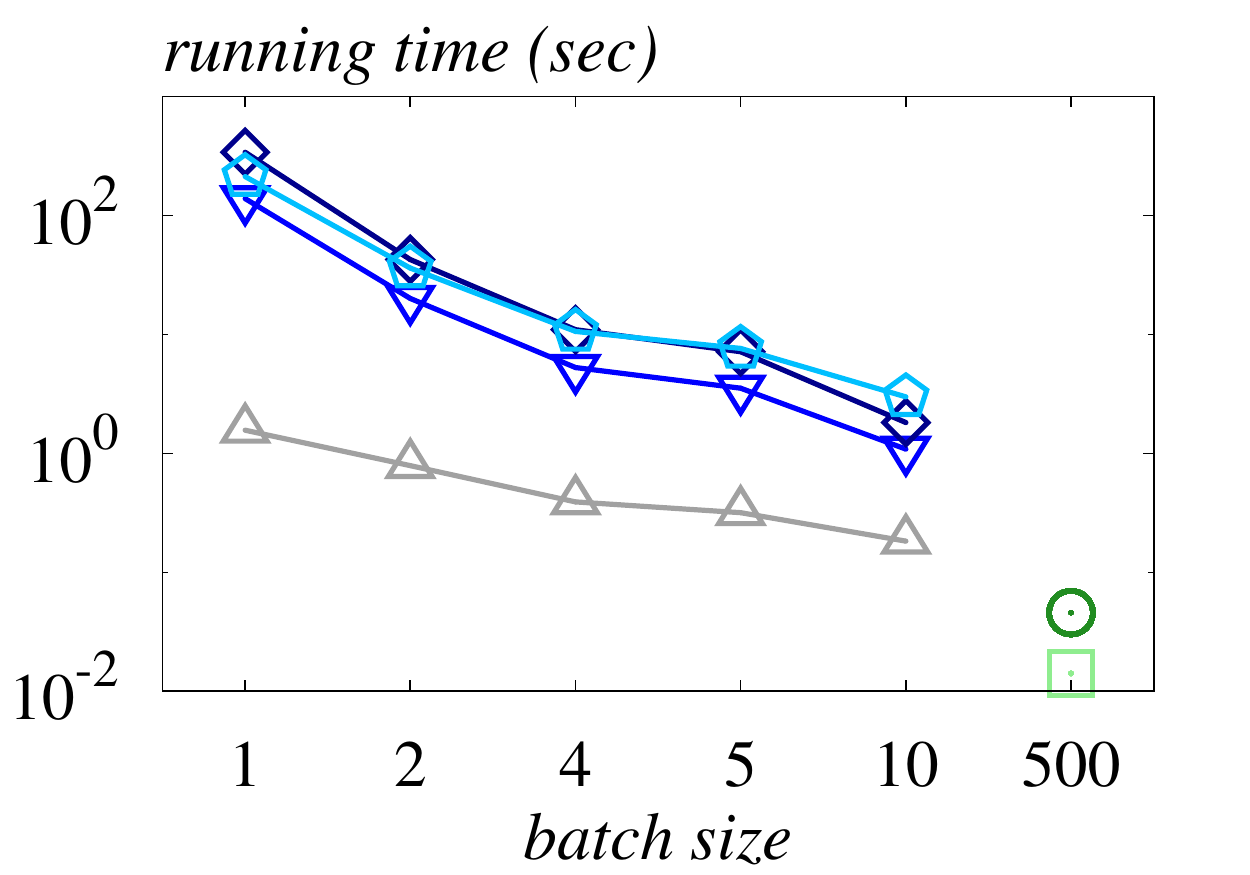}\label{fig:NetHEPT_batch_time}}
	\subfloat[Epinions]{\includegraphics[width=0.205\linewidth]{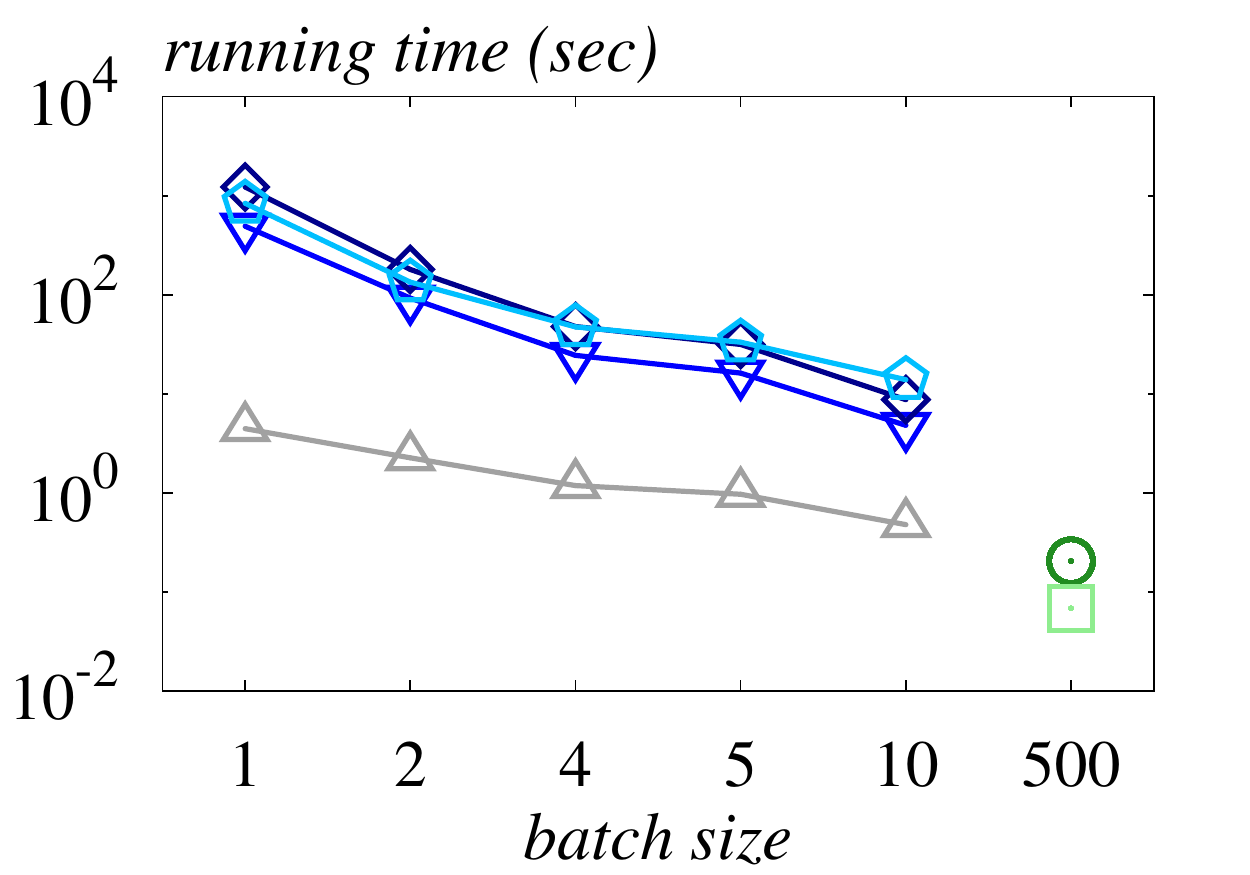}\label{fig:Epinions_batch_time}}
	\subfloat[DBLP]{\includegraphics[width=0.205\linewidth]{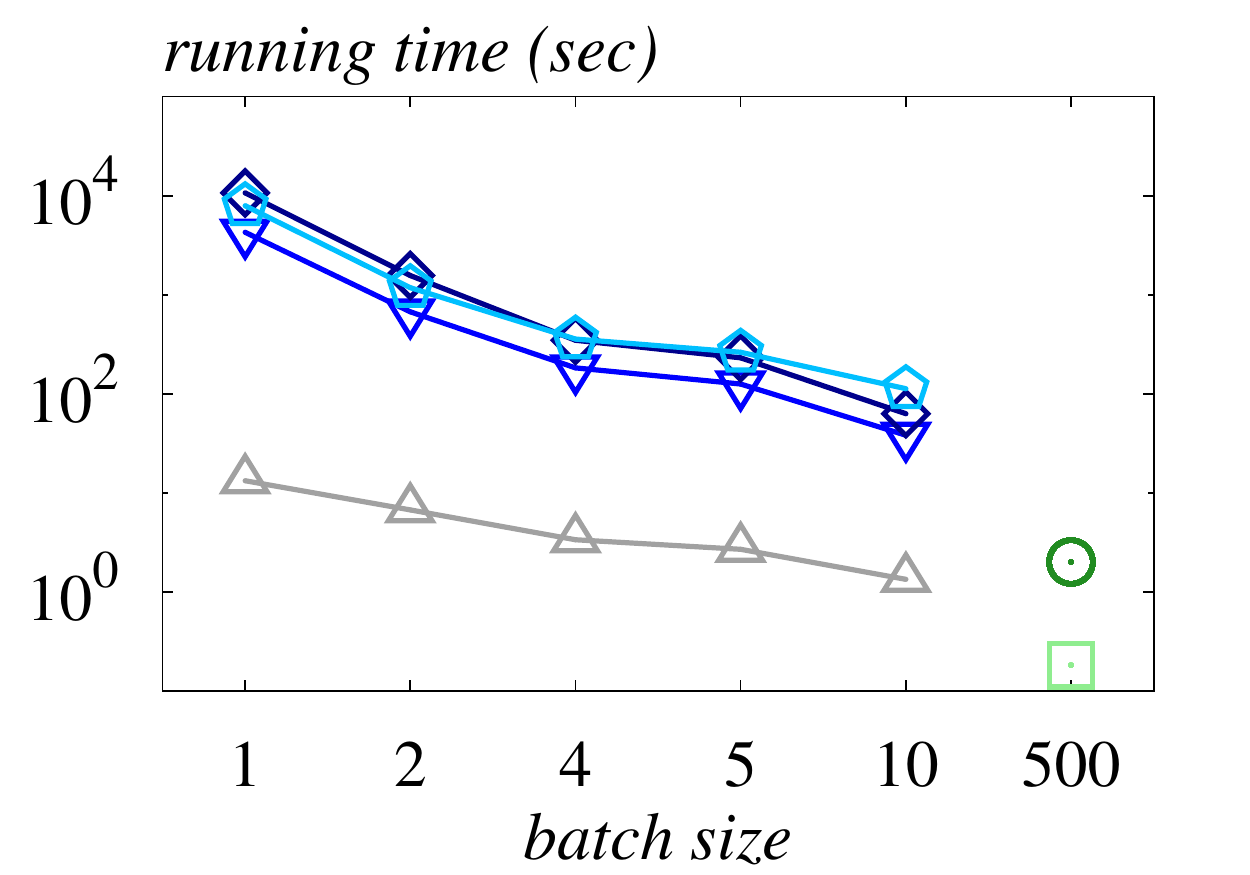}\label{fig:DBLP_batch_time}}
	\subfloat[LiveJournal]{\includegraphics[width=0.205\linewidth]{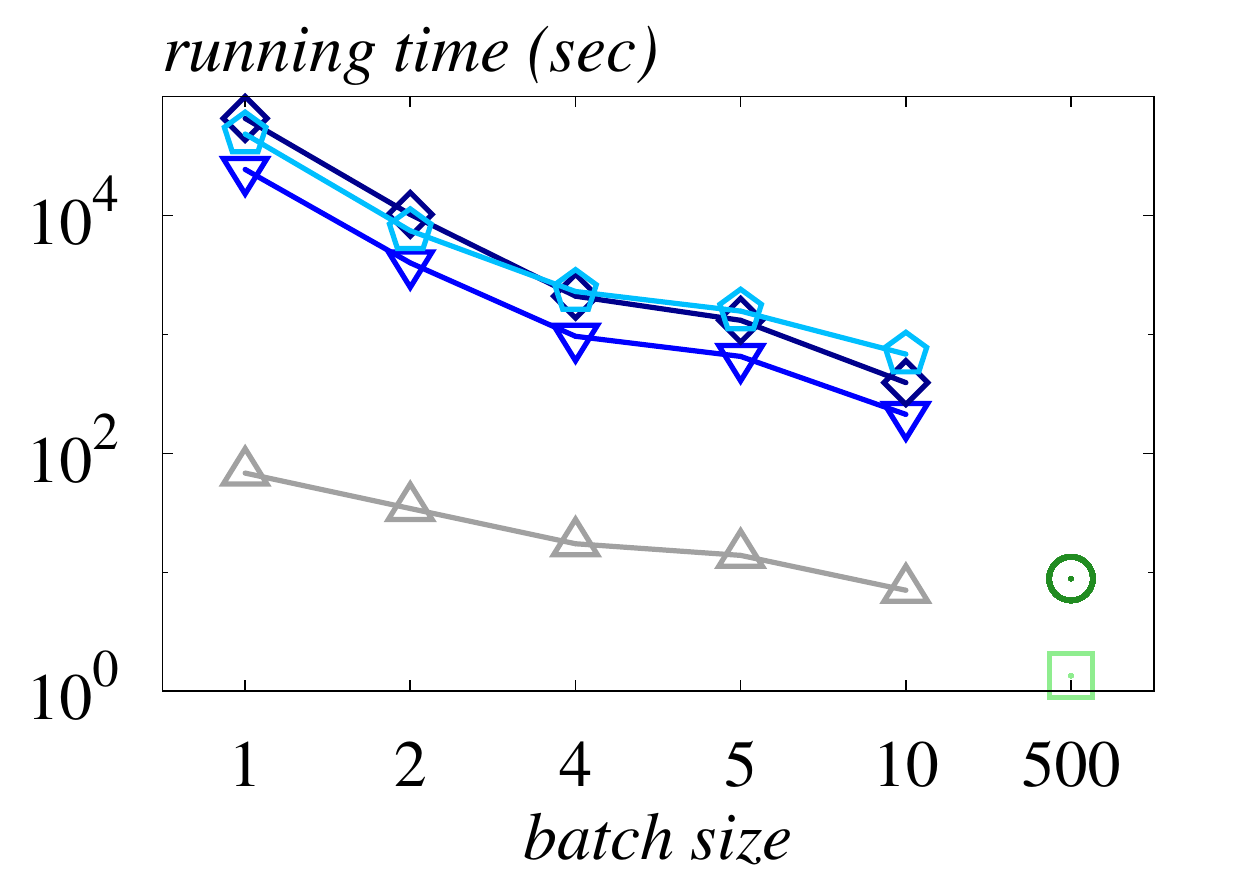}\label{fig:LiveJournal_batch_time}}
	\subfloat[Orkut]{\includegraphics[width=0.205\linewidth]{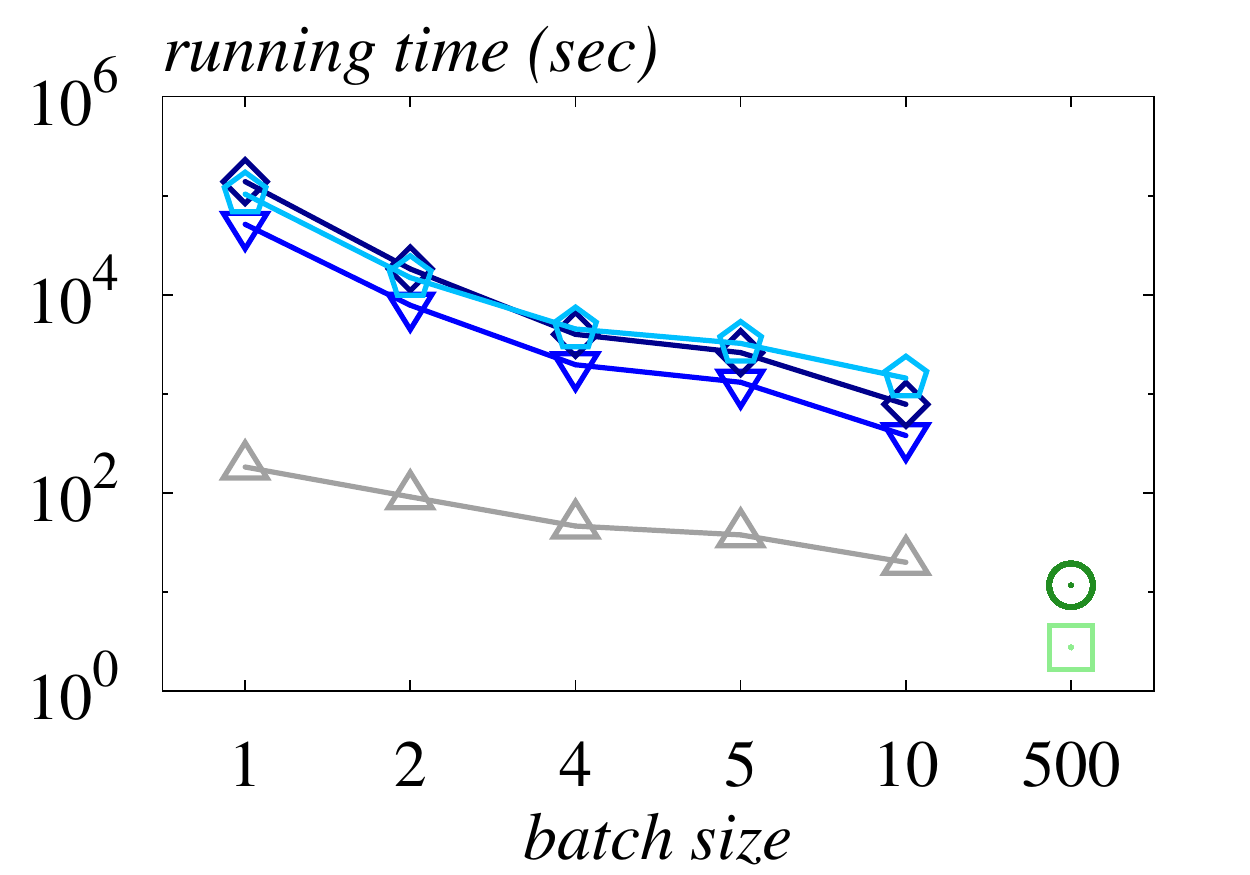}\label{fig:Orkut_batch_time}}
	\caption{Running time vs. batch size.}\label{fig:time-batch}
\end{figure*}

\begin{figure*}[!t]
	\centering
    \includegraphics[height=10pt]{legend}\vspace{-0.1in}\\
	\subfloat[NetHEPT]{\includegraphics[width=0.205\linewidth]{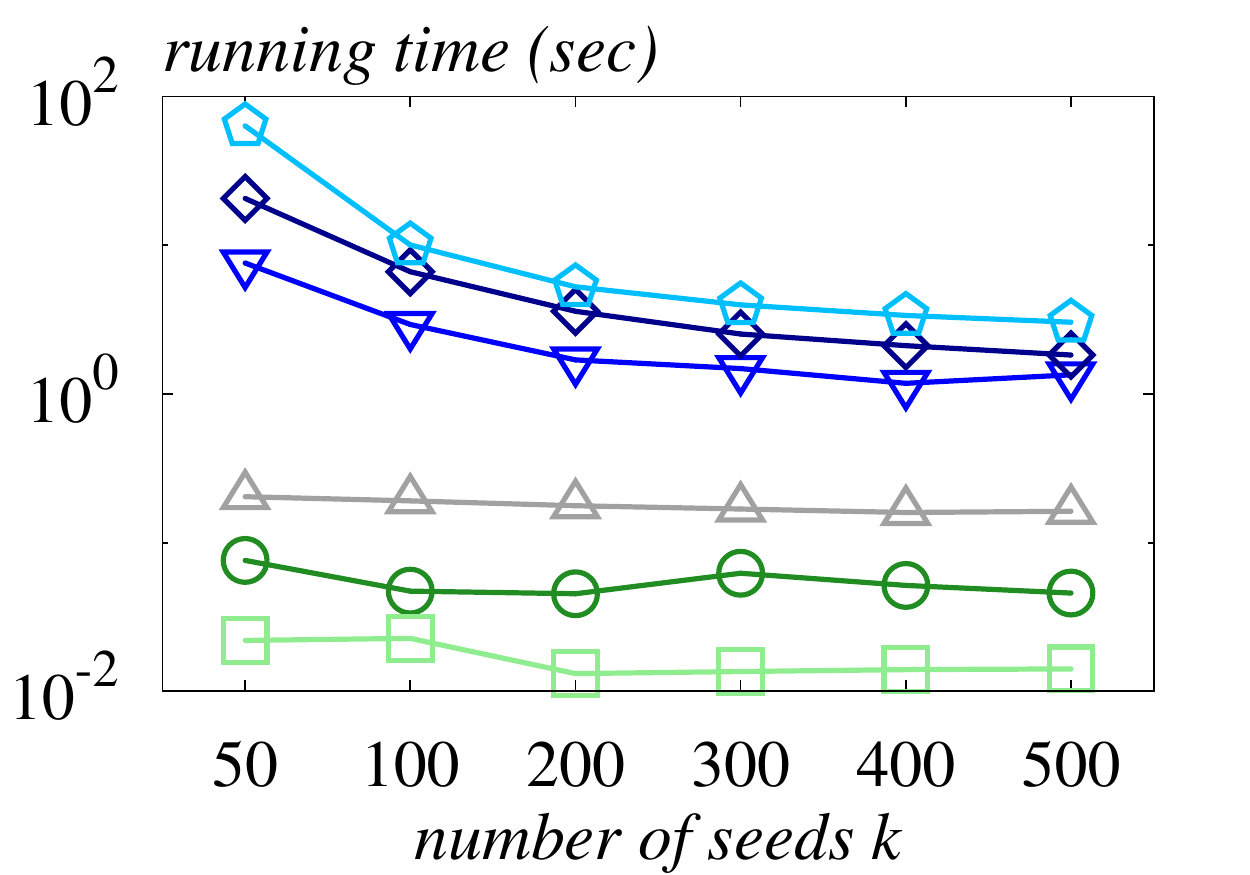}\label{fig:NetHEPT_k_time}}
	\subfloat[Epinions]{\includegraphics[width=0.205\linewidth]{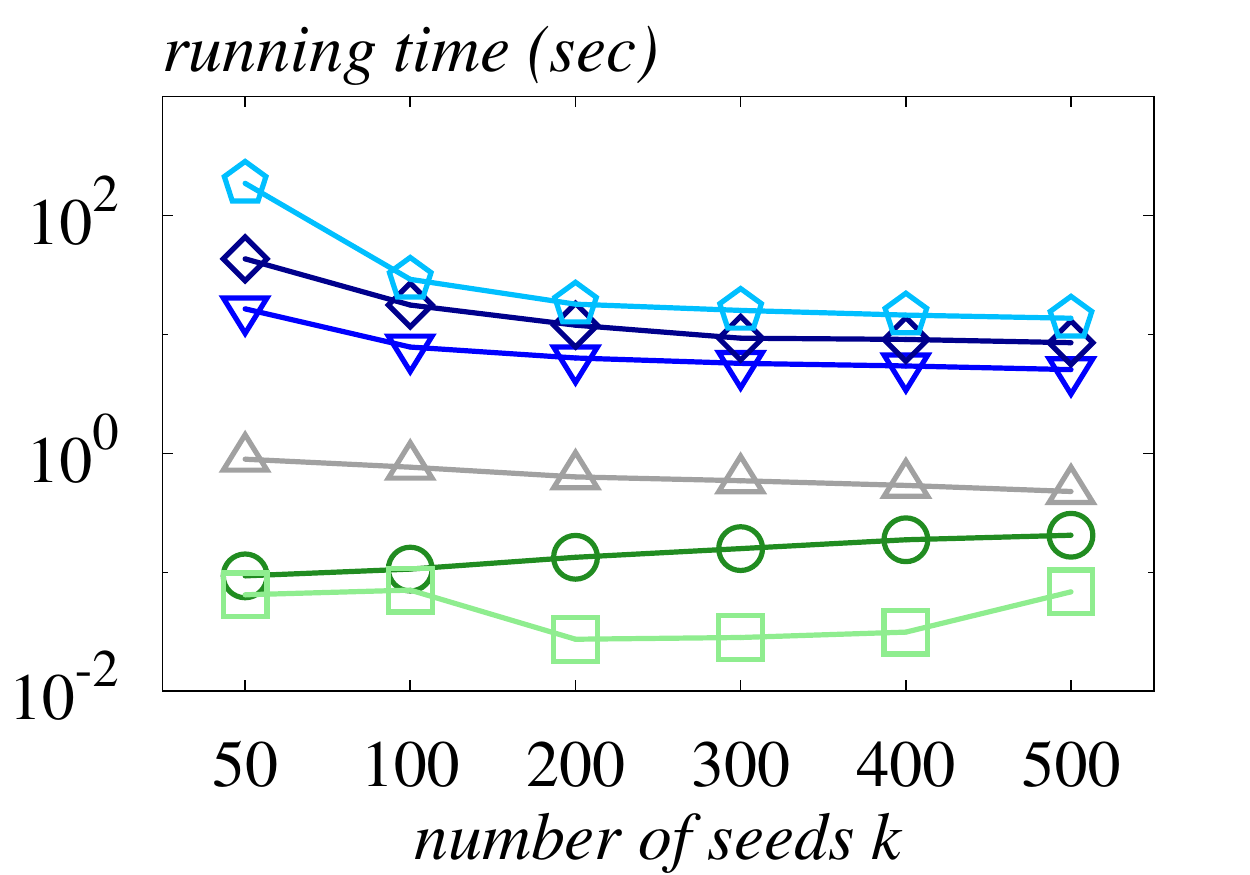}\label{fig:Epinions_k_time}}
	\subfloat[DBLP]{\includegraphics[width=0.205\linewidth]{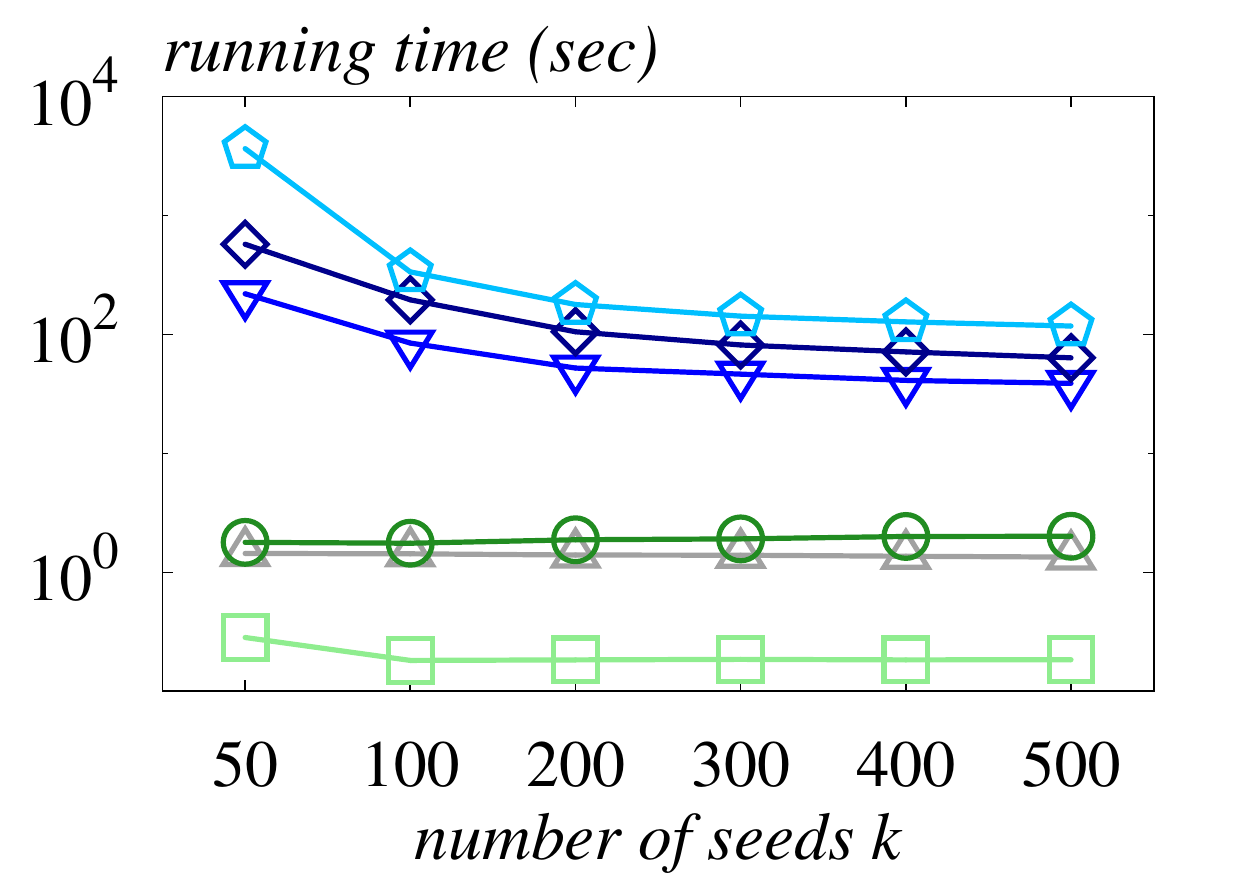}\label{fig:DBLP_k_time}}
	\subfloat[LiveJournal]{\includegraphics[width=0.205\linewidth]{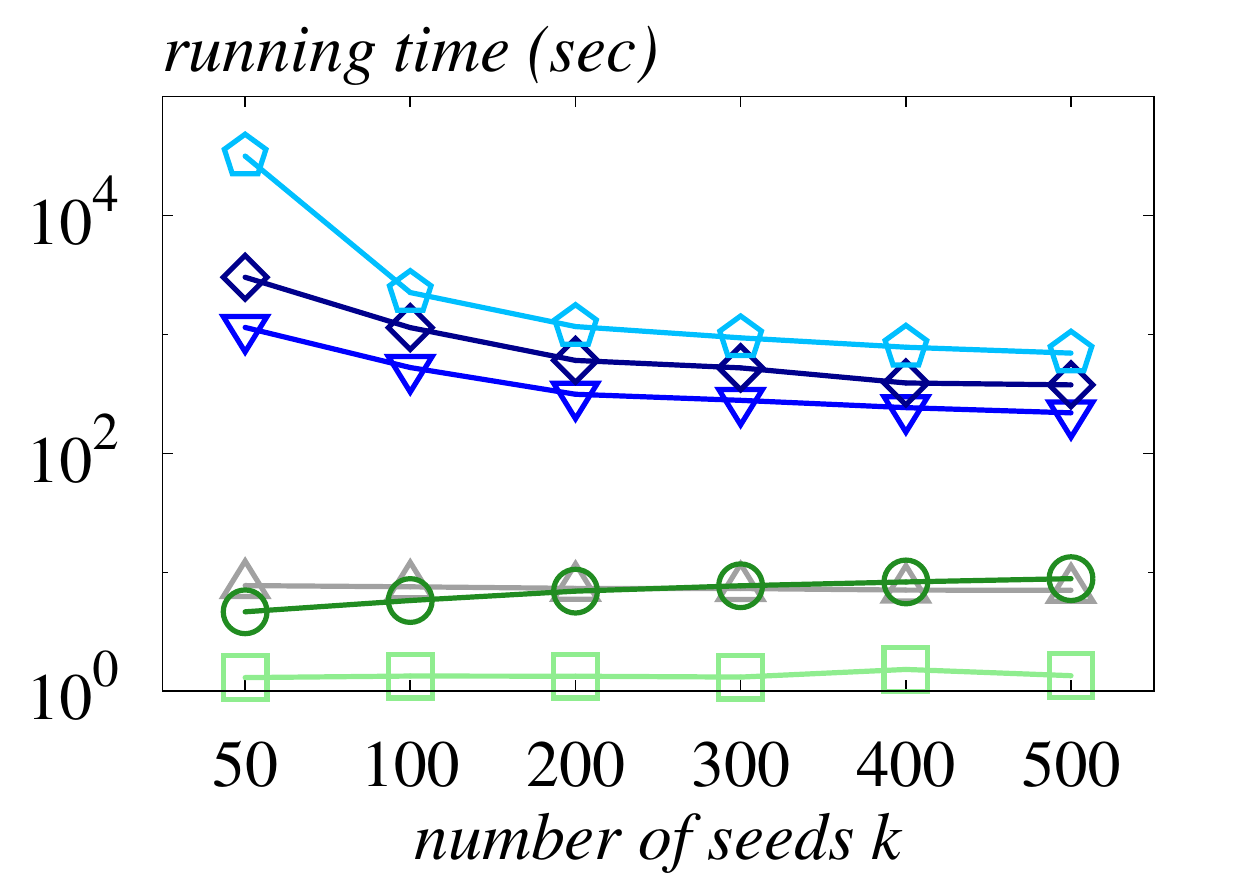}\label{fig:LiveJournal_k_time}}
	\subfloat[Orkut]{\includegraphics[width=0.205\linewidth]{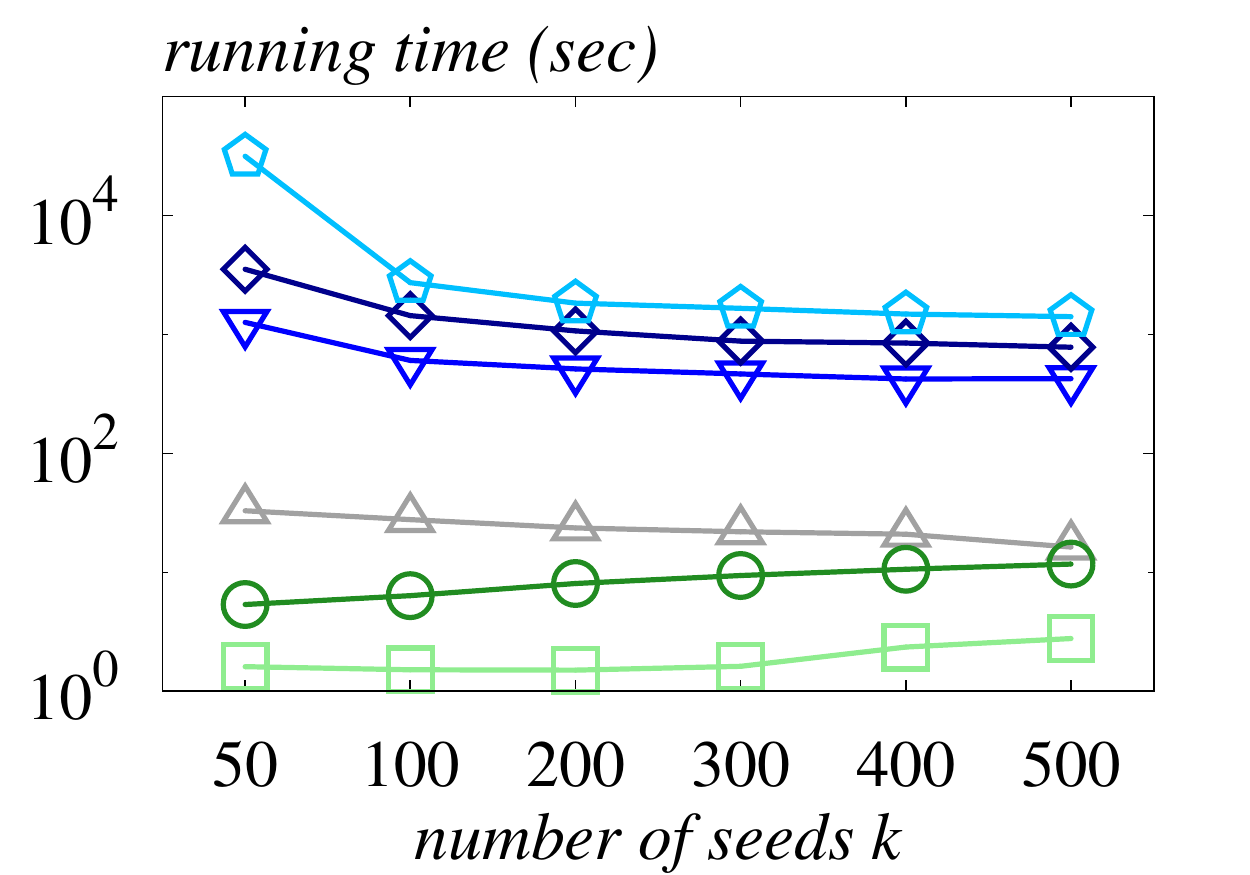}\label{fig:Orkut_k_time}}
	\caption{Running time vs. seed size.}\label{fig:time-seed}
\end{figure*}

\begin{figure*}[!t]
	\centering
    \includegraphics[height=10pt]{legend}\vspace{-0.1in}\\
	\subfloat[NetHEPT]{\includegraphics[width=0.205\linewidth]{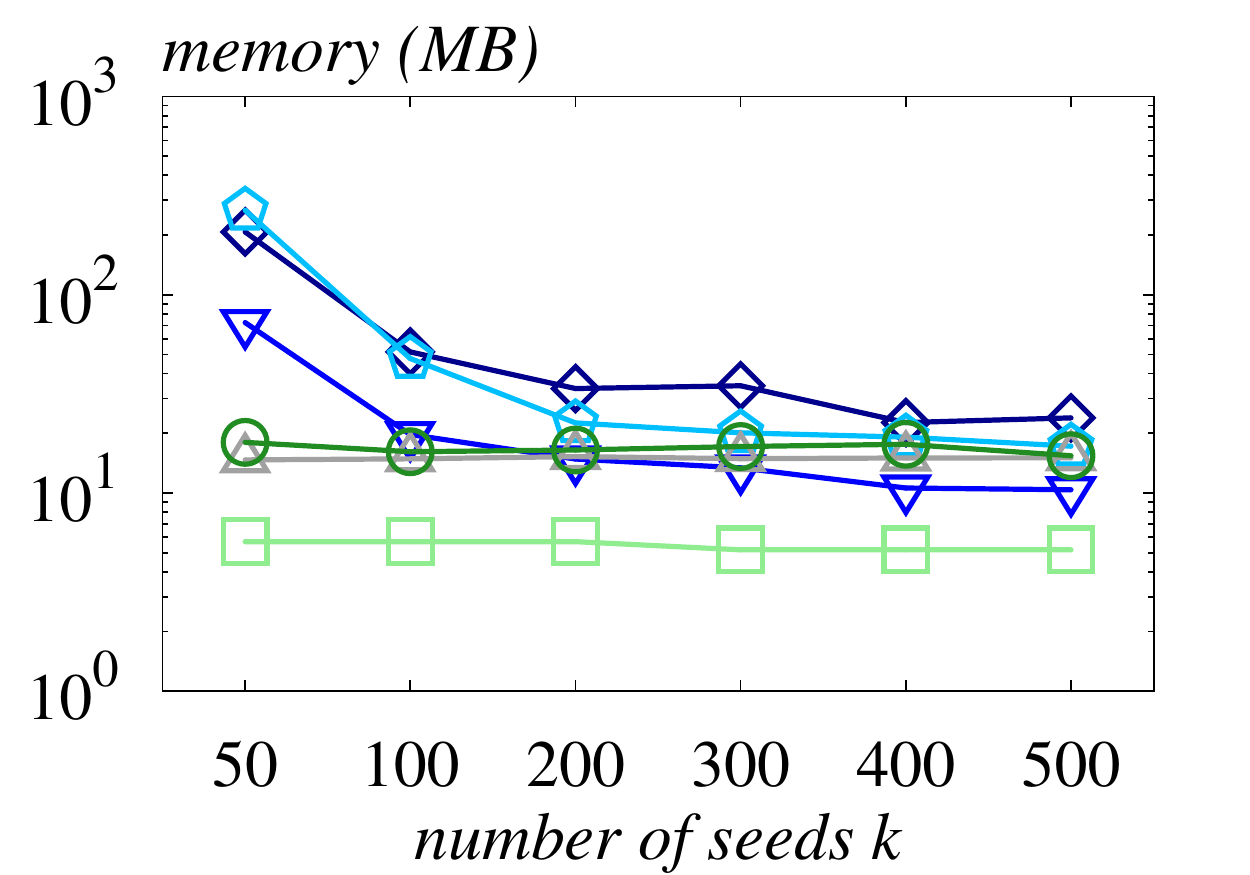}\label{fig:NetHEPT_k_mem}}
	\subfloat[Epinions]{\includegraphics[width=0.205\linewidth]{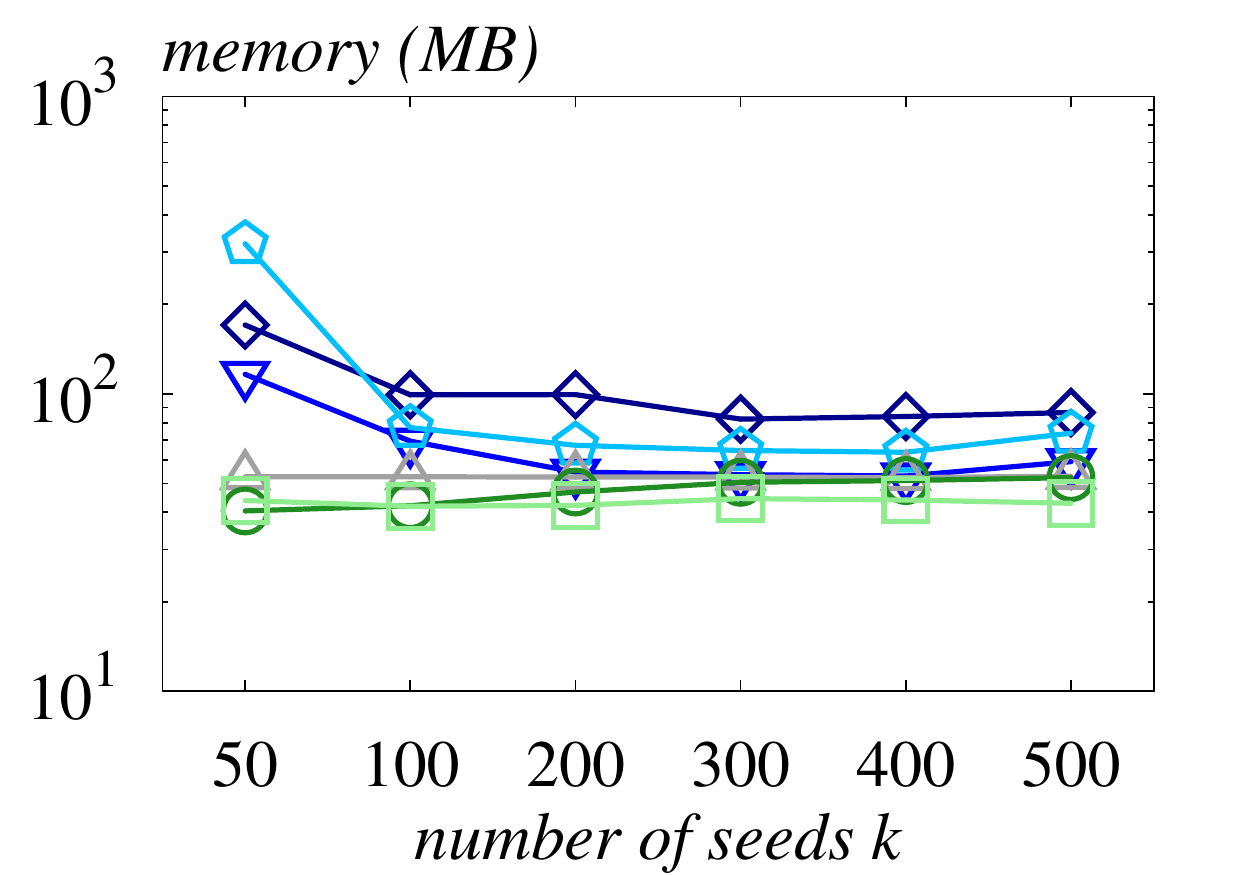}\label{fig:Epinions_k_mem}}
	\subfloat[DBLP]{\includegraphics[width=0.205\linewidth]{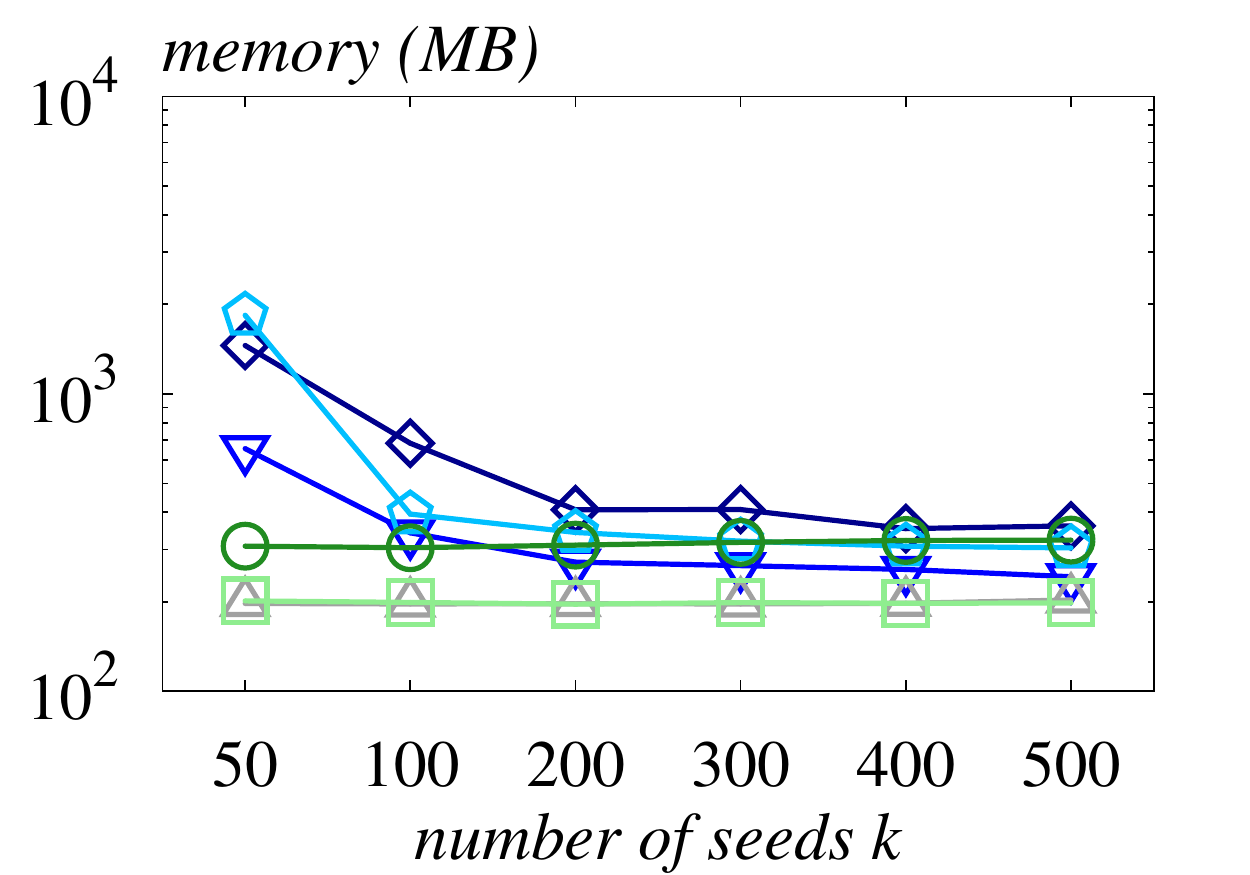}\label{fig:DBLP_k_mem}}
	\subfloat[LiveJournal]{\includegraphics[width=0.205\linewidth]{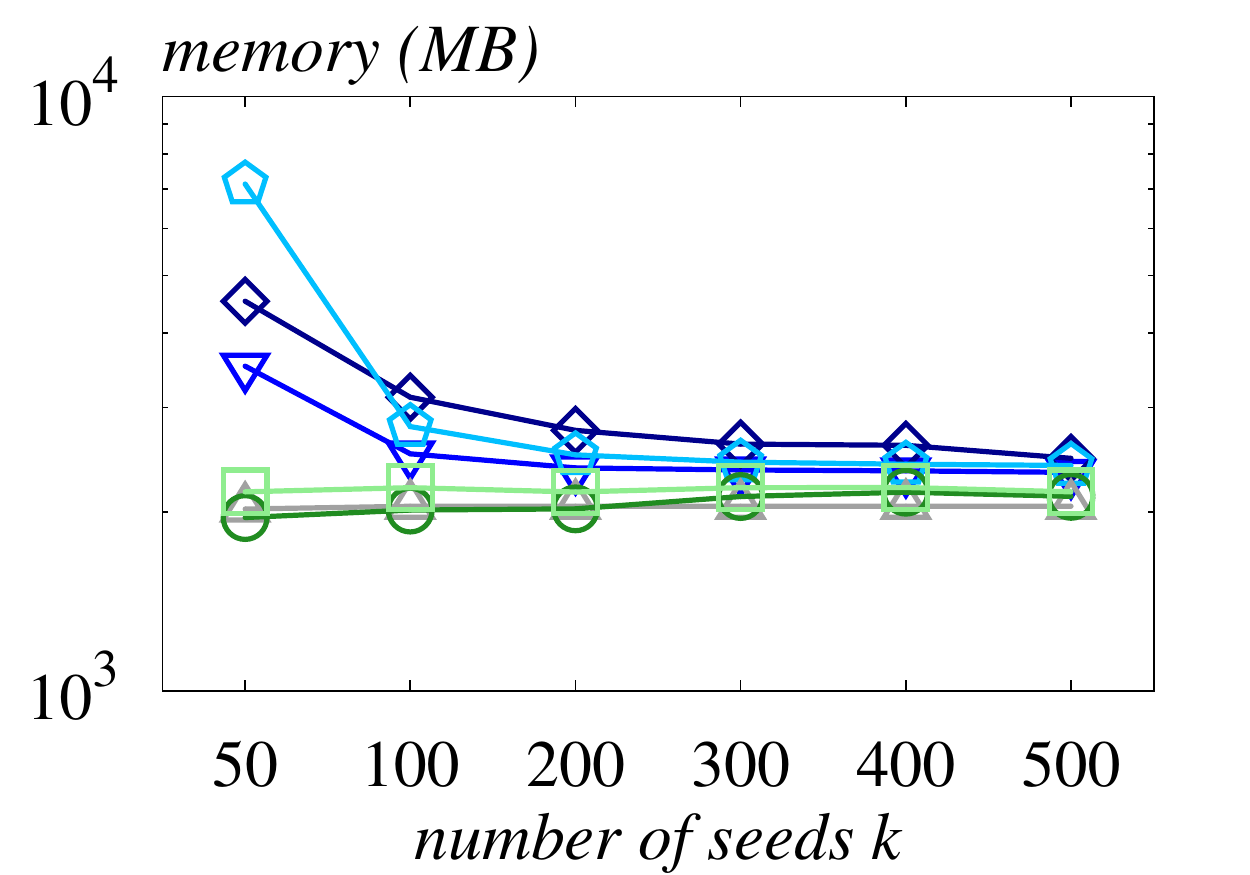}\label{fig:LiveJournal_k_mem}}
	\subfloat[Orkut]{\includegraphics[width=0.205\linewidth]{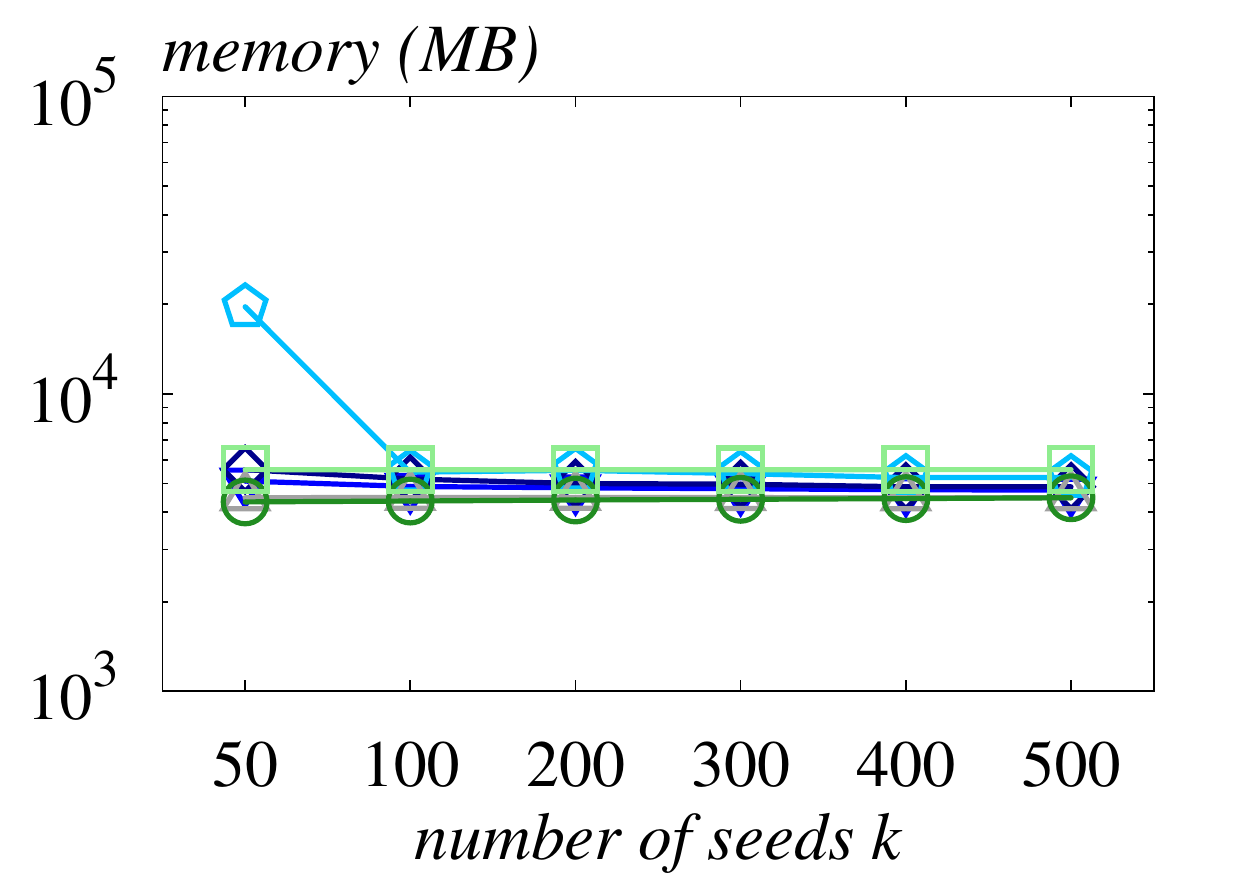}\label{fig:Orkut_k_mem}}
	\caption{Memory consumption vs. seed size.}\label{fig:memory-seed}
\end{figure*}
\subsection{Comparison of Running Time}\label{sec:running-time}

In this section, we investigate the efficiency of all tested algorithms under various seed node numbers $k$ and batch sizes $b$.

The settings of \figurename~\ref{fig:time-batch} and \figurename~\ref{fig:time-seed} follow the settings of \figurename~\ref{fig:spread-batch} and \figurename~\ref{fig:spread-seed}, respectively. In particular, \figurename~\ref{fig:time-batch} reports the running time with $k=500$ and various $b$ values under the four datasets. We observe that among the four adaptive algorithms, \fixaim surpasses the other three adaptive algorithms significantly as expected, and can even beat the non-adaptive algorithm \imm in some circumstances. This is because \fixaim generates a small number of samples (\ie~$10\mathrm{K}$) for each batch. In addition, \expepic dominates the other two adaptive algorithms on all datasets with a non-negligible advantage. Specifically, the performance gap between \expepic and \worepic tends to enlarge along the increase of the batch size $b$. This expanding gap is due to that (i) to maintain the same approximation ratio, \worepic needs to compensate for an extra factor $\sqrt{1/(2r)\cdot\ln (1/\delta)}$ on approximation error for each batch, as explained in Theorem~\ref{thm:worst-approx}, and (ii) when the seed number $k$ is fixed, this compensation factor gets larger since the batch number $r=k/b$ gets smaller. Note that \expopim runs slower than \expepic for all cases on the four datasets. When batch size $b=1$, the efficiency gap can be up to $3$ times, which demonstrates the speed improvement of our optimization in \expepic. One interesting observation is that \worepic has a slightly edge over \expopim when the batch size $b\le 5$.  

Another noticeable observation is that the running time increases along the decrease of batch size $b$. There are two main reasons. First, when the batch size $b$ becomes smaller, the marginal spread drops significantly. To maintain the same approximation, more samples are generated, which incurs considerable overhead. Second, when the number of seeds $k$ is fixed, larger $b$ value means smaller value of $r$. As mentioned, RR-sets are regenerated for each batch, and thus, a fewer number of batches leads to less sampling overhead.

\figurename~\ref{fig:time-seed} plots the running time when both seed number $k$ and batch size $b$ vary while the batch number is fixed to $r=50$. Again, \fixaim runs faster than the other three adaptive algorithms, and \imm for some cases. We also observe that the running time of \fixaim remains approximately constant, since its running time is roughly linear in the number of rounds $r$ which is a constant, \ie~$r=50$. In addition, we can see that \expepic outperforms the other two adaptive algorithms with around $1.5$--$3$ times speedup. Second, under this setting, the running time of \worepic is comparable with that of \expopim. Observe that the running time of the adaptive algorithms does not fluctuate as much as that in \figurename~\ref{fig:time-batch} when the seed size $k$ changes. This observation demonstrates that adaptive algorithms are more sensitive to the value of batch $r$ than the seed size $k$.

Note that the two non-adaptive algorithms dominate the three adaptive algorithms in efficiency, as expected. This is because non-adaptive algorithms can be seen as special adaptive algorithms with just running in one batch, which avoids enormous sampling time.

\section{Comparison of Memory Consumption}\label{sec:memory-consumption}
\figurename~\ref{fig:memory-seed} presents the memory consumptions of the tested algorithms. As shown, \fixaim and two non-adaptive algorithms, \ie \imm and \dssa, use the least memory, which remains nearly constant along with the seed size $k$. The other three adaptive algorithms, \ie \expepic, \expopim, and \worepic, consume relatively larger memory, especially for $k=50$, in which case the batch size $b=1$. Among them, \worepic needs the most memory. Observe that the memory consumptions of the adaptive algorithms approach to those of the non-adaptive algorithms when the seed size $k$ increases. To explain, the batch size $b$ increases along with the seed size $k$, which indicates that less samples would be generated for each batch of seed selection for \OPIMC \cite{Tang_OPIM_2018} based adaptive algorithms. Meanwhile, adaptive algorithms would remove all samples generated in previous batches, which could save memory significantly. Note that all memory consumptions are close on the {Orkut} dataset, since the memory taken up to store the graph itself dominates the whole memory usage.

%% file: 8-conclusion.tex
\section{Conclusion and Future Work} \label{sec:conclude}

We have studied the adaptive Influence Maximization (IM) problem, where the seed nodes can be selected in multiple batches to maximize their influence spread. We have proposed the first practical algorithm to address the adaptive IM problem that achieves both time efficiency and provable approximation guarantee. Specifically, our approach is based on a novel \AG framework instantiated by a new non-adaptive IM algorithm \EP, which has a provable {\em expected} approximation guarantee for non-adaptive IM. Meanwhile, we have clarified some existing misunderstandings in two recent work towards the adaptive IM problem and laid solid foundations for further study. Our solution to the adaptive influence maximization is based on our general solution to the adaptive stochastic maximization problem with a randomized approximation algorithm at every adaptive greedy step, and this general solution could be useful to many other settings besides adaptive influence maximization. We have also conducted extensive experiments using real social networks to evaluate the performance of our algorithms, and the experimental results strongly corroborate the superiorities and effectiveness of our approach. 

For future work, we aim to devise new algorithms that could reuse the samples generated in previous batches to further boost the efficiency. Specifically, for unbiased spread estimation in each batch of seed selection, our current algorithms generate sufficient number of RR-sets by abandoning all samples generated in previous batches. The reason behind is that reusing the ``old'' samples generated in previous batches could incur bias for spread estimation, which will affect seed selection. To tackle this issue, we aim to develop new techniques to fix or bound the bias by sample reuse, which is expected to boost the efficiency remarkably.